\documentclass[usenames,dvipsnames,svgnames,table]{lmcs} %


\usepackage{hyperref}
\usepackage{url}
\usepackage{stmaryrd}
\usepackage{amsmath}
\usepackage{amssymb}
\usepackage{cleveref}
\usepackage{graphicx}
\usepackage{amscd}
\usepackage{mdframed}
\usepackage{extarrows}
\usepackage{mathtools}
\usepackage{thm-restate}
\usepackage{tikz}
\usepackage{diagbox}
\usepackage{hhline}
\usepackage{afterpage}
\usetikzlibrary{shapes,arrows,automata,positioning}
\usepackage{tikz-cd}
\tikzstyle{vertex}=[circle,fill=black,minimum size=7pt,inner sep=0pt]
\tikzstyle{bigvertex}=[circle,draw,thick,fill=black!5,minimum size=16pt,inner
sep=0pt]
\tikzstyle{automaton}=[circle,draw,thick,minimum size=20pt,inner sep=0pt]
\tikzset{square/.style={regular polygon,regular polygon sides=4,inner sep=0}}

\Crefname{lem}{Lemma}{Lemmas}
\Crefname{defi}{Definition}{Definitions}
\Crefname{thm}{Theorem}{Theorems}
\Crefname{exa}{Example}{Examples}
\Crefname{cor}{Corollary}{Corollaries}
\theoremstyle{definition}\newtheorem{question}[thm]{Question}

\newcommand\cW{\mathcal{W}}
\newcommand\cE{\mathcal{E}}

\newcommand\backgroundcolor{white}

\newcommand\longto\longrightarrow

\newcommand{\Cantor}{2^\omega}

\newcommand{\id}{\textrm{id}}
\newcommand{\dom}{\operatorname{dom}}

\newcommand{\Baire}{{\mathbb{N}^\mathbb{N}}}
\newcommand{\hide}[1]{}
\newcommand\reAut{\mathfrak{re}_\Sigma}

\newcommand{\leqW}{\leq_{\textrm{W}}}

\newcommand{\equivW}{\equiv_{\textrm{W}}}

\newcommand{\geqW}{\geq_{\textrm{W}}}

\newcommand{\powerset}{\mathcal{P}}

\newcommand\mfA{\mathfrak{A}}
\newcommand\mfB{\mathfrak{B}}
\newcommand{\Wei}{\mathfrak{W}}
\newcommand{\ptWei}{\mathfrak{W}_\bullet}

\newcommand{\Nat}{\mathbb{N}}
\newcommand\bN\Nat
\newcommand\bbN\Nat

\newcommand{\partto}{\rightharpoonup}

\newcommand{\interp}[1]{\llbracket #1 \rrbracket}
\newcommand{\interpt}[1]{\interp{\reAut^{(#1)}}}
\newcommand{\interpaut}[1]{\interp{\mfA^{(#1)}}}

\newcommand{\bnfeq}{\mathrel{::=}}
\newcommand{\bnfalt}{\; | \;}
\newcommand\cL{\mathcal{L}}

\newcommand\RegE[1]{\mathrm{RE}_{#1}}

\newcommand\coNP{\textsc{coNP}}
\newcommand\Pspace{\textsc{Pspace}}

\newcommand\Ptime{\textsc{Ptime}}
\newcommand\Exptime{\textsc{Exptime}}

\newcommand\pol{\mathrm{pol}}

\DeclareMathOperator*{\Compo}{\bigstar}

\newcommand\done{\checkmark}

\newcommand\Pol{\mathrm{Pol}_\Sigma}

\newcommand\polarity{\mathrm{pol}}
\newcommand\simGamePlain{\mathcal{SG}}
\newcommand\simGame{\mathcal{SG}^{\reAut}}
\newcommand\simGameAut{\mathcal{SG}^{\mathfrak{A}}}
\newcommand\psimGame{\mathcal{SG}_\bullet^{\reAut}}
\newcommand\psimGameAut{\mathcal{SG}^{\mathfrak{A}}_\bullet}

\newcommand\ExeStrat{\mathsf{ES}}
\newcommand\price{\mathsf{price}}
\newcommand\execute{\mathsf{exe}}

\newcommand\code[1]{\lceil #1 \rceil}

\newcommand\tuple[1]{\langle #1 \rangle}

\newcommand\cA{\mathcal{A}}

\newcommand{\duppol}{\exists}
\newcommand{\spopol}{\forall}
\newcommand{\orapol}[1]{{#1}}

\newcommand\pSRKA{\mathrm{RSKA}_\bullet}
\newcommand\SRKA{\mathrm{RSKA}}
\newcommand\SRKAM{\mathrm{RSKA}_{s\sqcap}}
\newcommand\SRKAMT{\mathrm{RSKA}_{s\sqcap}}
\newcommand\pSRKAM{\mathrm{RSKA}_{s\sqcap,\bullet}}

\newcommand\size[1]{\mathrm{size}(#1)}

\newcommand{\position}[3]{#1 \mid #2 \vdash #3}
\newcommand{\pposition}[2]{#1 \vdash #2}

\newcommand{\inc}{\mathsf{in}}
\newcommand{\pcaapp}{\cdot}

\newcommand{\cons}{\frown}
\newcommand{\last}{\mathrm{last}}

\newcommand\rank{\mathrm{rank}}

\newcommand\unitelt{{\langle\rangle}}

\newcommand\cC{\mathcal{C}}

\newcommand\multiset[1]{\ensuremath{\{\kern-.3em\{ #1 \}\kern-.3em\}}}

\newcommand\token{\mathsf{token}}
\newcommand\Tokens{\mathsf{Tokens}}

\begin{document}

\title[Theory of the Weihrauch lattice with composition]{The equational theory of the Weihrauch lattice with (iterated) composition}

\author[C.~Pradic]{Cécilia Pradic\lmcsorcid{0000-0002-1600-8846}}

\address{Department of Computer Science, Swansea University, Wales}	%
\email{cecilia.pradic@ens-lyon.org}

\begin{abstract}
\noindent
We study the equational theory of the Weihrauch lattice with composition and iteration,
meaning the collection of equations between terms built from
variables, the lattice operations $\sqcup,\sqcap$, the composition operator $\star$ and its iteration $(-)^\diamond$, which are true however
we substitute (partial) Weihrauch degrees for the variables.
We characterize them using B\"uchi games on finite graphs and
give a complete axiomatization that derives them.
The term signature and the axiomatization are reminiscent of Kleene algebras,
except that we additionally have meets and the lattice operations do
not fully distribute over composition.
The game characterization gives a variant of the notion of simulation for
alternating automata. It also implies that
it is decidable whether an equation is universally valid.
We give some complexity bounds; in particular,
the problem is $\Pspace$-hard in general
and we conjecture that it is solvable in $\Pspace$.
We expect that the axiomatizations and games can also be applicable as-is
to characterize the existence of strong natural transformation between polynomial
functor expressions.
\end{abstract}

\maketitle

\section{Introduction}\label{S:one}
Weihrauch reducibility allows one to quantify how undecidable certain problems
are in computable analysis, much like Turing reducibility may be used to
quantify how uncomputable certain functions are.
Weihrauch reductions relate multivalued functions
on Baire space that we call \emph{Weihrauch problems}.
The intuition is that the domain of a problem
contains instances that are mapped to sets of valid solutions. Unlike
Turing reducibility, for which most degrees are either somewhat
artificially built or iterations of the Turing jump, there are a great number
of concrete Weihrauch problems of interest inspired by Reverse Mathematics.
An example of such a concrete problem is ``given an infinite binary tree $T$,
output an infinite path in $T$'', which corresponds to \emph{Weak K\H{o}nig's Lemma}.
For two Weihrauch problems $P$ and $Q$, a Weihrauch reduction from $P$ to $Q$
models the notion of making exactly one oracle call to $Q$ to solve $P$ in a type-2
setting.
Here ``type-2'' means that we are working computability on infinite
bit strings: a type-2 Turing machine is a deterministic Turing machine which has
a read-only input tape containing initialized with an infinite bit string encoding the
input, a number of working tapes, and a write-and-move-right-only output tape
whose limit state should be the desired output.
More abstractly, a partial function $f : \Cantor \partto \Cantor$ is type-2
computable if and only if there is a computable map $\hat{f} : \{0,1\}^* \partto \{0,1\}^*$
(in the more traditional sense) which is monotone for the prefix ordering on strings and
such that, for every $p \in \dom(f)$ and $n \in \Nat$, $f(p)(n) = \hat{f}(\sigma)(n)$ for a long
enough finite prefix $\sigma$ of $p$ (see e.g.~\cite{weihrauchbook} for a textbook introduction).
A Weihrauch reduction then consists of a forward function $\varphi : \dom(P) \to \dom(Q)$
corresponding to calling the oracle $Q$ after receiving a $P$-question, and a backward function $\psi$
corresponding to the computation of the answer to the original $P$-question, with $\varphi$ and
$\psi$ type-2 computable.
Equivalence classes of problems under reduction are called Weihrauch degrees.

A common cause of complexity for Weihrauch problems is that they
are discontinuous, in the sense that there is no continuous map $f : \dom(P) \to \bigcup_{u \in \dom(P)} P(u)$
with $f(u) \in P(u)$ for all $u \in \dom(P)$. Another one is that we may
have some $u \in \dom(P)$ such that $P(u)$ does not contain any solution $x$
Turing-reducible to $u$. It is possible that $\dom(P)$ does not
contain any computable point, although this does not tend to happen often with
``natural'' problems. We call \emph{pointed} those Weihrauch problems
$P$ which contain a computable point (accordingly,
we call a degree which contain a pointed problem pointed as well).

There are a number of natural operators on Weihrauch problems which
are used in the study of concrete problems. They also equip the Weihrauch degrees
with a rich algebraic structure that is worth studying. In this paper,
we are going to consider the lattice operations $\sqcap$ and $\sqcup$, the operator $\star$ that \emph{composes}
problems, its iterated variant $(-)^\diamond$
and relevant constants.
Given problems $P$ and $Q$, $P \sqcup Q$ and $P \sqcap Q$ both correspond to solving
either $P$ or $Q$, but with a different policy on the inputs: for $\sqcup$, the
input specifies which problem should be solved, while for $\sqcap$, inputs for
both problems should be provided and there is no guarantee on which problem is solved
(although the output of $P \sqcap Q$ does indicate with a tag whether it is $P$ or $Q$ which
is solved). The neutral element for $\sqcup$, which we call $0$, is the problem with an empty set of input.
On the other hand, we write $\top$ for the unit for $\sqcap$, which is the problem
with one trivial computable input, but no output whatsoever; this is usually not
considered a legal Weihrauch problem. In this paper, we will want to work with $\top$
sometimes and thus move to a slight extension to Weihrauch problems that we dub
\emph{partial Weihrauch problems}.

$P \star Q$ is constructed so that a Weihrauch reduction to it is morally
able to make an oracle call to $Q$ and then an oracle call to $P$ (and \emph{must} perform
these two oracle calls, in order). In more details, an input for $P\star Q$
consists of a question $q$ for $Q$ and a code for a map $f$ taking an answer to $q$
to a question for $P$. The answer to the question $(q, f)$ to $P \star Q$ is then
a pair $(a, b)$ where $a$ answers $q$ and $b$ answers $f(a)$.
There is a neutral element $1$ for $\star$, which is the problem with a trivial
computable input-output pair.
The iteration $(-)^\diamond$ then has a couple of descriptions. It was first
introduced by way of type-2 oracle Turing machines in~\cite{topol-comput-neumann-pauly},
the idea being that an input to $P^\diamond$ should be a 
machine\footnote{Including an infinite advice string.}
that always produce an output after finitely many oracle calls to $P$ along
any run\footnote{
Note that the number of calls to $P$ across all possible runs has no reason to be bounded if $P$ is non-deterministic.},
and the output should be the output of that machine.
Alternatively, one can note that $(-)^\diamond$ is a least fixpoint of $X \mapsto 1 + X \star P$,
and use that as the basis for a more convenient alternative definition we will work with here~\cite[Definition 3]{westrick2020}.

Any expression built of variables taken from an alphabet $\Sigma$,
the constants $0,1,\top$ and the operators $\sqcap$, $\sqcup$, $\star$,
$(-)^\diamond$ can be regarded as a regular expression with intersections
via a straightforward analogy: the lattice operations correspond to union
and intersection ($0$ and $\top$ to the empty and the full languages), composition to concatenation ($1$ to $\{\varepsilon\}$) and $(-)^\diamond$ to the Kleene
star. Because of this correspondence, we call $\RegE{\Sigma}$ the set of
such expressions in the sequel. The main question to be addressed in this paper
is the following.

\begin{question}[{towards~\cite[Open Question 2]{pauly2020update}}]
Can we characterize or decide the universal validity of $e \le f$ in the (partial) Weihrauch
degrees for $e, f \in \RegE{\Sigma}$?
\end{question}

The main answer we offer is that the universal validity of an inequality $e \le f$
can be characterized using a B\"uchi game $\simGame(\position{\emptyset}{\{e\}}{f})$
on a finite graph. The intuition is that $e$ and $f$ can be
regarded as alternating automata over a common alphabet and that the
game $\simGame(\position{\emptyset}{\{e\}}{f})$ captures a notion of step-by-step simulation of
$e$ by $f$ (on the other hand, the languages recognized by the expressions $e$ and $f$
are not so important). While we postpone the technical definition of $\simGame$ to later,
we can look at some examples of valid reductions, and intuit how simulations on
the corresponding automata may be induced by such reductions.

\begin{exa}
  \label{ex:intro-init}
Consider the two expressions $(b \star a) \sqcup (c \star a)$ and $(b \sqcup c) \star a$.
They correspond to the non-deterministic automata pictured below\footnote{Where
expressions are processed right-to-left. This matches the convention imposed by composition,
which is opposite to the usual convention for reading words in automata theory.
On the other hand, a good intuition is that words along paths correspond to
possible sequences of oracle calls.}:

  \begin{center}
\begin{tabular}{c !\qquad c}
  \begin{tikzpicture}[->,>={Stealth[round]},shorten >=1pt,
                    node distance=1.5cm,semithick,
                    inner sep=2pt,bend angle=45]
  \tikzset{every state/.style={minimum size=20pt}}
  \tikzset{initial text={}}
  \node[initial,state] (top) {};
  \node[state]         (l) [above right=of top] {};
  \node[state]         (ld) [right=of l] {};
  \node[state,accepting right]         (ldd) [right=of ld] {};
  \node[state]         (r) [below right=of top] {};
  \node[state]         (rd) [right=of r] {};
  \node[state,accepting right]         (rdd) [right=of rd] {};

  \path [every node/.style={fill=\backgroundcolor,circle}]
        (top) edge (l)
        (l) edge  node {$a$} (ld)
        (ld) edge node {$b$} (ldd)
        (top) edge (r)
         (r)  edge   node {$a$} (rd)
         (rd)   edge node {$c$} (rdd);
\end{tikzpicture}
  &
\begin{tikzpicture}[->,>={Stealth[round]},shorten >=1pt,
                    node distance=1.5cm,semithick,
                    inner sep=2pt,bend angle=45]
  \tikzset{every state/.style={minimum size=20pt}}
  \tikzset{initial text={}}
  \node[initial,state] (top) {};
  \node[state]         (a) [right=of top] {};
  \node[state, accepting right]         (ab) [above right=of a] {};
  \node[state,accepting right]         (ac) [below right=of a] {};

  \path [every node/.style={fill=\backgroundcolor,circle}]
        (top) edge node {$a$} (a)
        (a) edge node {$b$} (ab)
         (a)   edge node {$c$} (ac);
\end{tikzpicture}
\\
  $(b \star a) \sqcup (c \star a)$
  &
  $(b \sqcup c) \star a$
\end{tabular}
  \end{center}

The language, which correspond to traces of oracle calls to the problems,
recognized by these two machines is the same ($\{ab,ac\}$). But while there
is a reduction $(b \sqcup c) \star a \ge (b \star a) \sqcup (c \star a)$,
the converse is not possible. Accordingly, the reduction corresponds to
a simulation of NFAs in the usual sense\footnote{See e.g. ~\cite[Definition 7.47]{baierkatoen} for details.}, and the lack of reduction the
other way around can be witnessed by a witness that there can be no
simulations.
\end{exa}

\begin{exa}
  \label{ex:introJunk}
Before moving to examples with more operators, let us observe
that even just with $\star$, the notion of simulation needs to be
relaxed somewhat as we have $b \star a \le a \star b \star a$: the corresponding
generic reduction simply asks the same $a$-question twice in the oracle call, and
then discards the second answer.
In the corresponding game, this will correspond to the ability to skip
a step in the simulator without having to take a matching transition
in the simulated machine.

On the other hand, it can be the case that $a \star b \not\le a \star b \star a$.
This is because it could be
the case that $a$ is \emph{not} pointed, and thus there could
be a question to $a \star b$ which cannot compute
any questions to $a$.
\end{exa}

\begin{exa}
  In more complex situations, typically involving \emph{alternating}
automata, it might be the case that a machine is only able
to simulate another one by making several attempts in parallel, which is
a legal behaviour in the context of Weihrauch reducibility.
Consider the following examples of expressions and corresponding machines,
where we use square nodes to denote states where the non-determinism should
be resolved adversarially.

  \begin{center}
\begin{tabular}{c !\qquad c}
\begin{tikzpicture}[->,>={Stealth[round]},shorten >=1pt,
                    node distance=1.5cm,semithick,
                    inner sep=2pt,bend angle=45]
  \tikzset{every state/.style={minimum size=20pt}}
  \tikzset{initial text={}}
  \tikzset{opp state/.style={draw,square,minimum size=20pt}}
  \node[initial,opp state] (top) {};
  \node[state]         (ba) [above right=of top] {};
  \node[state]         (ca) [below right=of top] {};
  \node[state,accepting right]         (bab) [right=of ba] {};
  \node[state,accepting right]         (cac) [right=of ca] {};

  \path [every node/.style={fill=\backgroundcolor,circle}]
        (top) edge node {$a$} (ba)
        (top) edge node {$a$} (ca)
        (ba) edge node {$b$} (bab)
        (ca) edge node {$c$} (cac);
\end{tikzpicture}
  &
\begin{tikzpicture}[->,>={Stealth[round]},shorten >=1pt,
                    node distance=1.5cm,semithick,
                    inner sep=2pt,bend angle=45]
  \tikzset{every state/.style={minimum size=20pt}}
  \tikzset{initial text={}}
  \tikzset{opp state/.style={draw,square,minimum size=20pt}}
  \node[initial, state] (top) {};
  \node[opp state]         (a) [right=of top] {};
  \node[state, accepting right]         (ab) [above right=of a] {};
  \node[opp state]         (aa) [below right=of a] {};
  \node[state,accepting right]         (aac) [below right=of aa] {};
  \node[state,accepting right]         (aab) [above right=of aa] {};

  \path [every node/.style={fill=\backgroundcolor,circle}]
        (top) edge node {$a$} (a)
        (a) edge node {$b$} (ab)
        (a) edge node {$a$} (aa)
        (aa) edge node {$c$} (aac)
        (aa) edge node {$b$} (aab);
\end{tikzpicture}
\\
  $(b \star a) \sqcap (c \star a)$
  &
  $(((c \sqcap b) \star a) \sqcap b) \star a$
\end{tabular}
  \end{center}
Here since all branching states are squares, there is no genuine alternation.
Seen as string-reading automata, both of these machines recognize the empty language, but
we are again in the situation where we have a Weihrauch reduction in only one direction.
A simulation witnessing the reduction
  $(b \star a) \sqcap (c \star a) \le (((c \sqcap b) \star a) \sqcap b) \star a$ must
first attempt to use the right-hand side to simulate $b \star a$. But then, if
the non-determinism in the simulator exposes $(c \sqcap b) \star a$, an attempt
at simulating $c \star a$ should also be started in parallel. The thread that
can be simulated to completion will then depend on the last non-deterministic step
of $c \sqcap b$.
\end{exa}

\begin{exa}
\label{ex:intro-3}
Let us end with an example of a reduction between expressions that give
  rise to genuinely alternating machines with runs of arbitrary (finite) length.
  \begin{center}
\begin{tabular}{c !\qquad c}
\begin{tikzpicture}[->,>={Stealth[round]},shorten >=1pt,
                    node distance=1.5cm,semithick,
                    inner sep=2pt,bend angle=45]
  \tikzset{every state/.style={minimum size=20pt}}
  \tikzset{initial text={}}
  \tikzset{opp state/.style={draw,square,minimum size=20pt}}
  \node[initial,opp state] (top) {};
  \node[state]         (u) [above right=of top] {};
  \node[state]         (d) [below right=of top] {};
  \node[state,accepting right]         (ub) [right=of u] {};
  \node[state,accepting right]         (dc) [right=of d] {};

  \path [every node/.style={fill=\backgroundcolor,circle}]
        (top) edge (u)
        (top) edge (d)
        (u) edge node {$b$} (ub)
        (d) edge node {$c$} (dc)
        (u) edge [out=55,in=125,looseness=7] node {$a$} (u)
        (d) edge [out=-55,in=-125,looseness=7] node {$a$} (d);
\end{tikzpicture}
  &
\begin{tikzpicture}[->,>={Stealth[round]},shorten >=1pt,
                    node distance=1.5cm,semithick,
                    inner sep=2pt,bend angle=45]
  \tikzset{every state/.style={minimum size=20pt}}
  \tikzset{initial text={}}
  \tikzset{opp state/.style={draw,square,minimum size=20pt}}
  \node[initial,state] (top) {};
  \node[opp state]         (a) [right=of top] {};
  \node[state, accepting right]         (ab) [above right=of a] {};
  \node[state,accepting right]         (ac) [below right=of a] {};

  \path [every node/.style={fill=\backgroundcolor,circle}]
        (top) edge [out=55,in=125,looseness=7] node {$a$} (top)
        (top) edge (a)
        (a) edge node {$b$} (ab)
        (a) edge node {$c$} (ac);
\end{tikzpicture}
\\
  $(b \star a^\diamond) \sqcap (c \star a^\diamond)$
  &
  $(b \sqcap c) \star a^\diamond$
\end{tabular}
  \end{center}
There are reductions both ways. The harder direction is establishing that
$(b \star a^\diamond) \sqcap (c \star a^\diamond) \le
(b \sqcap c) \star a^\diamond$. The forward pass of the reduction is given
as input a tuple $(i_0, f_0, i_1, f_1)$ where $i_0$ and $i_1$ are input to $a^\diamond$
and $f_0$ is (the code of) a function that turns an $a^\diamond$-answer to $i_0$
into a question to $b$ (and similarly for $f_1$, $i_1$ and $c$).
The oracle call that makes the reduction works in two steps:
first $i_0$ and $i_1$ are combined into a single question to $a^\diamond$,
which is morally\footnote{Up to re-encoding, which can be thought of as induced
by the equivalence $a^\diamond \star a^\diamond \equiv a^\diamond$.} answered by $(x_0, x_1)$ where $x_0$ answers $i_0$ and $x_1$ answers $i_1$.
Then the second part encodes the map $(x_0, x_1) \mapsto (f_0(x_0), f_1(x_1))$.
The answer to the oracle call is then some $(x_0, x_1, k, y)$ where $k \in 2$
is the bit witnessing whether $b$ or $c$ was answered, and where accordingly
$y$ is answering $f_k(i_k)$.
Then the original question is answered by $(k, x_k, y)$.

Turning to the automata, we again have a phenomenon where the
correct notion of simulation must match several state to the left-hand side;
here for instance, this means that the square state on the right-hand-side must
cover the exits of all $a$-looping states of the left-hand side.
\end{exa}

Thanks to the characterization via simulation games, one can
translate a pair of expressions $(e,f) \in \RegE{\Sigma}$ to a game
$\simGame(\position{\emptyset}{\{e\}}{f})$ on a finite arena where Duplicator wins if and only we always
have $e \le f$.
This can be turned into a decision procedure, as deciding who is the winner in
B\"uchi games on finite graphs is computable in polynomial time.
But in general, our simulation games have arenas of
size exponential in $e$. This is because the notion of simulation that
captures Weihrauch reducibility allows $f$ to make several attempts in
parallel to simulate $e$. Taking this into account, we are able
to offer some complexity bounds for deciding whether $e \le f$ holds in
natural substructures of $(\mathfrak{W}, \sqcap, \sqcup, \star, 0, 1, (-)^\diamond)$,
where $\mathfrak{W}$ are Weihrauch degrees. The substructures we consider are
obtained by either dropping operators from the signature, or restricting the
carrier to pointed degrees $\mathfrak{W}_\bullet$. See~\Cref{fig:complexity-summary}
for a summary of these results.

This game correspondence also leads us to a complete axiomatization of inequalities
in partial Weihrauch degrees
equipped with the full set of constants
$0,1, \top$ and connectives $(-)^\diamond, \sqcup, \sqcap$.
The axiomatization, given in
\Cref{fig:axioms}, is reminiscent of right-handed Kleene algebras~\cite{LKA-KozenSilva} save for
the following aspects:
\begin{itemize}
\item $\star$ does not left-distribute over $\sqcup$. This is
because we do not have $(P \sqcup Q) \star R \leqW (P \star R) \sqcup (Q \star R)$
in general: the second component $f$ of a question $\tuple{w, f} \in \dom((P \sqcup Q) \star R)$
might decide whether a question should be asked to $P$ or $Q$ depending on the $R$-answer to $w$.
\item we add axioms dealing with distributive meets $\sqcap$, which
also satisfy the following half-distributivity axiom.
\[ (a \star b) \sqcap c \le (a  \sqcap c) \star b \]
This does not admit a very natural automata-theoretic interpretation, but is
rather linked to the \emph{strength} of polynomial functors induced by problems.
In line with this, we also generalize the induction principle for $(-)^\diamond$ to the following.
\[a \sqcap (b \star c) \le b ~~ \Rightarrow ~~ a \sqcap (b \star c^\diamond) \le b\]
\end{itemize}

For these reasons, we dub $\SRKAM$ the theory corresponding to the axioms of~\Cref{fig:axioms},
an acronym standing for \emph{right-skewed Kleene algebras with strong meets}.

\afterpage{

\begin{figure}[h]
\begin{center}
\begin{tabular}{|l||c|c|}
\hline
\diagbox{operators}{carrier} & Weihrauch degrees $\Wei$ & pointed Weihrauch degrees $\ptWei$ \\
\hhline{|=||=|=|}
$1, \star, (-)^\diamond, \sqcup$
& $\coNP$-hard & $\Ptime$ \\
\hline
$1, \star, \sqcup, \sqcap$
& \multicolumn{2}{c|}{$\Pspace$-complete} \\
\hline
$1, \star, (-)^\diamond, \sqcup, \sqcap$
& \multicolumn{2}{c|}{$\Exptime$} \\
\hline
\end{tabular}
\end{center}
\caption{Some complexity bounds for deciding ``Is $e \le f$ valid?'' in substructures of Weihrauch
degrees with the operators we discuss.}
\label{fig:complexity-summary}
\end{figure}

\begin{figure}[h]
\begin{center}
\[\begin{array}{c!\qquad r}
\multicolumn{2}{c}{\text{Distributive lattice with $0$ and $\top$}}\\
a \le a & \text{\footnotesize reflexivity}\\
a \le b \; \wedge \; b \le c ~~ \Rightarrow ~~ a \le c & \text{\footnotesize transitivity}\\
0 \le a \qquad a \le \top & \text{\footnotesize bottom and top elements}
\\
a \le a \sqcup b \qquad b \le a \sqcup b  & \text{\footnotesize $\sqcup$ is an upper bound} \\
b \le a \; \wedge \; c \le a ~~ \Rightarrow ~~ b \sqcup c \le a & \text{\footnotesize$\sqcup$ is the least upper bound} \\
a \sqcap b \le a \qquad a \sqcap b \le b  & \text{\footnotesize$\sqcap$ is a lower bound} \\
a \le b \; \wedge \; a \le c ~~ \Rightarrow ~~ a \le b \sqcap c & \text{\footnotesize$\sqcup$ is the least lower bound} \\
a \sqcap (b \sqcup c) \le (a \sqcap b) \sqcup (a \sqcap c) &\text{\footnotesize distributivity of $\sqcap$ over $\sqcup$}
\\\\
\multicolumn{2}{c}{\text{Ordered monoid axioms + distributivity}}\\
  a \star 1 = a = 1 \star a \qquad a \star (b \star c) = (a \star b) \star c
  & \text{\footnotesize monoid axioms}\\
  a \le a' \wedge b \le b' \Rightarrow a \star b \le a' \star b' & \text{\footnotesize monotonicity of $\star$}\\
 a \star 0 \le 0 \qquad \top \le a \star \top & \text{\footnotesize $0$ and $\top$ $\star$-absorptions}\\
  a \star (b \sqcup c)
\le  (a \star b) \sqcup (a \star c)
&\text{\footnotesize left-distributivity of $\star$ over $\sqcup$}\\
(a \star b) \sqcap (a \star c) \le a \star (b \sqcap c) &\text{\footnotesize left-distributivity of $\star$ over $\sqcap$}\\
(a \star b) \sqcap c \le (a  \sqcap c) \star b & \text{\footnotesize left-half-distributivity
of $\sqcap$ over $\star$}\\
\\\\
\multicolumn{2}{c}{\text{Iteration of composition as a least fixpoint}}\\
1 \le a^\diamond \qquad a^\diamond \star a \le a^\diamond & \text{\footnotesize fixpoint unfolding} \\
a \sqcap (b \star c) \le b ~~ \Rightarrow ~~ a \sqcap (b \star c^\diamond) \le b & \text{\footnotesize parameterized $\diamond$-induction}\\
\end{array}
\]
\end{center}
\caption{Axioms of right-skewed Kleene algebras with distributive meets ($\SRKAM$), where $=$ stands for
    having inequalities both ways. If we want to regard it as an
  equational theory, we can take $a \le b$ to be a notation for $a = a \sqcap b$ and throw in the
axioms that make $\sqcap$ associative, commutative and idempotent.}
\label{fig:axioms}
\end{figure}
}

All in all, our main results are essentially captured in (variations of) the following theorem and
\Cref{fig:complexity-summary}.

\begin{restatable}{thm}{mainloop}
\label{thm:mainloop}
The following are equivalent for any $e, f \in \RegE{\Sigma}$:
\begin{enumerate}
\item
\label{enumitem:loop-proof}
$e \le f$ is derivable from the axioms of $\SRKAM$. \hfill {\footnotesize (see~\Cref{fig:axioms})}
\item
\label{enumitem:loop-valid}
$e \le f$ is valid in the partial Weihrauch degrees
\item
\label{enumitem:loop-game}
Duplicator has a winning strategy in the simulation game $\simGame(\position{\emptyset}{\{e\}}{f})$\\ \phantom{a} \hfill {\footnotesize(see~\Cref{def:simgame})}
\end{enumerate}
\end{restatable}

We define Weihrauch reducibility and the operators under consideration
in~\Cref{sec:operators}. We remark that the axioms are sound, that is
\ref{enumitem:loop-proof} $\Rightarrow$ \ref{enumitem:loop-valid}. 
In~\Cref{sec:game}, we prove \ref{enumitem:loop-valid} $\Leftrightarrow$ \ref{enumitem:loop-game}
and discuss analogous statements for ordinary Weihrauch degrees and pointed degrees.
We then prove
\ref{enumitem:loop-game} $\Rightarrow$ \ref{enumitem:loop-proof} in \Cref{sec:completeness} and thus completeness of our axiomatization
for partial Weihrauch degrees.
We also discuss how one can adapt those proofs for pointed degrees
and briefly comment on the use of $\sqcap$ in the signature, but leave a
finding a nice complete axiomatization of $(\Wei, \sqcap, \sqcup, \star, 0, 1, (-)^\diamond)$ as an open problem. 
We finally establish the results of \Cref{fig:complexity-summary} in \Cref{sec:complexity}.

\subsection*{Related work}

Most of the energy of the community in Weihrauch complexity 
is arguably focused on studying specific degrees or properties of discontinuous
functions. In contrast to, say, Turing degrees, many Weihrauch problems model more
popular problems in analysis. We refer to~\cite{survey-brattka-gherardi-pauly} for an overview of Weihrauch complexity.

Other works touching on the theory of the Weihrauch lattice beyond a listing of
useful properties of operators include~\cite{lmpsv} and~\cite{NPP24}; the former
studies purely order-theoretic properties of the lattice, while the latter is
a more direct predecessor of this work which considers the parallel product instead
of $\star$.

It is striking that the soundness of axioms proposed here and in~\cite{NPP24} does
not follow from any sort of refined computability-theoretic considerations.
This is because the notion of Weihrauch reductions and the algebra of (partial) Weihrauch
problems actually correspond to well-behaved categories of \emph{containers}
within the world of type-2 computability~\cite{PricePradic25} (i.e. over regular projective
represented spaces and computable maps).
All the aforementioned axioms hold generically in categories of containers
over categories with sufficient structure. So the main challenge here is in
establishing completeness, which happens to hold because the category of regular
projective spaces fails to be well-pointed in a very strong way. In this paper,
this manifests whenever we ask for certain antichains in the Turing degrees.

So while we nominally only discuss partial Weihrauch degrees in this paper (for a
combination of historical and legibility reasons), we may confidently
conjecture that the proofs we offer work with minute modifications to completely axiomatize
the existence of morphisms between containers, or equivalently, strong natural
transformation between polynomial functors over extensive locally cartesian closed categories
with $\cW$-types.

We can also
confidently conjecture that the proof we offer work as-is for the \emph{extended Weihrauch
problems} introduced in~\cite{Bauer22}. Extended Weihrauch problems generalize
the partial Weihrauch problems by introducing a noncomputable part in forward reductions.
These extended degrees and their type-1 counterparts have received sustained recent
attention, especially those degrees closed under iterations~\cite{kihara2022rethinking,kihara2023lawvere,maschio2025,abou2026order, kihara2026katvetov1}.
Extended Weihrauch degrees also admit a presentation in terms of containers over regular projective
multirepresented spaces~\cite{PP26}.
But, contrary to the partial Weihrauch degrees, they may also be presented as internally monotone operators $j : \Omega \to \Omega$
in the Kleene-Vesley topos
(see~\cite[Theorem 3.1]{kihara2023lawvere}\footnote{This can be adapted more
generically to prove an equivalence of propositional containers in an arbitrary topos
with internally monotone operators in the style of~\cite{AhmanBauer24,AhmanBauer26}.}).
We may thus also confidently conjecture that, more generally, our axioms are sound when
interpreted in the internally monotone operators $j : \Omega \to \Omega$ of
any elementary topos $\cE$ (where the order is taken pointwise and $j \star j' = j' \circ j$),
and that it is complete for single inequalities whenever there are strong antichains
of arbitrary finite sizes in the lattice of subterminals of $\cE$. Note that this does
\emph{not} hold in well-pointed toposes, such as the effective topos used to
model the usual (type-1) notion of computability on finite objects.
The more visible difference is that the axiom $1 \le a$ always hold; even if it were
the only difference, then a proof of completeness (for games) would presumably still require
very different ideas than those presented here.

Finally, another aspect of our work here is that our axiomatizations and
games seem to capture a notion of simulation between alternating automata, modulo
the removal of the left-half-distributivity axiom in~\Cref{fig:axioms} and $\textsc{junk}$
moves in~\Cref{def:simgame}. Removing these two axioms seem to capture
the usual notion of simulation between non-deterministic automata.
This hints at potential for further application of $\SRKAM$ and our techniques
in mathematics and computer science. There is in particular a rich literature in concurrency
theory concerned with axiomatizing bisimulation and simulations
in various process calculi, sometimes described by regular
expressions~\cite{AcetoFGI14, BaetenCG07, GrabmayerF20, Milner84, SilvaBR10}.
We are not aware of close variations of the simulation game we propose nor
$\SRKAM$ appearing in the literature before, so we hope there might be further
connections to draw, especially with extensions of the present paper we suggest in the conclusion.

\subsection*{Notations and conventions}
We may use $\Nat$ or $\omega$ for the set of natural numbers. For any sort of
fixed set which computably embeds in $\Baire$ in a reasonable way (such as $\Nat$
or $\RegE{\Sigma}$), we write $\code{-}$ for that embedding. We also
assume that for terms of $\RegE{\Sigma}$, the coding is such that there is a fixed computable total order that agrees with the subterm ordering.
Given a finite sequence $p$ over $A$ and some $m \in A$, we write $p \cons m$
for the sequence obtained by appending $m$ at the end of $p$. If $p$ is
non empty, we write $\last(p)$ for its last element.

As the sequential composition of Weihrauch problem involves encoding of functions,
we will find it at times convenient to sometimes adopt notations that come from
the $\lambda$-calculus, that may elegantly be encoded in Turing machines
using the notion of \emph{partial combinatory algebra} (PCA). In short, a PCA
is a set $A$ together with a partial application operation $\cdot : A^2 \partto A$ and
combinators that allow to encode the full untyped $\lambda$-calculus. We refer
to~\cite{VanOosten} for formal definitions and an introduction. The broad idea
is that elements $e \in A$ are codes for the partial functions
$x \mapsto e \cdot x$ over $A$. In this paper we will consider the PCA
$\mathcal{K}_2 = (\Baire, \cdot)$ known as \emph{Kleene's second algebra}.
One idea to define the application is then a code $e  \in \Baire$ encodes the number of a
type 2 Turing machine \emph{and} an advice string $\alpha \in \Cantor$ (and that
$\Baire$ can certainly be encoded in $\Cantor$), from which one can cook up a partial
computable function $\Baire \partto \Baire$ using a universal Turing machine.
There is also another equivalent strategy that views labelled trees as codes
for continuous functions; see e.g.~\cite[Section 1.4.3]{VanOosten} for details,
and~\cite[Appendix A]{golov2023embeddings} for a proof of equivalence.
As usual, we call $\mathcal{K}_2^\mathrm{rec}$ the PCA whose carrier consists of computable
elements of $\Baire$ and with application defined in the same way as $\mathcal{K}_2$.
It is a sub-PCA, i.e., the interpretation of $\lambda$-terms (with parameters in $\mathcal{K}_2^\mathrm{rec}$)
is the same in $\mathcal{K}_2$ and $\mathcal{K}_2^\mathrm{rec}$.

With that in mind, we will use the application operation $\cdot$ in the sequel.
Since $\mathcal{K}_2$ is a PCA, we have a sensible coding of tuples $\tuple{p_1, \ldots, p_n} \in \Baire$
for  $p_1, \ldots, p_n \in \Baire$
(i.e., we have a computable family of
projections $\pi_i : \Baire \partto \Baire$, we can compute the length of a tuple from its code,
and tuples are uniquely represented).
Similarly, we have sensibly coded tagged
unions, i.e., we have computable injections $\inc_1, \inc_2 : \Baire \to \Baire$
with disjoint and computable images, as well as computable partial inverses.
We will sometimes write $\lambda$-terms with tupling and parameters in $\mathcal{K}_2$,
which are always meant to be interpreted in $\mathcal{K}_2$.
We will also allow ourselves to make recursive definitions, using implicitly
a standard fixpoint combinator such as $\lambda f. \; (\lambda x. \; f \; (x \; x)) \; (\lambda x. \; f \; (x \; x))$.

\section{Introducing the operators and the axioms}
\label{sec:operators}

\subsection{Partial Weihrauch problems, reduction and operators}

We begin by formally defining the notions related to Weihrauch
reducibility and the operators in question. All definitions will be,
up to details, completely standard, except maybe for the notion of
\emph{partial Weihrauch problem} (and the relevant degrees)
that we introduce.

The notion of partial Weihrauch problems and reductions below is equivalent to 
the (posetal reflection) of the category of containers over partitioned modest
sets described in~\cite[\S 3.3]{PricePradic25}, which only extends usual Weihrauch
problems by allowing questions that have no answers. Those are a strict subclass
of the extended Weihrauch degrees of~\cite{Bauer22}, which
captures a more general class of problems with advice.

\begin{figure}

\begin{center}
  \includegraphics[scale=0.8]{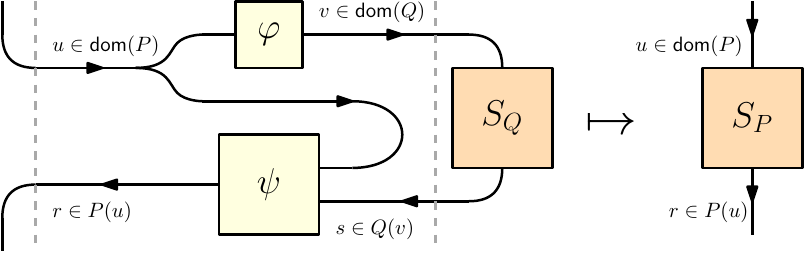}
\end{center}
  \caption{Informal picture of a Weihrauch reduction $(\varphi, \psi)$ from $P$ to $Q$
and how it turns
a solver $S_Q$ for $Q$ into a solver $S_P$ for $P$.}
\label{fig:weired}
\end{figure}

\begin{figure}
\begin{center}
\includegraphics[width=\linewidth]{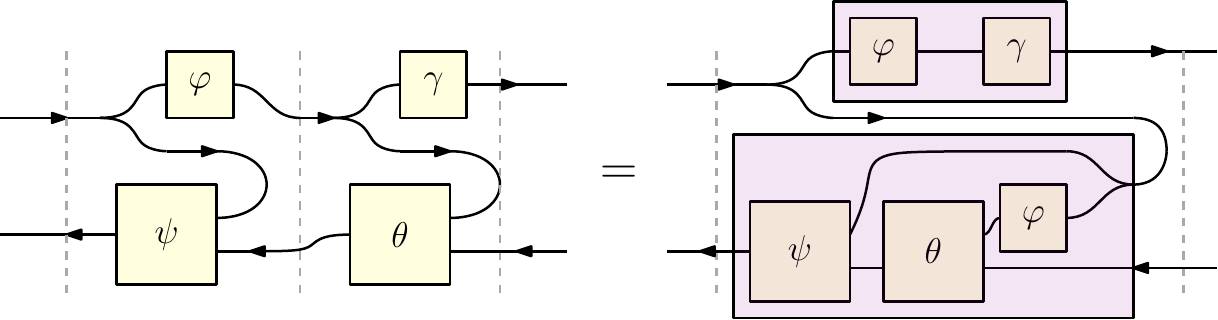}
\end{center}
\caption{Informal picture of the composition of the Weihrauch reductions
$(\gamma, \theta)$ and $(\varphi, \psi)$. The left-hand side illustrate the
intuitive plugging of the picture in~\Cref{fig:weired}, to which the right-hand
side is equivalent by re-arranging the position of $\varphi$. From this, the formal
definition of the composition can be read off: 
$(\gamma \circ \varphi, \psi \circ \langle \pi_1, \theta \circ (\varphi \times \id)\rangle)$.}
\label{fig:weiredcomp}
\end{figure}

\begin{defi}
\label{def:WeihrauchProblem}
An \emph{partial Weihrauch problem} $P$ is given by:
\begin{itemize}
\item a subset $\dom(P) \subseteq \Baire$ of \emph{instances} (that we may think
of as questions) and
\item for every $u \in \dom(P)$, a subset $P(u) \subseteq \Baire$ of \emph{solutions}.
\end{itemize}
$P$ is called \emph{pointed} if $\dom(P)$ has a computable point.
A \emph{Weihrauch problem} is a partial Weihrauch problem $P$ such that
every question has a solution, i.e., for every $u \in \dom(P)$, $P(u) \neq \emptyset$.

A \emph{Weihrauch reduction} $(\varphi, \psi)$ from problem $P$ to problem $Q$ is given by computable maps $\varphi : \dom(P) \to \dom(Q)$ and
$\psi$, where $\psi$ is defined over every pair $(u,y)$ with $u \in \dom(P)$ and $y \in Q(\varphi(u))$ and we have $\psi(u,y) \in P(u)$
(see~\Cref{fig:weired}). When there exists such a reduction
from $P$ to $Q$, we write $P \leqW Q$. Reductions compose (see~\Cref{fig:weiredcomp}) and there is
an identity reduction witnessing $P \leqW P$ for every $P$, so $\leqW$ is a preorder.
The equivalence classes of ${\equivW} = {\leqW} \cap {\geqW}$ are called \emph{partial Weihrauch degrees}\footnote{Or just Weihrauch degree if a Weihrauch problem belongs to the class in question.}.
\end{defi}

We then formally introduce all the
operators relevant to this paper in \Cref{fig:operators}.
The paradigmatic example of a partial Weihrauch problem which is not a Weihrauch
problem is $\top$, which we\footnote{
While our $\top$ coincides with the one of Bauer's extended Weihrauch degrees (up to
a trivial embedding), note that it does not coincide with every extension of the
Weihrauch degrees with a $\top$ element that have been considered in the past. For
instance, in the setting of~\cite{paulybrattka4}, $\infty$ does not satisfy the same
inequalities as $\top$ in the partial Weihrauch degrees with respect to $\star$.
}
define by taking $\dom(\top) = \{\unitelt\}$
and $\top(\unitelt) = \emptyset$. It is easy to check that for every problem $P$, there
is a unique reduction witnessing $P \leqW \top$.

All the rest of the operations and problems of \Cref{fig:operators} preserve or belong to
the class of ordinary Weihrauch problems. The meet and join operations
give the Weihrauch degrees a distributed lattice structure with a least element $0$.
The unit $1$ is the problem with a single computable question which admits an
equally computable answer; it morally corresponds to a trivial oracle in the
sense that if $P \leqW 1$, then there is a type-2 computable map taking questions
of $P$ to suitable $P$-answers.
The composition operator $\star$ models the composition of problems:
an input for $P\star Q$ consists of a question $q$ for $Q$ and a code for a map $f$ taking an answer to $q$
to a question for $P$. The answer to the question $(q, f)$ to $P \star Q$ is then
a pair $(a, b)$ where $a$ answers $q$ and $b$ answers $f(a)$.
Reducing a problem $R$ to $P \star Q$ essentially amounts to having the right (and obligation)
to make an oracle call to $Q$ and then a call to $P$ before returning an answer to $R$.
The unit of this operation is the problem $1$, which is the identity over a computable
singleton\footnote{The identity over any set containing a computable point would be
Weihrauch equivalent to that as well, and often in the literature $1$ is defined
as the identity on Baire space (as in~\cite{survey-brattka-gherardi-pauly}). If one looks at the finer notion of isomorphism
in containers (as in~\cite{PricePradic25}), the correct unit for both composition and
parallel product is the one used here.}.
Iterated composition, as the name implies, allows to make a finite number of
calls to an oracle solving $P$ in a reduction to $P^\diamond$.
$P^\diamond$ is defined as a least fixpoint of the following map of problems: $X \mapsto 1 \sqcup (X \star P)$.
As a consequence, an input to $P^\diamond$ is either the trivial tag
$\inc_1(\langle\rangle)$ which calls for an equally trivial answer, or a tagged
pair $\inc_2(u, f)$ where $u \in \dom(P)$ is an initial question to $P$, and
$f$ codes a map $P(u) \to \dom(P^\diamond)$. An answer to such a question is a pair
of an answer $a$ to $u$ and of an answer to $f(a)$; unravelling the recursion, this
means that an input to $P^\diamond$ is a recipe to make finitely many calls to $P$,
and an answer is a list of consistent $P$-answers.
Note that since $P$ may answer non-deterministically, there is not necessarily
a uniform bound on the ``numbers of calls'' to $P$. The idea is rather that
a $P^\diamond$-question implicitly encode an interaction tree which, while well-founded, may have
a very high branching factor. Let us give a (somewhat contrived) example.

\begin{exa}
  \label{ex:omega2}
Call $\mathfrak{T} \subseteq \powerset(\Cantor)$ the set of Turing degrees
and $\mathfrak{c} = 2^{\aleph_0}$ be the cardinality of the continuum.
Using the axiom of choice, pick an $\mathfrak{c}$-sequence $\alpha : \mathfrak{c} \to \mathfrak{T}$
whose image is an antichain; large enough antichains can be found as there are
$\mathfrak{c}$-many minimal degrees (see e.g.~\cite[Exercise 13-35]{rogers1987}). Let us consider a problem 
$\mathsf{down}_\alpha$, where questions are either the computable point $\langle\rangle \in \Cantor$
or some $p$ in some $\alpha(i)$, and answers as follows:
\begin{itemize}
  \item to $\langle\rangle$, $\mathsf{down}_\alpha$ is free to answer any point
    in $\bigcup\limits_{ j < \mathfrak{c}} \alpha(j)$
  \item to $p \in \alpha(i)$, $\mathsf{down}_\alpha$ answers with a pair $(b, p) \in 2 \times \Cantor$
    such that either $p \in \alpha(j)$ for $j < i$ and $b = 0$, or
    with $p \in \alpha(0)$ and $b = 1$ 
\end{itemize}
Then consider the uncomputable problem $\mathsf{get}_{\alpha(0)}$ which ignores its trivial input
and returns some point in $\alpha(0)$. It is reducible to $\mathsf{down}_\alpha^\diamond$:
first call $\mathsf{down}_\alpha$ with $\langle\rangle$
and get some $p_0$, and then get some pair $(b_0, p_1) \in \mathsf{down}_\alpha$ with another call.
Generally, we run into a loop with some $(b_n, p_{n+1})$ at hand.
If $b_n = 1$ we can halt as $p_{n + 1} \in \alpha(0)$, otherwise
we get some $(b_{n+1}, p_{n+2}) \in \mathsf{down}_\alpha(b_n, p_{n+1})$.
By induction over $\mathfrak{c}$, this loop is guaranteed to terminate.
The code for this procedure is obviously computable, but there is no uniform
bound on the number of calls to $\mathsf{down}_\alpha$.
\end{exa}

\begin{rem}
There is a natural notion of rank $\mathrm{rk}$ for inputs to $P^\diamond$.
\[ \mathrm{rk}(\inc_1(\langle \rangle)) = 0 \qquad \qquad \text{and} \qquad \qquad \mathrm{rk}(\inc_2(\langle u , f \rangle)) = \sup \{ \mathrm{rk}(f
\cdot x) + 1 \mid x \in P(u)\}\]
\Cref{ex:omega2} shows a reduction that produces
a computable input to some $P^\diamond$ of rank $\mathfrak{c}$.
\end{rem}

\begin{rem}
$(-)^\diamond$ is the more prevalent notion of sequential iteration in Weihrauch
degrees which has been used throughout the literature since its introduction~\cite{topol-comput-neumann-pauly}.
It has nice closure properties, in the sense that it is a closure operator (${P^\diamond}^\diamond \le P^\diamond$)
and it generates problems closed under composition ($P^\diamond \star P^\diamond \le P^\diamond$).
The following other constructions that capture variants of sequential iterations:
\begin{itemize}
  \item The iteration $P^{[*]}$ where we additionally ask the input to give
  the exact number of oracle calls to $P$ along any run. They are called \emph{for loops}
  in~\cite{brattka2025loops} and it can be computed as a least fixpoint
  of $X \mapsto 1 + P \star X$ in the spirit of~\cite[\S 4.2]{PP26}.
  Note that $(P + 1)^{[*]}$ correspond to the relaxation where a (finite) upper bound
  on the number of calls to $P$ must be given.
\item An infinite iteration $P^{\infty}$ introduced in~\cite{brattka2025loops}
  where $\omega$ successive calls to $P$ are allowed.
\item A more general ordinal iteration $P^{(\dagger)}$ that allows to further
  iterate $P$ along any countable ordinal coded in a suitable way as part of
  the input~\cite{paulycountableordinals}.
\end{itemize}
Those operators sit in a hierarchy, that can be strict for suitable $P$s.
\[ P^{[*]} \le {\left. P^{[*]}\right.}^{[*]} \le
  \ldots \le P^\diamond = {P^\diamond}^\diamond \le (1 \sqcup P)^\infty \le {(1 \sqcup P)^{\infty}}^\infty
\le \ldots \le P^{\dagger} = {P^{\dagger}}^\dagger\]
Finding purely equational characterization of our signature $(-)^\infty$ and
$(-)^\dagger$ sounds tricky. Getting a characterization via games of length  $< \omega_1$ on finite
graphs is left for exciting future work.
\end{rem}

\begin{figure}
\[
\begin{array}{rcl !\qquad rclr}
\dom(0) &=& \emptyset & & &
& \text{\footnotesize (bottom)}
\\
\\
\dom(P \sqcup Q) &=&
\multicolumn{4}{l}{\{\inc_1(u) \mid u \in \dom(P)\} \cup \{\inc_2(v) \mid v \in \dom(Q)\}}
& \text{\footnotesize (join)}\\
(P \sqcup Q)(\inc_1(u)) &=& P(u) &
(P \sqcup Q)(\inc_2(v)) &=& Q(v) \\
\\\\
\dom(\top) &=& \{\unitelt\}&
\top(\unitelt) &=& \emptyset &  \text{\footnotesize (top)}
\\\\
\dom(P \sqcap Q) &=&
\multicolumn{4}{l}{\{\tuple{u, v} \mid u \in \dom(P), v \in \dom(Q)\}}
& \text{\footnotesize (meet)}\\
(P \sqcap Q)(\tuple{u, v}) &=&
\multicolumn{4}{l}{\{ \inc_1(x) \mid x \in P(u)\} \cup
\{ \inc_2(y) \mid x \in Q(v)\}} \\
\\
\dom(1) &=& \{\unitelt\} &
1(\unitelt) &=& \{\unitelt\}
 & \text{\footnotesize (unit)}
\\\\
\dom(P \star Q) &=& \multicolumn{4}{l}{\{ \tuple{v, f} \mid v \in \dom(Q), \forall y \in Q(v). \; f \cdot x \in \dom(P)\}}  &\text{\footnotesize (composition)} \\
(P \star Q)(\tuple{v,f}) &=& \multicolumn{4}{l}{\{\tuple{y,x} \mid y \in Q(v), x \in P(f \cdot y)\}}
\\\\
\dom(P^\diamond) &=&
\multicolumn{5}{l}{\{\inc_1(\unitelt)\} \cup \{ \inc_2(\tuple{u, f}) \mid u \in \dom(P), \forall x \in P(u). \; f \cdot x \in \dom(P^\diamond))\}}  \\
(P^\diamond)(\inc_1(\unitelt)) &=& \{\unitelt\} & & & \multicolumn{2}{r}{\text{\footnotesize (iterated composition)}}\\
(P^\diamond)(\inc_2(\tuple{u,f})) &=&
\multicolumn{4}{l}{\{\tuple{x,y} \mid x \in P(u), y \in P^\diamond(f \cdot x)\}}
\\\\
\end{array}
\]
\caption{Operators on partial Weihrauch problems and assorted constants,
assuming $P$ and $Q$ are partial Weihrauch problems.
$\dom(P^\diamond)$ should be the smallest subset of $\Baire$ satisfying the
advertised conditions; the sets of valid outputs of $P^\diamond$ are then
computed by well-founded recursion on the inputs.}
\label{fig:operators}
\end{figure}

We now turn to the definition of the syntax of $\SRKAM$ and of universal
validity in the (partial) Weihrauch degrees/problems.

\begin{defi}
\label{def:regE}
When $\Sigma$ is a finite alphabet, call $\RegE{\Sigma}$ the set of expressions
(that we sometimes also call terms)
generated by the following grammar (we overload the notations for the operators)
\[
e, f \bnfeq
a \in \Sigma \bnfalt
0 \bnfalt
e \sqcup f \bnfalt
\top \bnfalt
e \sqcap f \bnfalt
1 \bnfalt
e \star f \bnfalt
e^\diamond
\]
We define the size
of an expression $e$ in the obvious way:
\[
  \begin{array}{l!\qquad lcll}
    &\size{e\mathrel{\boxempty} f} &=& \size{e} + \size{f} + 1 & \text{for $\boxempty \in \{\star,\sqcap, \sqcup\}$}\\
    &\size{e^\diamond} &=& \size{e} + 1 &\\
    \text{and} &
    \size{e} &=& 1 & \text{otherwise.}
  \end{array}
\]

An \emph{interpretation} $\rho$ of variables of $\Sigma$, that is,
a map from $\Sigma$ to partial Weihrauch problems,
extends to an interpretation $\interp{-}_\rho$ from $\RegE{\Sigma}$ to
partial Weihrauch problems by taking $\interp{a}_\rho = \rho(a)$,
$\interp{\mathbf{k}}_\rho = \mathbf{k}$ for $\mathbf{k} \in \{0,1,\top\}$,
$\interp{e^\diamond}_\rho = \interp{e}_\rho^\diamond$ and
$\interp{e \boxempty f}_\rho = \interp{e}_\rho \boxempty \interp{f}_\rho$
for $\boxempty \in \{\sqcap, \sqcup, \star\}$.
Given $e, f$ in $\RegE{\Sigma}$, we say that $e \le f$ is universally valid
in the (pointed and/or partial) Weihrauch degrees/problems if for
every $\rho$ valued in the (pointed and/or partial) Weihrauch problems,
$\interp{e}_\rho \leqW \interp{f}_\rho$ holds.
\end{defi}

\subsection{$\SRKAM$ is sound for the partial Weihrauch degrees}

Recall the axiomatization of $\SRKAM$ given in \Cref{fig:axioms}.
The first step is to show that these axioms are valid in the 
partial Weihrauch degrees. The only non-trivial axiom to
check is parameterized $\diamond$-induction.
To do so, one needs to generalize the construction of
of~\cite{westrick2020}.

\begin{lem}
\label{lem:diamond-fp}
For any partial Weihrauch problems $P, Q$ and $R$ with
${R \sqcap (P \star Q)} \leqW P$, we have
$R \sqcap (P \star Q^\diamond) \leqW P$.
\end{lem}

Before giving some formal details, let us try to give some intuition
as to how this works in a restricted case.
The following scheme, that we call $\diamond$-induction, is admissible in
$\SRKAM$ by taking one parameter to be $\top$.
\[ a \star b \le a \qquad \Longrightarrow \qquad a \star b^\diamond \le a\]
\begin{rem}
It is also admissible in $\SRKAM$ restricted to $\top$-free
expressions, as the premise of $\diamond$-induction also implies
$(a \star b^\diamond) \sqcap (a \star b) \le a$ and $\sqcap$ is an idempotent operator.
\end{rem}

So let us intuit how one deals with building a reduction $(\varphi^\diamond, \psi^\diamond) : P \star Q^\diamond \leqW P$
from a reduction $(\varphi, \psi) : P \star Q \leqW P$.
The forward pass $\varphi^\diamond: \dom(P \star Q^\diamond) \to \dom(P)$ is
the more instructive. Writing $\langle h, g\rangle \in \dom(P \star Q^\diamond)$,
we can define $\varphi^\diamond$ by recursing over $h$, taking care that $\mathrm{rk}(h)$ decreases strictly:
\begin{itemize}
\item when $h = \inc_1(\langle\rangle)$ and no question are asked of $Q$,
  we can recover a question to $P$ and
  simply pass it along. Formally, this is done by setting 
  $\varphi^{\diamond}(\langle \inc_1(\langle\rangle), g \rangle) = g \cdot \langle\rangle \in \dom(P)$.
\item otherwise, $h = \inc_2(\langle v, f\rangle)$ where $v$ is a question to $Q$ and
$f$ is a recipe to cook up a question to $Q^\diamond$ from an answer $y$ to $v$ (recall then that
$g$ in turn can cook up a question to $P$ from $y$ and the answers $y'$ induced
by $f$). The informal idea is then that by re-arranging the tuple $\langle h, g \rangle$
and currifying $g$ into $\tilde{g}$ such that $g \cdot \langle y , y' \rangle = \tilde{g} \cdot y \cdot y'$, the input
$\langle h, g \rangle$ can be turned into
\[\langle v, \lambda y. \; \langle f \cdot y, \tilde{g} \cdot y\rangle\rangle \in \dom((P \star Q^\diamond) \star Q) \qquad \qquad \text{with $\mathrm{rk}(f \cdot y) < \mathrm{rk}(\inc_2
(\langle v, f\rangle))$}\]
So we may recursively get a valid code $m$ for a map
\[
\begin{array}{llcl}
Q(v) &\longto& \dom(P)\\
y &\longmapsto& \varphi^\diamond(\langle f \cdot y, \tilde{g} \cdot y \rangle)
\end{array}\]
Then we have that $\langle v, m\rangle \in \dom(P \star Q)$ by definition. We
can then apply $\varphi$ to get $\varphi(\langle v, m\rangle) \in \dom(P)$.
\end{itemize}
Now that we have a well-defined forward function $\varphi^\diamond$, we can
similarly define $\psi^\diamond$, which will again take $\langle h, g\rangle \in \dom(P \star Q^\diamond)$
as first argument and then some $a \in P(\varphi^\diamond(\langle h, g\rangle))$.
The recursion will be well-defined again by making the rank of $h$ decrease across recursive calls.
The base case where $h = \inc_1(\langle\rangle)$ is again easy, so let us focus
on $h = \inc_2(\langle v, f\rangle)$, in which case we know that
$\varphi^\diamond(\langle h, g\rangle) = \varphi(\langle v, m\rangle)$ as before.
Then we know that $\psi(\langle v, m\rangle, a) \in (P \star Q)(\langle v, m \rangle)$,
so $\psi(\langle v, m\rangle, a) = \langle y, a' \rangle$ with $y \in Q(v)$ and
$a' \in P(m \cdot y)$. Recalling how $m$ is defined, this means that, by recursion, we have
\[ \psi^\diamond(\langle f \cdot y, \tilde{g} \cdot y\rangle, a') \in (P \star Q^\diamond)(\langle f \cdot y, \tilde{g} \cdot y)\]
(note in particular that $\mathrm{rk}(f \cdot y) < \mathrm{rk}(h)$, which makes
this valid). Then an answer to the
original instance can be obtained by joining $y$ to answer the initial value $v$
(and moving the parentheses in the tuple).

This concludes the informal exposition of how to prove soundness of $\diamond$-induction.
Formally, it should be noted that one can simply perform the definition of $\varphi^\diamond$
and $\psi^\diamond$ in the PCA $\mathcal{K}_2^{\mathrm{rec}}$, by fixing
codes for $\varphi, \psi$ and using a fixpoint combinator to perform the
recursion. Then correctness and totality can be checked a posteriori, by induction
on ranks of would-be inputs.

The proof sketch we offer of parameterized $\diamond$-induction follows this pattern.
We offer less justifications as the rationale for the definitions is
the same. The key difference is that the question $w$ corresponding to the extra parameter
needs to be passed along throughout for the benefit of $\varphi$, and that $\psi$ may
decide to answer $w$ rather than $\langle v, m\rangle$; in which case $\psi^\diamond$
can simply pass along the answer to $w$.

\begin{proof}[Proof of \Cref{lem:diamond-fp}]
Assume that we are given a forward functional $\varphi$
and a backward functional $\psi$ witnessing
$R \sqcap (P \star Q) \leqW P$. We can define directly
a pair of functionals $(\varphi^\diamond, \psi^\diamond)$
witnessing $R \sqcap (P \star Q^\diamond) \leqW P$
by recursion as follows (we adopt the same notational convention as above for
variables, with the addition that $w \in \dom(R)$, $z \in R(w)$, and $a''$ is introduced
as a name for a solution to a $P$ question (that was left implicit above in the
very last step)):
\[
\begin{array}{lcl}
\varphi^\diamond(\tuple{w, \tuple{\inc_1(\unitelt), g}})
&=& g \pcaapp \unitelt
\\
\varphi^\diamond(\tuple{w, \tuple{\inc_2(\tuple{v, f}), g}})
&=& \varphi(\tuple{w, \tuple{v, \Phi^\diamond(w,v,f,g)}})\\
\Phi^\diamond(w,v,f,g) &=&\lambda y. \; \varphi^\diamond(\tuple{w, \tuple{f \pcaapp y, \lambda y'. \; g \pcaapp \tuple{y, y'}}})
\end{array}
\]
\[
\begin{array}{lcl}
\psi^\diamond(\tuple{w, \tuple{v , \inc_1(\unitelt)}}, a)
&=& \inc_2(\tuple{\unitelt,a})
\\
\psi^\diamond(\tuple{w, \tuple{\inc_2(\tuple{v, f}), g}},a)
&=& \left\{ \begin{array}{ll}
\multicolumn{2}{l}{\inc_1(z)}
\\
& \text{when $\psi(\tuple{w,\tuple{v, \Phi^\diamond(w,v,f,g)}}, a) = \inc_1(z)$}  \\
\multicolumn{2}{l}{\Psi^\diamond(w,v,f,h,y,a')}\\
& \text{when $\psi(\tuple{w,\tuple{v, \Phi^\diamond(w,v,f,g)}}, a) = \inc_2(\tuple{y,a'})$} \\
\end{array}\right.\\
\Psi^\diamond(w,v,f,h,y,a')
&=& \left\{ \begin{array}{ll}
\multicolumn{2}{l}{\inc_1(z)} \\ & \text{when $\psi^\diamond(\tuple{w, \tuple{f \pcaapp y, \lambda r. \; g \pcaapp \tuple{y, r}}}, a') = \inc_1(z)$}  \\
\multicolumn{2}{l}{\inc_2(\langle \langle y, y'\rangle, a''\rangle)} \\ & \text{when $\psi^\diamond(\tuple{w, \tuple{f \pcaapp y, \lambda y'. \; h \pcaapp \tuple{y,
y'}}}, a') = \inc_2(\langle y',  a''\rangle)$} \\
\end{array}\right.\\
\end{array}
\]
We are only using operators of the $\lambda$-calculus, recursion
and $(\varphi, \psi)$, so everything given codes for $\varphi$ and $\psi$ in $\mathcal{K}_2^\mathrm{rec}$, one
can define suitable codes for $\varphi^\diamond$ and $\psi^\diamond$. 
That this is a reduction for every possible input $\tuple{w, \tuple{h, g}} \in \dom(R \sqcap (P \star Q^\diamond))$ is provable by recursion on
the rank\footnote{It is also possible to do so without reference to ordinal ranks
as well, just referring to the Knaster-Tarski theorem; but we refer to ranks
here (and later) for notational convenience.}
of $h \in \dom(Q^\diamond)$.
\end{proof}

\begin{rem}
There is also a more elegant proof of soundness of
parameterized $\diamond$-induction from $\diamond$-induction in the style of~\cite[Lemma 54]{PP26}.
The corresponding proof can be replayed with the assumption that the partial Weihrauch
degrees form an Heyting algebra, which is the case.
This Heyting algebra structure this does not appear in the literature stricto sensu (in
part because it does not restrict to usual Weihrauch degree~\cite{paulykojiro}), but this
can be proven by either abstract nonsense in the style of~\cite[\S 4.4]{PricePradic25}
or by replaying the construction in Bauer's extended degrees~\cite{Bauer22,maschio2025}.
\end{rem}

\begin{lem}
  \label{lem:soundness}
The axioms of~\Cref{fig:axioms} are sound in the partial Weihrauch degrees:
if $e \le f$ is derivable, then it is universally valid. 
\end{lem}
\begin{proof}
The validity of every axiom not involving $(-)^\diamond$ or $\top$ has been established previously for Weihrauch degrees~\cite{paulybrattka4}, and the proof that
they hold in the partial Weihrauch degrees is exactly the same.
The more
obscure one is left-half-distributivity of $\sqcap$ over $\star$, which follows
from~\cite[Proposition 7 (4)]{paulybrattka4} by taking $\mathbf{d} = 1$ there (a more direct approach in a similar setting is given in~\cite[Appendix C.2]{PP26}).
\Cref{lem:diamond-fp} covers the parameterized $\diamond$-induction axiom,
and the other two axioms involving $(-)^\diamond$ follow from $P^\diamond = 1 + P^\diamond$.
\end{proof}

To illustrate the axioms of \Cref{fig:axioms}, 
let us use them to derive some valid inequalities in the Weihrauch degrees.

\begin{exa}
\label{ex:derivation-intro}
Let us derive the inequality of \Cref{ex:intro-3}, $(b \star a^\diamond) \sqcap (c \star a^\diamond) \le
(b \sqcap c) \star a^\diamond$.
First, by left-half-distributivity of $\sqcap$ over $\star$, we have
\[
\textcolor{purple}{(b \star a^\diamond)} \sqcap (\textcolor{teal}{c} \star \textcolor{brown}{a^\diamond}) \quad \le \quad
\left(\textcolor{purple}{(b \star a^\diamond)} \sqcap \textcolor{teal}{c}\right) \star \textcolor{brown}{a^\diamond}\]
Our axioms show that $\sqcap$ is commutative in a standard way. By applying that, associativity of $\star$ and using left-half-distributivity again,
we thus get
\[
\left((\textcolor{teal}{b} \star \textcolor{brown}{a^\diamond}) \sqcap \textcolor{purple}{c}\right) \star a^\diamond
\quad
\le
\quad
\left(\textcolor{teal}{b}
  \sqcap \textcolor{purple}{c}\right)
\star
\textcolor{brown}{a^\diamond}
  \star a^\diamond
\]
and therefore that $(b \star a^\diamond) \sqcap (c \star a^\diamond) \le
(b \sqcap c)\star a^\diamond \star a^\diamond$.
It then suffices to show $a^\diamond \star a^\diamond \le a^\diamond$ to conclude, which can be done by applying the (parameterless) $\diamond$-induction
scheme
\[ \textcolor{purple}{a^\diamond} \star \textcolor{teal}{a} ~~ \le~~ \textcolor{purple}{a^\diamond} \qquad \Longrightarrow \qquad \textcolor{purple}{a^\diamond}
\star \textcolor{teal}{a}^\diamond ~~ \le ~~ \textcolor{purple}{a^\diamond}\]
Its premise is one of the fixpoint unfolding axiom, so we can conclude.
\end{exa}

Aside from $a^\diamond \star a^\diamond \le a^\diamond$, there are two other
remarkable properties of $a^\diamond$ that are derivable from $\diamond$-induction.
\begin{itemize}
\item First, that $(-)^\diamond$ is a monotone operator. Indeed, assuming $a \le b$, we
further have that $b^\diamond \star a \le b^\diamond \star b \le b^\diamond$.
Hence, by $\diamond$-induction, $b^\diamond \star a^\diamond \le b^\diamond$, which
then entails $a^\diamond \le b^\diamond$ since $1 \le b^\diamond$.
\item Second, that $(-)^\diamond$ is a closure operator: Applying the $\diamond$-induction
  scheme to $a^\diamond \star a^\diamond \le a^\diamond$, we get $a^\diamond \star {a^\diamond}^\diamond \le a^\diamond$,
which similarly entails ${a^\diamond}^\diamond \le a^\diamond$.
\end{itemize}

\begin{exa}
\label{ex:derivation}
Let us show how to establish that $a^\diamond \sqcap b^\diamond \le (a \sqcap b)^\diamond$.
Starting from the right-hand side, a fixpoint unfolding axiom yields
\[\textcolor{purple}{(a \sqcap b)}^\diamond \star \textcolor{purple}{(a \sqcap b)} \quad\le\quad \textcolor{purple}{(a \sqcap b)}^\diamond\]
Then we can use a general left-distributivity of $\star$ over $\sqcap$ (only the right-to-left
inequality is part of the official axioms, but the other is derivable purely from the
meet axioms):
\[
  (\textcolor{purple}{(a \sqcap b)^\diamond} \star \textcolor{teal}{a})
  \sqcap (\textcolor{purple}{(a \sqcap b)^\diamond} \star \textcolor{brown}{b}) \quad=\quad \textcolor{purple}{(a \sqcap b)^\diamond} \star (\textcolor{teal}{a}
  \sqcap \textcolor{brown}{b})
\]
All in all we have $((a \sqcap b)^\diamond \star a) \sqcap ((a \sqcap b)^\diamond \star b) \le (a \sqcap b)^\diamond$,
which is the premise of the following instance of parameterized $\diamond$-induction:
\[
  \textcolor{brown}{((a \sqcap b)^\diamond \star a)} \sqcap (\textcolor{purple}{(a \sqcap b)^\diamond} \star \textcolor{teal}{b}) ~~\le~~ \textcolor{purple}{(a \sqcap
  b)^\diamond}
  \quad \Longrightarrow \quad
  \textcolor{brown}{((a \sqcap b)^\diamond \star a)} \sqcap (\textcolor{purple}{(a \sqcap b)^\diamond} \star \textcolor{teal}{b}^\diamond) ~~\le~~ \textcolor{purple}{(a \sqcap
  b)^\diamond}
\]
The conclusion itself is the premise of the following variant of parameterized $\diamond$-induction,
up to commutativity of $\sqcap$
(in both premise and conclusion):
\[
  (\textcolor{purple}{(a \sqcap b)^\diamond} \mathrel{\star} {\textcolor{teal}{a}}) \sqcap
\textcolor{brown}{({(a \sqcap b)^\diamond} \star b^\diamond)} ~~\le~~
\textcolor{purple}{(a \sqcap b)^\diamond}
  \quad \Longrightarrow \quad
  (\textcolor{purple}{(a \sqcap b)^\diamond} \mathrel{\star} \textcolor{teal}{a}^\diamond) \sqcap
\textcolor{brown}{((a \sqcap b)^\diamond \star b^\diamond)} ~~\le~~
\textcolor{purple}{(a \sqcap b)^\diamond}
\]
Hence it suffices to show
$a^\diamond \sqcap b^\diamond \le
({(a \sqcap b)^\diamond} \star {a}^\diamond) \sqcap
{({(a \sqcap b)^\diamond} \star b^\diamond)}$ to conclude.
We can start by noting that the right-hand side can be rewritten
as $(a \sqcap b)^\diamond \star (a^\diamond \sqcap b^\diamond)$ by using 
the left-distributivity of $\star$ over $\sqcap$. Then we can use
the fixpoint unfolding  axiom $1 \le \textcolor{purple}{(a \sqcap b)}^\diamond$
to get 
$1 \star (a^\diamond \sqcap b^\diamond) \le (a \sqcap b)^\diamond$ and finally
the unitality of $\star$ to have $a^\diamond \sqcap b^\diamond$ below all of that.
\end{exa}

\subsection{Incompleteness for Weihrauch problems}

As $\SRKAMT$ is sound for partial Weihrauch degrees, it is also sound for
Weihrauch problems. It is however not complete.
In $\SRKAM$, we do not have any proof of $0 \star a \le a$, which would be unsound
for $a = \top$. However, that principle is true in Weihrauch degrees.

\begin{lem}
For any Weihrauch problem $P$, we have $0 \star P \leqW 0$.
\end{lem}
\begin{proof}
If $u \in \dom(P)$, then since $P$ is a Weihrauch problem, then $P(u) \neq \emptyset$
and there is no map $P(u) \to \emptyset$. A fortiori there can be no $\tuple{u, f} \in \dom(0 \star P)$.
\end{proof}

We leave it as an open problem to fully axiomatize equations in
$(\mathfrak{W}, \sqcap, \sqcup, \star, 0, 1, (-)^\diamond)$,
although we shall be able to handle pointed degrees in~\Cref{sec:completeness}.

\section{Alternating automata as problems and the simulation game}
\label{sec:game}

The main goal of this section is, given $e, g \in \RegE{\Sigma}$,
define simulation games $\simGame(\emptyset \mid \{e\} \vdash f)$
in which Duplicator has a winning strategy if and only if $e \le f$ is universally
valid in the partial Weihrauch degrees. The simulation game will be a B\"uchi
game over a finite graph, which will yield decidability of universal validity.

However, it is conceptually cleaner to prove a more general result involving
automata, where bureaucratic details pertaining to the syntax
of regular expressions are abstracted away. So we first introduce such automata $\mathfrak{A}$,
how they can be interpreted as problems $\interp{\mathfrak{A}}_\rho$ (when given
a valuation $\rho$ on their alphabet), and how a classical translation of
regular expressions $e$ to finite automata $\mathfrak{re}_\Sigma^{(e)}$
gives equivalent Weihrauch problems $\interp{\mathfrak{re}_\Sigma^{(e)}}_\rho \equivW \interp{e}_\rho$.

Then we define a simulation game where Duplicator wins if and only if $\interp{\mathfrak{A}}_\rho \leqW \interp{\mathfrak{B}}_\rho$
for all valuations $\rho$. The interesting part is the converse, which is proven
by finding a generic valuation $G$ such that Spoiler winning implies that
$\interp{\mathfrak{A}}_G \not\leqW \interp{\mathfrak{B}}_G$
when Spoiler wins. To make problems in $G$ hard enough, we crucially use that
the Turing degrees contain a strong countable antichain and that it is impossible
to continuously tell two Turing degrees apart. Then, after a technical argument,
we can conclude by determinacy of B\"uchi games.

We conclude by briefly commenting on the variant of those theorems for pointed
(partial) Weihrauch degrees.

\subsection{Alternating automata as problems, and regular expressions}
\label{subsec:transterm}

We first start with our working definition of alternating automata.
Now and for the rest of
the paper, we assume distinct \emph{polarity} symbols $\forall$, $\exists$ and
$\done$, assumed to be disjoint from the working alphabet $\Sigma$. They are
meant to inform the nature of a given state: whether the non-determinism is angelic ($\exists$)
or demonic ($\forall$), or if the state is final ($\done$).
A non-standard helpful convention we will take is to put letter labels $a \in \Sigma$
not on transitions, but on states.
As such, we will regard elements of $\Sigma$ as special polarities in our automata,
and all the actual transitions will be ``silent''.
We write $\Pol = \Sigma \cup \{\forall, \exists, \done\}$ the set of polarities.
All automata we will consider might be alternating by default.

\begin{defi}
  \label{def:automaton}
A preautomaton over the alphabet $\Sigma$ is a tuple
$\mathfrak{A} = (Q, \to, \pol)$ where:
\begin{itemize}
\item $(Q, \to)$ is a computable directed graph. Its vertices are called \emph{states}
and edges \emph{transitions}. Computability means that we are implicitly given
a canonical injection $Q \to \bN$ whose range is a computable subset and that
the lifting of the edge relation is decidable. Further we require that each
vertex has finitely many successors.
\item $\pol : Q \to \Pol$ is a polarity assignment.
States with polarity $\checkmark$ will be called \emph{final} and states $e$
with $\pol(e) \in \Sigma$ will be called \emph{oracle} states.
\item We further impose the following constraints on the graph $(Q, \to)$:
\begin{itemize}
\item if $e$ is an oracle state, it has exactly one
incoming edge and one outgoing edge in $(Q, \to)$
\item if $e$ is a final state, it has no outgoing edges.
\end{itemize}
\end{itemize}
An automaton is given by a preautomaton $\mathfrak{A}$ together with
a state $\iota$ of $\mathfrak{A}$, which we regard as the \emph{initial state}.
The \emph{accessible part} $\mfA^{(\iota)}$ of an automaton $(\mfA, \iota)$ is
its natural restriction to the set of states reachable from $\iota$.
\end{defi}

The automata pictures we have given in the introduction implicitly give
examples of such automata. We give additional examples and explain conventions
in~\Cref{fig:autos}. Before giving the semantics of these automata in
partial Weihrauch problems, let us define the translation of regular expressions
into automata.

\begin{defi}
Define the $\RegE{\Sigma}$-preautomaton $\reAut$ over $\Sigma$
as $(\RegE{\Sigma}, \to, \pol)$ where $\to$ is the least relation such that,
for every $e, f \in \RegE{\Sigma}$ we have
\[\begin{array}{cr}
  e \sqcup f \to e \qquad\qquad e \sqcup f \to f \qquad\qquad
  e \sqcap f \to e \qquad\qquad e \sqcap f \to f \\
  a \to 1 ~\text{\small (for $a \in \Sigma$)} \qquad\qquad e \star 1 \to e \qquad\qquad e^\diamond \to 1 \qquad\qquad e^\diamond \to e^\diamond \star e
\end{array}\]
and such that for every $g \in \RegE{\Sigma}$, 
$e \to f$ implies $g \star e \to g \star f$.

The polarities are given by $\pol(a) = a$ for $a \in \Sigma$ and the following clauses:
\[\begin{array}{lcl !\qquad lcl !\qquad lcl}
  \polarity(e \sqcup f) &=& \duppol
  &
  \polarity(e \sqcap f) &=& \spopol
  &
  \polarity(e \star f) &=& \polarity(f) ~~\text{if $f \neq 1$}\\
  \polarity(e^\diamond) &=& \duppol & &&& \polarity(e \star 1) &=& \exists
  \\
  \polarity(0) &=& \duppol
  &
  \polarity(\top) &=& \spopol
&
\polarity(1) &=& \done
\end{array}
  \]
The automaton associated to some $e \in \RegE{\Sigma}$ is $\reAut^{(e)}$.
We call states of $\reAut^{(e)}$ \emph{spawns} of $e$.
\end{defi}

\begin{rem}
It would be equivalent to take the more standard approach of defining
counterpart of the operators of $\RegE{\Sigma}$ as operation over automata, but
we find it more notationally convenient to proceed with $\reAut$ in this paper.
\end{rem}

\begin{figure}
  \begin{center}
\begin{tabular}{c !\qquad c}
\begin{tikzpicture}[->,>={Stealth[round]},shorten >=1pt,
                    node distance=1.5cm,semithick,
                    inner sep=2pt,bend angle=45]
  \tikzset{every state/.style={minimum size=20pt}}
  \tikzset{initial text={}}
  \tikzset{opp state/.style={draw,square,minimum size=20pt}}
  \node[initial,opp state] (top) {};
  \node[state]         (u) [above right=of top] {};
  \node[state]         (d) [below right=of top] {};
  \node[state,accepting right]         (ub) [below right=of u] {};

  \path [every node/.style={fill=\backgroundcolor,circle}]
        (top) edge [bend left] (u)
        (top) edge [bend right] (d)
        (u) edge [bend left] (ub)
        (d) edge [bend right] (ub)
        (u) edge 
        [out=-55,in=-125,looseness=7] 
        node {$a$} (u)
        (d) edge 
        [out=55,in=125,looseness=7]
        node {$b$} (d);
\end{tikzpicture}
  &
\begin{tikzpicture}[->,>={Stealth[round]},shorten >=1pt,
                    node distance=1.5cm,semithick,
                    inner sep=2pt,bend angle=45]
  \tikzset{every state/.style={minimum size=20pt}}
  \tikzset{initial text={}}
  \tikzset{opp state/.style={draw,square,minimum size=20pt}}
  \node[initial,state] (top) {};
  \node (prout) [below=of top] {};
  \node[opp state]         (low) [below=of top] {};
  \node[state, accepting right]         (acc) [right=of top] {};
  \node[state] (dummy) [right=of low] {};

  \path [every node/.style={fill=\backgroundcolor,circle}]
        (top) edge [bend left] (acc)
        (top) edge [bend right] (low)
        (low) edge [bend left] node {$a$} (dummy)
        (low) edge [bend right] node {$b$} (dummy)
        (dummy) edge [bend right] (top);
\end{tikzpicture}
\\
  $a^\diamond \sqcap b^\diamond$
  &
  $(a\sqcap b)^\diamond$
\end{tabular}
  \end{center}
\caption{Two depictions of automata, which formally both have 6 states respectively.
$\exists$ and final states are circles, $\forall$ states are squares and oracle
states are just given by their polarity in $\Sigma$. No boundaries are drawn
around oracle states and the tip of the of the incoming edge is omitted.
Outgoing arrows with no target states allow to distinguish final and $\exists$ states.
The unique arrow with no source state indicates which is the initial state.}
\label{fig:autos}
\end{figure}

\begin{rem}
$\pol(e \star 1)$ could have also been $\forall$ without affecting the rest
of the paper.
\end{rem}

\begin{exa}
The 6 spawns of $(a \sqcap b)^\diamond$ are
\[
(a \sqcap b)^\diamond \qquad
(a \sqcap b)^\diamond \star (a \sqcap b) \qquad
(a \sqcap b)^\diamond \star a \qquad
(a \sqcap b)^\diamond \star b \qquad
(a \sqcap b)^\diamond \star 1 \quad \text{and} \quad
1\]
The automaton $\reAut^{((a \sqcap b)^\diamond)}$ is faithfully depicted on
the right of~\Cref{fig:autos}.
However, $\reAut^{(a^\diamond \sqcap b^\diamond)}$ is \emph{not quite} the automaton
depicted on its left, although it is semantically equivalent for all intent and
purposes.
Formally, the picture is missing states for the spawns $a^\diamond \star 1$
and $b^\diamond \star 1$.
\end{exa}

\begin{rem}
An expression $e$ has at most $2 \cdot \size{e}$-many spawns and
if $e'$ is a spawn of $e$, then $\size{e'} \le \size{e}^2$.
Furthermore the map $e \mapsto \reAut^{(e)}$ is computable in
quadratic time.
\end{rem}

We now define the semantics of finite automata $\mfA^{(e)}$, which given a valuation $\rho$
of letters into partial Weihrauch problems yields a partial Weihrauch problem
$\interp{\mfA^{(e)}}_\rho$.
At a very high level and in the style of~\cite{topol-comput-neumann-pauly}, 
the elements of $\dom(\interp{\mfA^{(e)}}_\rho)$ can be thought of as codes for
type-2 programs (and infinite advice strings) that can make oracle calls to some $\rho(a)$ for $a \in \Sigma$,
but with some bespoke restrictions. The programs should have some additional
memory $M$ that contains either a state or a finite set of states $C$.
This memory, initialized with $e$, drives
whether the program can/should terminate and make oracle calls (the program
is otherwise free to read or write other parts of memory):
\begin{itemize}
\item if $M$ contains a state $e'$, then the program cannot terminate
nor make oracle calls until it writes in $M$ a finite set $C$
of states of $\mfA$ that are winning position in a game played on the
subgraph of $(Q, \to)$ induced by states of $\mfA^{(e')}$ without intermediate
oracle states\footnote{Later in this section, we make formal a similar notion of \emph{execution strategy} for the sake of a proof.}.
\item if $M$ contains a finite set $C \subseteq Q$, then we have several scenarios:
\begin{itemize}
\item if $C$ contains any $\exists$ state, then the program is forced to loop forever
\item if $C$ contains only final states, then the program is free to declare it won't make
further oracle calls and terminate by outputting something (in $\Baire$)
\item otherwise, after saving some finite string $s$ to be used later, the program
should write an input 
to the problem
\[\bigsqcap_{e' \in C} \rho(\pol(e'))\]
and get back an answer after the oracle call that will point to
a successor state $e''$ to some $e' \in C$. The memory $M$ then gets
rewritten to be $e''$, and the program has in memory $s$, the input and all results of oracle calls.
The program can \emph{crash} if $\bigsqcap_{e' \in C} \rho(\pol(e'))$ has no
possible answers.
\end{itemize}
\end{itemize}
Then $\dom(\mfA^{(e)})$ consists of the programs that always either crash or
terminate, no matter how the oracle calls are resolved. Then an answer to such
a question can be a successful execution trace, or equivalently, a list of
consistent oracle answers along a successful execution.

The informal description we just concluded has a ``big step'' flavour, where
the running of $\mfA$ is squashed into a game definition. Below we formalize an
equivalent definition which has rather a more ``small steps'' flavour; we leave
figuring out the slick formal definition of the ``big step'' definition and
proving it to be equivalent to the next definition as a somewhat tedious
exercise.

\begin{defi}
For any preautomaton $\mfA = (Q, \to, \pol)$ and valuation $\rho$
mapping $\Sigma$ to problems, let us simultaneously define problems
$\interp{\mfA^{(e)}}$ for all $e \in Q$. We define the domains inductively
as follows:
\begin{itemize}
\item If $\pol(e) = \done$, then $\dom(\interp{\mfA^{(e)}}_\rho) = \{\tuple{ }\}$.
\item If $\pol(e) = \duppol$,
$e \to e'$ and $c \in \dom(\interp{\mfA^{(e')}}_\rho)$,
then $\tuple{\code{e'},c} \in \dom(\interp{\mfA^{(e)}}_\rho)$.
\item If $\pol(e) = \spopol$, note that
$\{e' \mid e \to e'\}$ is finite and is canonically ordered by the computable
structure of $(Q, \to)$. Let 
$e'_1, \ldots, e'_n$ be the enumeration of its elements 
in order. Then for every tuple $(c_1, \ldots, c_n) \in \prod_{i = 1}^n \dom(\interpaut{e'_i}_\rho)$,
we have $\tuple{c_1, \ldots,c_n} \in \dom(\interpaut{e}_\rho)$.
\item If $\pol(e) = a \in \Sigma$ and $e \to e'$, then when $u \in \dom(\rho(a))$ and $\alpha$
is a code for a function $\rho(a)(u) \to \dom(e')$, then $\tuple{u,\alpha} \in \dom(\interpaut{e}_\rho)$.
\end{itemize}
We assign to every state $e \in Q$ and input $u \in \dom(\interpaut{e}_\rho)$
an ordinal rank that we call $\rank^{(e)}(u)$ in the obvious way, keeping the same notations
as above for successor states:
\begin{itemize}
\item for $\pol(e) = \done$, $\rank^{(e)}(\langle\rangle) = 0$,
\item for $\pol(e) = \duppol$, $\rank^{(e)}(\langle \code{e'}, c\rangle) = \rank^{(e')}(c) + 1$,
\item for $\pol(e) = \spopol$, $\rank^{(e)}(\langle c_1, \ldots, c_n\rangle) =
  \bigoplus\limits_{i = 1}^n \{ \rank^{(e_i')}(c_i) \mid 1 \le i \le n\} + 1$ where
  $\oplus$ denotes the natural sums of ordinals\footnote{Recall that $\alpha \oplus \beta$ is the smallest ordinal
    strictly dominating $\{\alpha \oplus \gamma \mid \gamma < \beta\} \cup \{ \gamma \oplus \beta \mid \gamma < \alpha\}$,
  which makes it the smallest sensible commutative and associative sum on ordinals.},
\item and for $\pol(e) = a \in \Sigma$,
$\rank^{(e)}(\langle u, \alpha\rangle) = \sup \{ \rank^{(e')}(\alpha \cdot x) \mid x \in \rho(a)(u)\} + 1$.
\end{itemize}
The solution sets $\interpaut{e}_\rho(u)$ are defined by mutual recursion on the ordinals
$\rank^{(e)}(u)$:
\begin{itemize}
\item If $\pol(e) = \done$, then $\interpaut{e}_\rho(\tuple{ }) = \{\tuple{ }\}$.
\item If $\pol(e) = \duppol$, a solution for
$\tuple{\code{e'},c} \in \dom(\interpaut{e}_\rho)$ is a solution for $c \in \dom(\interpaut{e'}_\rho)$.
\item If $\pol(e) = \spopol$, a solution for 
$\tuple{c_1, \ldots, c_n} \in \dom(\interpaut{e}_\rho)$
    is a pair $\tuple{\code{i}, s}$ where $1 \le i \le n$ and $s$ is a solution for $c_i$ in $\interpaut{e'_i}_\rho$.
\item If $\pol(e) = \orapol{a} \in \Sigma$ and $e \to e'$, a solution for $\tuple{u,f} \in \dom(\interpaut{e}_\rho)$ is a pair $\tuple{x, s}$ where $x$ is a solution for $u$ in $\rho(a)$ and
$s$ is a solution for $f \pcaapp s$ in $\interpaut{e'}_\rho$.
\end{itemize}
\end{defi}

\begin{lem}
\label{lem:regE-to-aut}
For every $e, f \in \RegE{\Sigma}$ and $\rho$, we have Weihrauch equivalences
\[\begin{array}{rcl !\qquad rcl}
 \multicolumn{6}{c}{
 \interpt{0}_\rho ~~\equivW~~ 0 \qquad  
 \interpt{\top}_\rho ~~\equivW~~ \top \qquad
 \interpt{1}_\rho ~~\equivW~~ 1 \qquad
 \interpt{a}_\rho ~~\equivW~~ \rho(a)} \\\\
 \interpt{e}_\rho \sqcup \interpt{f}_\rho &\equivW& \interpt{e \sqcup f}_\rho &
 \interpt{e}_\rho \sqcap \interpt{f}_\rho &\equivW& \interpt{e \sqcap f}_\rho \\\\
 \interpt{e}_\rho \star \interpt{f}_\rho &\equivW& \interpt{e \star f}_\rho
 & \interpt{e}_\rho^\diamond &\equivW& \interpt{e^\diamond}_\rho\\
\end{array}
  \]
\end{lem}
\begin{proof}
The first six reductions are easy, and we thus skip them.
The cases of composition and iteration are more involved, although the
proof is rather dry.

Now let us sketch how to prove 
$\interpt{e}_\rho \star \interpt{f}_\rho \equivW \interpt{e \star f}_\rho$
by programming reductions both ways.
Before going further,
note that the cases of $f = 1$ and $\pol(f) \in \Sigma$ are easily handled.
\begin{itemize}
\item For any preautomaton $\mfA$ and state $e$ of $\mfA$ with
$e \to e'$ with $\pol(e) = a \in \Sigma$, by definition, we
have the equality of problems $\interp{\mfA^{(e)}}_\rho = \interp{\mfA^{(e')}}_\rho \star \rho(a)$.
\item 
Similarly, in the specific case of $\mfA = \reAut$, we easily have that
$\interpt{e \star 1}_\rho \equivW \interpt{e}_\rho \equivW \interpt{e}_\rho \star 1$.
The first reductions are easy by inspecting the definitions: $\pol(e \star 1) = \exists$ and thus elements of $\dom(\interpt{e \star 1}_\rho)$
have shape $\langle \langle\rangle, h\rangle$ for $h$ such that $h \cdot \tuple{ } \in \dom(\interpt{e}_\rho)$,
and $\interpt{e \star 1}_\rho(\tuple{ \tuple{ }, h}) = \interpt{e}_\rho(h \cdot \tuple{ })$.
\end{itemize}

We thus have to define those reductions for other $f$s. In those cases,
they will need to be be mutually defined for all spawns of $f$ and $e \star f$
that we have not treated yet. The definitions are recursive, and can
can be a posteriori checked to be well-defined and correct by
induction of the ranks of some questions:
\begin{itemize}
\item For the reduction 
$\interpt{e}_\rho \star \interpt{f}_\rho \leqW \interpt{e \star f}_\rho$, the
forward part of the reduction $\varphi^\to_{e \star f}$ takes as input pairs $\langle u, g\rangle$
with $u \in \dom(\interpt{f}_\rho)$ and $g$ a legal continuation of $u$ into $\interpt{e}_\rho$ (i.e., such that $g \cdot x \in \dom(\interpt{e}_\rho)$
for all $x \in \interpt{f}_\rho(u)$). The backward part $\psi^{\to}_{e \star f}$ additionally
takes as input some $y \in \varphi^\to_{e \star f}(\tuple{u, g})$.
We perform a case analysis according to $\pol(f)$ to define both by a recursive
algorithm.
The recursion can then be checked to be correct and total by induction on the rank of $u$.
The cases $\pol(f) \in \Sigma \cup \{\done\}$ being already
treated above, we have two cases left:
\begin{itemize}
  \item If $\pol(f) = \exists$, $u = \langle \code{f'}, c \rangle$ with $f \to f'$ and
    $c \in \dom(\interpt{f'}_\rho)$. We set
\[
  \begin{array}{lcl}
\varphi^{\to}_{e \star f}(\langle \langle \code{f'}, c\rangle, g \rangle)
&=&
\langle \code{e \star f'}, \varphi^{\to}_{e \star f'}(\langle c, g\rangle)\rangle
\\
\psi^{\to}_{e \star f}(\langle \langle \code{f'}, c\rangle, g \rangle, x)
&=&
\psi^{\to}_{e \star f}(\langle c, g\rangle, x)
  \end{array}\]
  \item If $\pol(f) = \forall$, $u = \langle c_1, \ldots, c_n \rangle$ with $c_i \in \dom(\interpt{f'_i}_\rho)$ where $f'_i$ is the $i$th successor of $f$.
    We set the following, noting that $x$ must have shape $\tuple{\code{i}, x'}$ with $1 \le i \le n$ and $x' \in \interpt{f'_i}_\rho(c_i)$.
\[\small
  \begin{array}{lcl}
    \varphi^{\to}_{e \star f}(\tuple{\langle c_1, \ldots, c_n \rangle, g}) &=&
  \tuple{
    \varphi^\to_{e \star f'_1}(\tuple{c_1, \lambda x. \; g \cdot \tuple{\code{1}, x}}), \ldots,
    \varphi^\to_{e \star f'_n}(\tuple{c_n, \lambda x. \; g \cdot \tuple{\code{n}, x}})
}\\
    \psi^{\to}_{e \star f}(\tuple{\langle c_1, \ldots, c_n \rangle, g}, \langle \code{i}, x'\rangle) &=&
    \langle\langle \code{i}, y\rangle, z\rangle \quad \text{when $\psi^\to_{e \star f'_i}(\tuple{c_i, g}, x') = \tuple{y, z}$}
  \end{array}
\]
\end{itemize}
\item The case of the reductions $(\varphi^\leftarrow_{e\star f}, \psi^\leftarrow_{e\star f})$
  is handled similarly. Let us simply give the key definitions, first for $\pol(f) = \exists$,
when $\varphi^{\leftarrow}_{e \star f'}(c) = \tuple{c', g}$:
\[
  \begin{array}{lcl}
\varphi^{\leftarrow}_{e \star f}(\langle \code{e \star f'}, c\rangle)
&=&
\tuple{\tuple{\code{f'}, c'}, g}
\\
\psi^{\leftarrow}_{e \star f}(\langle \code{f'}, c\rangle, \langle x, y\rangle)
&=&
\psi^{\leftarrow}_{e \star f'}(c, \tuple{x, y})
  \end{array}\]
  And then for $\pol(f) = \forall$,
where $\varphi^\leftarrow_{e \star f'_i}(c_i) = \tuple{c'_i, g_i}$ for $1 \le i \le n$ and $\hat{g}$ is a code realizing
$\tuple{\code{i}, x'} \mapsto g_i
    \cdot x$:
\[\small
  \begin{array}{lcl}
    \varphi^{\leftarrow}_{e \star f}(\tuple{c_1, \ldots, c_n}) &=&
    \tuple{\tuple{c'_1, \ldots, c'_n}, \hat{g}}
\\
    \psi^{\leftarrow}_{e \star f}(\tuple{c_1, \ldots, c_n}, \tuple{\tuple{\code{i}, x'}, y}) &=&
    \tuple{\code{i}, \psi^{\leftarrow}(c_i, \tuple{x', y})}
  \end{array}
\]
\end{itemize}

This concludes our treatment of composition.

Now, turning to iterations, let us first prove that
$\interpt{e}_\rho^\diamond \leqW \interpt{e^\diamond}_\rho$. To do this,
by $\diamond$-induction, it suffices to observe that there is an easy reduction $1 \leqW \interpt{e^\diamond}_\rho$
and show
\[\interpt{e^\diamond}_\rho \star \interpt{e}_\rho \leqW \interpt{e^\diamond}_\rho\]
Since we know that $\interpt{e^\diamond}_\rho \star \interpt{e}_\rho \equivW \interpt{e^\diamond \star e}_\rho$,
it suffices to build a reduction $\interpt{e^\diamond \star e}_\rho \leqW \interpt{e^\diamond}_\rho$, which is again
trivial as $e^\diamond \to e^\diamond \star e$.

For the other direction, we want to define a reduction
\[(\varphi^\leftarrow_{e^\diamond}, \psi^\leftarrow_{e^\diamond}) \quad : \qquad \interpt{e^\diamond}_\rho \quad \leqW \quad \interpt{e}_\rho^\diamond\]
It will be obtained by composing two reductions together, using a recursive definition of $(\varphi^\leftarrow_{e^\diamond},
\psi^\leftarrow_{e^\diamond})$. In the right-to-left order, this goes as follows:
\begin{itemize}
\item First we start with a a preprocessing reduction, that inspects the input and decomposes it using our previous result on composition.
\[(\varphi_{\mathsf{pre}}, \psi_{\mathsf{pre}}) \quad:\qquad \interpt{e^\diamond}_\rho \quad \leqW \quad 1 \sqcup \interpt{e^\diamond}_\rho \star \interpt{e}_\rho\]
An important point is the following strict decrease in ranks for arbitrary $v\in \dom(\interpt{e^\diamond})$:
\[\varphi_{\mathsf{pre}}(v) = \inc_2(\langle u, g\rangle) \qquad \Longrightarrow \qquad \rank^{e^\diamond}(g \cdot x) < \rank^{e^\diamond}(v) \qquad
\text{for every $x \in \interpt{e}_\rho(u)$}\]
\item Then, we act recursively by lifting our would-be reduction
$(\varphi^\leftarrow_{e^\diamond}, \psi^\leftarrow_{e^\diamond})$ to a reduction
\[1 \sqcup \interpt{e^\diamond}_\rho \star \interpt{e}^\diamond_\rho \quad \leqW \quad 1 \sqcup \interpt{e}_\rho^\diamond \star \interpt{e}_\rho\]
Overall the composition can be seen to be well-defined by using our previous remark on ranks.
Finally we recall that the we have an equality of problems $1 \sqcup P^\diamond \star P = P^\diamond$, so the target of this second reduction is what we
want.
\end{itemize}
The preprocessing reduction is itself built as a composition of two reductions.
First, we have an easy reduction
\[(\varphi_{\mathsf{unfold}}, \psi_{\mathsf{unfold}}) \quad :\qquad \interpt{e^\diamond}_\rho \quad\leqW\quad 1 \sqcup \interpt{e^\diamond \star e}_\rho\]
such that $\varphi_{\mathsf{unfold}}(v) = \inc_2(\tuple{u , g})$ implies $\rank(v) = \rank(\tuple{u, g}) + 1$.
Next, recall we have already defined a reduction above in the proof
\[(\varphi^\leftarrow_{e^\diamond \star e}, \psi^\leftarrow_{e^\diamond \star e}) \qquad : \quad \interpt{e^\diamond \star e}_\rho \quad\leqW\quad \interpt{e^\diamond}_\rho \star
\interpt{e}_\rho\]
Coming back to the definition of $\varphi^\leftarrow_{e \star f}$, we can generally prove by induction on $\rank^{f}(u)$ that
\[\varphi^\leftarrow_{e \star f}(v) = \tuple{u, g} \qquad\Longrightarrow\qquad
\rank^{e \star f}(v) = \sup \{\rank^{e^\diamond}(g \cdot x) \mid x \in \interpt{e}_\rho\} + 1 + \rank^{f}(u)\]
Together with the observation about $\varphi_{\mathsf{unfold}}$, this allow us
to conclude $\varphi_{\mathsf{pre}}$ behaves as prescribed above on ranks.
\end{proof}

\begin{cor}
For every $e \in \RegE{\Sigma}$ and $\rho$, we have equivalences $\interp{e}_\rho \equivW \interpt{e}_\rho$.
\end{cor}
\begin{proof}
By induction over $e$, using \Cref{lem:regE-to-aut} at every step.
\end{proof}

\subsection{The simulation game}
\label{subsec:game}

We finally introduce the simulation game we are going to use to characterize
reducibility. Throughout we call our players Spoiler and Duplicator; they
should intuitively be thought as controlling respectively expressions with polarity
$\spopol$ and $\duppol$ in the simulator, and the opposite polarities
in the simulation attempts.
This game is a B\"uchi game, meaning that the plays may be
infinite, and in such a case, Duplicator wins if and only if
they visit some states infinitely often.

\begin{defi}
A \emph{B\"uchi arena} $\cA$ is a tuple
$(V, E, \pol, F)$ where $(V, E)$ is a directed graph,
$\pol : V \to \{\duppol, \spopol\}$ is a polarity function
stating which player controls the position and $F \subseteq E$
is the set of winning transitions for Duplicator. We call $V$
the set of \emph{positions} and $E$ the set of \emph{moves}.
A \emph{partial play} in $\cA$ is a non-empty sequence $(p_i)_{i}$ of length
at most $\omega$ such that $(p_i, p_{i+1}) \in E$.
A \emph{play} is a partial play that cannot be extended.
An infinite play $p$ is \emph{winning} for Duplicator if and only if
$\{ i \mid (p_i, p_{i+1}) \in F\}$ is infinite. A finite play $p$ is winning
if and only if $\pol(\last(p)) = \forall$.
A \emph{B\"uchi game} is a pair of an arena $\cA$ and a position $v$ in $\cA$;
we write those $\cA(v)$ in the sequel.
A play in $\cA(v)$ is a play that begins with $v$.
A \emph{strategy} $s$ for Duplicator in $\cA(v)$ is a prefix-closed non-empty set of partial  plays of $\cA(v)$ such that
\begin{enumerate}
\item \label{enumitem:strat1} if all prefixes of an infinite play $p$ are in $s$, so is $p$
\item \label{enumitem:strat2} for every finite partial play $p$, if $\pol(\last(p)) = \spopol$, for all $m$ such that $(\last(p), m) \in E$, we have $p \cons m \in s$
\item \label{enumitem:strat3} if $\pol(\last(p)) = \duppol$, then there exists a unique $m$ such that $(\last(p), m) \in E$ and $p \cons m \in s$.
\end{enumerate}
A strategy is winning if all maximal plays therein are. Spoiler strategies (winning or not) are
defined in the expected dual manner. A Duplicator \emph{positional strategy} $s$ is one which
is induced by a partial map $s_{pos} : \pol^{-1}(\duppol) \partto V$, i.e. it is the
least strategy closed under criteria \ref{enumitem:strat1} to \ref{enumitem:strat3} such that we additionally have
$p \in s \wedge \pol(\last(p)) = \duppol$ implies that $s_{pos}(\last(p))$ is
defined and we have $p \cons s_{pos}(\last(p)) \in s$.
\end{defi}

Recall that an infinite game is called \emph{(positionally) determined} if one of the
player has a positional winning strategy.
Payoff sets for B\"uchi games are $\mathbf{\Pi}^0_2$ and have thus long been
known to be determined~\cite{wolfe1955strict}. For convenience, we will use
a strengthening of this result which says that such games are positionally
determined, and that positional strategies can be computed efficiently if the
arena is finite.

\begin{thm}[{see e.g. \cite[Theorem 14]{2023regulargameschapter}}]
\label{thm:positional-det}
B\"uchi games are positionally determined: there is always a winning strategy
for one of the two players, and a
positional winning strategy may be computed in quadratic time from the set
of reachable positions in the game.
\end{thm}

Now, we are almost ready to define our simulation game that will give us the ability
to compare two states in a preautomaton. In all that follows $\mfA$ is a fixed preautomaton
$(Q, \to, \pol)$ over $\Sigma$.

\begin{defi}
 \label{def:simgame}
Define the simulation arena $\simGameAut$ as follows:
\begin{itemize}
\item positions are triples $(\Gamma, X, e)$ with $\Gamma \subseteq \Sigma$, $X$ a finite
subset of $Q$ and $e \in Q$. To make clear that we are talking about
positions in the game, we write them $\position{\Gamma}{X}{e}$
\item
We distinguish three kinds of positions $\position{\Gamma}{X}{e}$ according to properties of $X$ and $e$,
which will determine their polarity.
\begin{itemize}
\item If $\pol(X) \cap \{\forall, \exists\} \neq \emptyset$, we call $\position{\Gamma}{X}{e}$ a \emph{left position}
and it is controlled by Spoiler.
\item Otherwise $\pol(X) \subseteq \Sigma \cup \{\checkmark\}$ and we call $\position{\Gamma}{X}{e}$ a \emph{right position}.
Such a position is a Spoiler position if and only if $\pol(e) = \forall$ or $\pol(e) = \done \in \pol(X)$.
We sometimes say right $P$ position if $P$ is the polarity of the position.
\end{itemize}
We overload the notation $\polarity$, and use it to denote the polarity of positions.
\item The set of moves is listed in \Cref{fig:simgame-moves}. By convention,
we call a move starting from an $\alpha$ position an
$\alpha$ move for $\alpha$ being any of the adjectives ``left'', ``right'', ``right $\forall$'' or ``right $\exists$''.
\item The winning transitions for Duplicator are the left moves (\textsc{explore}, \textsc{fork})
  and the \textsc{fill} moves.
\end{itemize}
\end{defi}

\begin{figure}
\[
\begin{array}{c@{\qquad}lcl@{\quad}l}
\multicolumn{5}{c}{
\text{
\begin{tabular}{c}
\Large{Left moves}\\
(potentially available when $\pol(X) \cap \{\forall, \exists\} \neq \emptyset$)
\end{tabular}}
}
\\\\
\text{\textsc{explore}} &
\position{\Gamma}{X \uplus \{t\}}{e} &\longto& \position{\Gamma}{X \cup \{t'\}}{e} &
        \text{when $\pol(t) = \duppol$ and $t \to t'$}\\
\text{\textsc{fork}} &
\position{\Gamma}{X \uplus \{t\}}{e} &\longto& \position{\Gamma}{X \cup \{t' \mid t \to t'\}}{e} &
\text{when $\pol(t) = \spopol$}\\\\
\multicolumn{5}{c}{
\text{
\begin{tabular}{c}
\Large{Right $\exists$ moves}\\
(potentially available when $\pol(X) \subseteq \Sigma \cup \{\done\}$ and $\pol(e) \in \Sigma \cup \{\exists\}$)
\end{tabular}}
}\\\\
\text{\textsc{choose}} &
\position{\Gamma}{X}{e} &\longto& \position{\Gamma}{X}{e'} &
\text{when $e \to e'$ and $\pol(e) = \exists$} \\
\text{\textsc{junk}} &
\position{\Gamma}{X}{e} &\longto& \position{\Gamma}{X}{e'} &
\text{when 
$\pol(e) \in \Gamma$ and $e \to e'$}\\
\text{\textsc{fill}} &
\position{\Gamma}{X \uplus \{t\}}{e} &\longto& \position{\Gamma \cup \{a\}}{X \cup \{t'\}}{e'} &
\hspace{-0.5em}
\left\{
\hspace{-0.5em}
\begin{array}{r@{~}l}
\text{when} & \text{$\pol(e) = \pol(t) = \orapol{a}$,} \\
& \text{$e \to e'$ and  $t \to t'$}
\end{array}
\right.\\\\
\multicolumn{5}{c}{
\text{
\begin{tabular}{c}
\Large{Right $\forall$ moves}\\
(potentially available when $\pol(X) \subseteq \Sigma \cup \{\done\}$ and $\pol(e) = \forall$)
\end{tabular}}
}\\\\
\text{\textsc{alea}} &
\position{\Gamma}{X}{e} &\longto& \position{\Gamma}{X}{e'} &
\text{when $e \to e'$} \\\\
\end{array}
\]
\caption{Moves in the arena $\simGameAut$ (this is part of \Cref{def:simgame}).}
\label{fig:simgame-moves}
\end{figure}

Let us now try to get a sense of how the game is played. For this, it is
helpful to analyze separately left and right positions, and note that
the $\textsc{fill}$ move is the only move that allows one to get from a right
position to a left position. So the broad dynamic of the game is captured in
the picture in~\Cref{fig:gameOutline}.
\begin{figure}
\begin{center}
\begin{tikzpicture}[->,>={Stealth[round]},shorten >=1pt,
                    node distance=0cm and 5cm,semithick,
                    inner sep=2pt,bend angle=45]
  \tikzset{every state/.style={minimum size=20pt}}
  \tikzset{opp state/.style={draw,minimum size=40pt}}
  \tikzset{initial text={}}
  \node[opp state] (left) {left $\forall$};
  \node[opp state]     (rightspo) [above right=of left] {right $\forall$};
  \node[state]     (rightdup) [below right=of left] {right $\exists$};

  \path [every node/.style={fill=\backgroundcolor,circle}]
        (left) edge [bend left=30] (rightspo)
        (left) edge [bend left=10] (rightdup)
        (rightspo) edge [<->] %
        (rightdup)
        (rightdup) edge [bend left] node {\textsc{fill}} (left);
\end{tikzpicture}
\end{center}
\caption{Broad outline of the game.}
\label{fig:gameOutline}
\end{figure}
The basic idea is that, from any given left position $\position{\Gamma}{X}{e}$,
Spoiler is responsible for explaining what is the shape of an input to the
problem $\bigsqcap_{t \in X} \interp{\mfA^{t}}_\rho$ via left moves, at least insofar
as Duplicator should be able to observe it. Then we move to a right position,
where all the problem $\bigsqcap_{t \in X} \interp{\mfA^{t}}_\rho$ turns out to be
equivalent to the problem
\[X_\done \sqcap \bigsqcap_{\substack{t \in X \\ t' \mid t \to t'}} \interp{\mfA^{t'}}_\rho \star \rho(\pol(t)) \qquad \qquad
  \text{where $X_\done = \left\{
      \begin{array}{ll}
        1 & \text{if $1 \in X$} \\
        \top & \text{otherwise} \\
  \end{array}\right.$}
\]
whose inputs may not be observed more precisely by Duplicator. In such a case,
they are forced to start running the right-hand side $e$ (with Spoiler chipping
in when $e$ features demonic non-determinism). This ends when at some point Duplicator
plays a \textsc{fill} move, which is the only type of move which is able to modify all
components of the game at once. There the idea is that $e$ with $\pol(e) = a \in \Sigma$
gives the ability to make an oracle call to $\rho(a)$. If $\pol(t) = a$ for
some $t \in X$, then we can thus take the corresponding $\rho(a)$-input from there
and make progress on $t$. This may make $t$ transition to some $t'$ which has
a partly observable input, which may get us back into a left position.

The $\Gamma$ component records all the $a \in \Sigma$ such that Duplicator
should have seen an input to $\rho(a)$ before in the play. This should give
Duplicator the ability to ask a $\rho(a)$-question at any time, even if there
is no guarantee of getting a useful answer. This capability corresponds to
the $\textsc{junk}$ move.

Finally, to understand the winning condition, recall that $\dom(\interp{\mfA^e}_\rho)$
is composed of inputs that have certain ordinal ranks. In the course of a run of
a simulation, those ranks are meant to strictly decrease. Hence Spoiler should be
forbidden from unfolding their input infinitely many times, which explains
why left moves and \textsc{fill} are winning. On the other hand, if Spoiler
does not adopt a non-terminating behaviour, Duplicator should not run their simulation
forever and eventually get Spoiler stuck in an unfavourable position.
This essentially happens in three ways:
\begin{itemize}
\item if there is some $t \in X$ with $\pol(t) = \exists$ and no successors (for $\mfA = \reAut$, $0 \in X$),
\item if $\pol(e) = \forall$ and $e$ has no successors (for $\mfA = \reAut$, $e = \top$),
\item or if $e$ is final and $X$ contains a final state (in $\mfA = \reAut$, $e = 1 \in X$).
\end{itemize}

Now, let us turn to a couple of examples, where we take $\mfA = \reAut$.
First in~\Cref{fig:arena-ex}, we picture the arena that corresponds to the
reduction $b \star a \sqcup c \star a \le (b \sqcup c) \star a$, which we alluded
to in~\Cref{ex:intro-init}). It is possible for Duplicator to lose, but it is
trivial to check that Duplicator has a winning strategy (and Spoiler does not
meaningfully move).

\begin{figure}
\[
\begin{tikzpicture}[
sponode/.style={rectangle, draw=orange!60, fill=orange!5, very thick, minimum size=7mm},
dupnode/.style={rectangle, draw=violet!60, fill=violet!5,  very thick, minimum size=5mm},
nodes={%
   execute at begin node=$,%
   execute at end node=$%
  }%
]
\node (dummyabove) { };
\node[sponode]
   (initpos)
   [below = of dummyabove]
   {\position{\emptyset}{\{ b \star a \sqcup c \star a\}}{(b \sqcup c) \star a}};
\node[dupnode]
   (b1)
   [below left=of initpos]
   {\position{\emptyset}{\{b \star a\}}{(b \sqcup c) \star a}};
\node[sponode]
   (b2)
   [below=of b1]
   {\position{\{a\}}{\{ b \star 1\}}{(b \sqcup c) \star 1}};
\node[dupnode]
   (b3)
   [below=of b2]
   {\position{\{a\}}{\{ b\}}{(b \sqcup c) \star 1}};
\node[dupnode]
   (b4)
   [below=of b3]
   {\position{\{a\}}{\{ b\}}{b \sqcup c}};
\node[dupnode]
   (b5)
   [below=of b4]
   {\position{\{a\}}{\{ b\}}{b}};
\node[dupnode]
   (b5l)
   [right=of b5]
   {\position{\{a\}}{\{ b\}}{c}};
\node[sponode]
   (b6)
   [below=of b5]
   {\position{\{a, b\}}{\{ 1\}}{1}};
\node[dupnode]
   (c1)
   [below right=of initpos]
   {\position{\emptyset}{\{c \star a\}}{(b \sqcup c) \star a}};
\node[sponode]
   (c2)
   [below=of c1]
   {\position{\{a\}}{\{ c \star 1\}}{(b \sqcup c) \star 1}};
\node[dupnode]
   (c3)
   [below=of c2]
   {\position{\{a\}}{\{ c\}}{(b \sqcup c) \star 1}};
\node[dupnode]
   (c4)
   [below=of c3]
   {\position{\{a\}}{\{ c\}}{b \sqcup c}};
\node[dupnode]
   (c5)
   [below=of c4]
   {\position{\{a\}}{\{ c\}}{c}};
\node[dupnode]
   (c5l)
   [left=of c5]
   {\position{\{a\}}{\{ c\}}{b}};
\node[sponode]
   (c6)
   [below=of c5]
   {\position{\{a, c\}}{\{ 1\}}{1}};

\path[->] (dummyabove.south) edge (initpos.north)
          (b3) edge node[left] {\footnotesize \textsc{choose}} (b4)
          (c3) edge node[left] {\footnotesize \textsc{choose}} (c4)
          (b4) edge node[left] {\footnotesize \textsc{choose}} (b5)
          (c4) edge node[left] {\footnotesize \textsc{choose}} (c5)
          (b4) edge [bend left] node[above right] {\footnotesize \textsc{choose}} (b5l)
          (c4) edge [bend right] node[above left] {\footnotesize \textsc{choose}} (c5l)
;
\path[->, line width = 0.5mm] (initpos.west) edge [bend right] node[above left] {\footnotesize \textsc{explore}} (b1)
(initpos.east) edge [bend left] node[above right] {\footnotesize \textsc{explore}} (c1)
          (b1) edge node[left] {\footnotesize \textsc{fill}} (b2)
          (c1) edge node[left] {\footnotesize \textsc{fill}} (c2)
          (b2) edge node[left] {\footnotesize \textsc{explore}} (b3)
          (c2) edge node[left] {\footnotesize \textsc{explore}} (c3)
          (b5) edge node[left] {\footnotesize \textsc{fill}} (b6)
          (c5) edge node[left] {\footnotesize \textsc{fill}} (c6)
;
\end{tikzpicture}
\]
\caption{
  The arena of $\simGame(\position{\emptyset}{\{b \star a \sqcup c \star a\}}{(b \sqcup c) \star a})$.
The Spoiler positions are coloured in orange (and always have multiple or no outgoing edges) while the Duplicator positions are coloured violet. Winning
Duplicator transitions are denoted with bold arrows. Recall that if a player is
stuck, they lose.
}
\label{fig:arena-ex}
\end{figure}

\begin{exa}
Coming back to~\Cref{ex:introJunk}, i.e. $b \star a \le a \star b \star a$,
we can see there is a unique path in the arena of
$\simGame(\position{\emptyset}{b \star a}{a \star b \star a})$ that Duplicator
wins.
\[\begin{tikzcd}[column sep=huge]
	{\position{\emptyset}{\{b \star a\}}{a \star b \star a}} & {\position{\{a\}}{\{b \star 1\}}{a \star b \star 1}} & {\position{\{a\}}{\{b\}}{a \star b \star 1}} \\
	{\position{\{a,b\}}{\{1\}}{a}} & {\position{\{a,b\}}{\{1\}}{a \star 1}} & {\position{\{a\}}{\{b\}}{a \star b}} \\
	& {\position{\{a,b\}}{\{1\}}{1}}
	\arrow["{\textsc{fill}}", from=1-1, to=1-2]
	\arrow["{\textsc{explore}}", from=1-2, to=1-3]
	\arrow["{\textsc{choose}}", from=1-3, to=2-3]
	\arrow["{\textsc{junk}}"', from=2-1, to=3-2]
	\arrow["{\textsc{choose}}"', from=2-2, to=2-1]
	\arrow["{\textsc{fill}}"', from=2-3, to=2-2]
\end{tikzcd}\]
On the other hand, $\simGame(\position{\emptyset}{a \star b}{a \star b \star a})$
contains a single $\exists$ position with no available moves, witnessing that
Duplicator has no generic way of producing an input for the problem denoted by $a$.
\end{exa}

Finally, in~\Cref{fig:ex-strat}, we give a picture of a non-trivial
Duplicator positional winning strategy that corresponds to the reduction
discussed in~\Cref{ex:derivation}. To make things more bearable, we omit the
$\Gamma$ component from positions. More general, we need to discuss systematically
the variant of $\simGameAut$, where we fix $\Gamma = \Sigma$ rather than $\emptyset$.
This essentially corresponds to unconditionally allowing the $\textsc{junk}$, and we shall see
that this captures universal validity for pointed partial degrees. Let us formally
thus define this game.

\begin{defi}
\label{def:psim-game}
Define the arena $\psimGameAut$
as follows:
\begin{itemize}
    \item take positions $V_{\psimGameAut}$ to be pairs $(X, e)$ with $X$ a finite subset of $\RegE{\Sigma}$
and $e \in \RegE{\Sigma}$; write them $\pposition{X}{e}$ in the sequel
    \item take the polarity of a position $\pposition{X}{e}$
to be the same as the polarity of $\position{\Sigma}{X}{e}$ in $\simGameAut$
    \item generate the set of moves $\pposition{X}{e} \longrightarrow \pposition{X'}{e'}$
    from all moves $\position{\Sigma}{X}{e} \longrightarrow \position{\Sigma}{X'}{e'}$
    \item say a position $\pposition{X}{e}$ is winning if and only if $\position{\Sigma}{X}{e}$ is
    winning in $\simGameAut$.
\end{itemize}
\end{defi}

\begin{figure}
\[
\begin{tikzpicture}[
rightsponode/.style={rectangle, draw=orange!60, fill=orange!5, very thick, minimum size=7mm},
sponode/.style={rectangle, draw=orange!60, fill=orange!5, very thick, minimum size=7mm},
dupnode/.style={rectangle, draw=violet!60, fill=violet!5,  very thick, minimum size=5mm},
nodes={%
   execute at begin node=$,%
   execute at end node=$%
  }%
]
\clip (-5,1) rectangle (15, -16);
\node (dummyabove) {};
\node[sponode]
   (initposa)
   [below=of dummyabove]
   {\pposition{\{a^\diamond \sqcap b^\diamond\}}{(a \sqcap b)^\diamond}};
\node[sponode]
   (initpos)
   [below=of initposa]
   {\pposition{\{a^\diamond, b^\diamond\}}{(a \sqcap b)^\diamond}};
\node[dupnode]
   (ezwin)
   [right=of initpos]
   {\pposition{\{1, 1\}}{(a \sqcap b)^\diamond}};
\node (dummyezwin) [below=of ezwin] {};
\node[sponode]
   (ezwon)
   [right=of ezwin]
   {\pposition{\{1, 1\}}{1}};
\node[dupnode]
   (main)
   [below=of dummyezwin]
   {\pposition{\{a^\diamond \star a, b^\diamond \star b\}}{(a \sqcap b)^\diamond}};
\node[rightsponode]
   (mainspo)
   [below=of main]
   {\pposition{\{a^\diamond \star a, b^\diamond \star b\}}{(a \sqcap b)^\diamond \star (a \sqcap b)}};
\node
  (dummy)
  [below= of mainspo] {};
\node[dupnode]
   (dupfilla)
   [left= of dummy]
   {\pposition{\{a^\diamond \star a, b^\diamond \star b\}}{(a \sqcap b)^\diamond \star a}};
\node[sponode]
   (spofilla)
   [below =of dupfilla]
   {\pposition{\{a^\diamond, b^\diamond \star b\}}{(a \sqcap b)^\diamond}};
\node[dupnode]
   (wina)
   [below =of spofilla]
   {\pposition{\{1, b^\diamond \star b\}}{(a \sqcap b)^\diamond}};
\node[sponode]
   (wona)
   [below =of wina]
   {\pposition{\{1, b^\diamond \star b\}}{1}};
\node[dupnode]
   (dupfillb)
   [right= of dummy]
   {\pposition{\{a^\diamond \star a, b^\diamond \star b\}}{(a \sqcap b)^\diamond \star b}};
\node[sponode]
   (spofillb)
   [below =of dupfillb]
   {\pposition{\{a^\diamond \star a, b^\diamond \}}{(a \sqcap b)^\diamond}};
\node[dupnode]
   (winb)
   [below =of spofillb]
   {\pposition{\{a^\diamond \star a, 1\}}{(a \sqcap b)^\diamond}};
\node[sponode]
   (wonb)
   [below =of winb]
   {\pposition{\{a^\diamond \star a, 1\}}{1}};
\draw[->] (dummyabove.south) to (initposa.north);
\draw[->, line width = 0.5mm] (initposa.south) -- node[left] {\footnotesize \textsc{fork}} (initpos.north);
\draw[->, line width = 0.5mm] (initpos.south) to[out=270,in=90] node[above] {\footnotesize {\textsc{explore}}} (main.north);
\draw[->, line width = 0.5mm] (initpos.east) to (ezwin.west);
\draw[->] (ezwin.east) to (ezwon.west);
\draw[->] (main.south) to node[left] {\footnotesize {\textsc{choose}}} (mainspo.north);
\draw[->] (mainspo.west) to[bend right] node[left] {\footnotesize {\textsc{alea}}} (dupfilla.north);
\draw[->, line width = 0.5mm] (dupfilla.south) to node[left] {\footnotesize {\textsc{fill}}} (spofilla.north);
\draw[->, line width = 0.5mm] (spofilla.south) to node[left] {\footnotesize {\textsc{explore}}} (wina.north);
\draw[->] (wina.south) to node[left] {\footnotesize {\textsc{choose}}} (wona.north);
\draw[->, line width = 0.5mm] (spofilla.west) to[out=180,in=180, looseness=1.5] node[left] {\footnotesize {\textsc{explore}}} (main.west);
\draw[->] (mainspo.east) to[bend left] node[right] {\footnotesize {\textsc{alea}}} (dupfillb.north);
\draw[->, line width = 0.5mm] (dupfillb.south) to node[right] {\footnotesize {\textsc{fill}}} (spofillb.north);
\draw[->, line width = 0.5mm] (spofillb.south) to node[right] {\footnotesize {\textsc{explore}}} (winb.north);
\draw[->] (winb.south) to node[right] {\footnotesize {\textsc{choose}}} (wonb.north);
\draw[->, line width = 0.5mm] (spofillb.east) to[out=0,in=0, looseness=1.5] node[right] {\footnotesize {\textsc{explore}}} (main.east);
\draw[->, line width = 0.5mm] (initpos.west) to[out=180,in=180] (wina.west);
\draw[->, line width = 0.5mm] (initpos.north east) to[out=45,in=0, looseness=1.75] (winb.east);
\end{tikzpicture}
\]
\caption{
A subgraph of $\psimGame$ corresponding to a positional  Duplicator winning strategy
in the game witnessing the valid inequality $a^\diamond \sqcap b^\diamond \le (a \sqcap b)^\diamond$ in the pointed Weihrauch degrees.
For legibility, we squashed some \textsc{explore} moves together
on the second position and omitted labelling some move names from some edges.
The Spoiler positions are coloured in orange (and always have multiple or no outgoing edges) while the Duplicator positions are coloured violet. Winning
Duplicator transitions are denoted with bold arrows.
}
\label{fig:ex-strat}
\end{figure}

\begin{rem}
  The game $\simGame$ can also be regarded, from a proof-theoretic
perspective, as a sort of infinitary proof system, which have been used in the
literature to study Kleene algebras~\cite{lhcka22} and regular languages~\cite{dasdeomega25}. The intuition there
is that Duplicator strategies can be regarded as (possibly infinitary)
\emph{preproofs} and
that the winning condition tells us which ones are genuine \emph{proofs}.
Positional winning strategies can be regarded as \emph{cyclic proofs}.
The main difference with this tradition is that here we use sets $\Gamma \subseteq \Sigma$
and $X \subseteq Q$ instead of lists of elements of $Q$ to maintain a finite arena
(at the cost of losing the ability to develop a principled theory of proof
composition, which we don't care about for this paper).
\end{rem}

\subsection{The game characterizes Weihrauch reducibility}

Now we show that universal validity is characterized by the simulation game $\simGameAut$
in the following sense.

\begin{thm}
  \label{thm:gameAutEquiv}
Let $\mfA = (Q, \to, \pol)$ be a preautomaton over $\Sigma$, $\Gamma \subseteq \Sigma$
and $t, e \in Q$.
Then Duplicator has a winning strategy in $\simGameAut(\position{\emptyset}{\{t\}}{e})$ if and only if,
for every valuation $\rho$ in the partial Weihrauch degrees,
\[%
\interp{\mfA^{(t)}}_\rho \quad \leqW \quad \interp{\mfA^{(e)}}_\rho\]
\end{thm}

From this will follow the following corollary that will go towards proving~\Cref{thm:mainloop}.

\begin{restatable}{cor}{maingame}
\label{cor:maingame}
For any $e,f \in \RegE{\Sigma}$, $e \le f$ is valid in the 
partial Weihrauch degrees if and only if Duplicator has a winning strategy in $\simGame(\position{\emptyset}{\{e\}}{f})$.
\end{restatable}

To prove~\Cref{thm:gameAutEquiv}, we will first use determinacy of B\"uchi games and then
analyze the situation according to who wins. If Duplicator wins, we can compute
a reduction witness that works universally. If Spoiler wins, then we exhibit that the
reduction cannot hold for a fixed generic interpretation of the letters.

Let us start with the easiest part, fixing the preautomaton $\mfA = (Q, \to, \pol)$ for the
rest of this subsection.

\begin{lem}
\label{lem:dup-strat-to-red}
If Duplicator has a winning strategy in $\simGameAut(\position{\Gamma}{X}{e})$,
then for every $\rho$, there is a Weihrauch reduction
\[\bigsqcap_{b \in \Gamma} (\top \star \rho(b)) ~ \sqcap ~ \bigsqcap_{t \in X} \interp{\mfA^{(t)}}_\rho \quad \leqW \quad \interp{\mfA^{(e)}}_\rho\]
\end{lem}
\begin{rem}
The problem $\bigsqcap_{b \in \Gamma} (\top \star \rho(b))$ is equivalent to
the problem which take as input a map $p \in (\Baire)^\Gamma$ such that
$p(b) \in \dom(\rho(b))$, and never answers. We will implicitly
use that in the proof.
\end{rem}
\begin{proof}
Fix a positional\footnote{It is not strictly necessary that the strategy
be positional to carry out this argument, but it makes the bureaucracy more bearable.} winning strategy for Duplicator.
We define by mutual recursion a family of maps $(\varphi_{\Gamma,X,f})_{\Gamma \subseteq \Sigma, X \subseteq Q, f \in Q}$ taking
as inputs:
\begin{itemize}
  \item $p \in (\Baire)^\Gamma$ that serves as input to $\bigsqcap_{b \in \Gamma} (\top \star \rho(b))$,
\item $q \in (\Baire)^X$ such that $q_t \in \dom(\interpaut{t})$ for all $t \in X$,
\end{itemize}
and outputs $\varphi_{\Gamma,X,e}(p,q) \in \dom(\interpaut{e}_\rho)$. The map
will be guided by the strategy as follows:
\begin{itemize}
  \item If we are in a left position, we can stay in the strategy by whichever
    move Spoiler wants to make. For this to happen, we need to have $X = X' \uplus \{t\}$
    for some $t$ with $\pol(t) \in \{\forall, \exists\}$, which we can pick in
    a canonical way. Let us treat the two cases:
    \begin{itemize}
      \item if $\pol(t) = \forall$, then $q_t = \tuple{c_1, \ldots, c_n}$.
        Set $X'' = X' \cup \{t_1, \ldots, t_n\}$, so that we know that
        $\position{\Gamma}{X''}{e}$ is still in the strategy following a $\textsc{fork}$
        move. Then we can set $\varphi_{\Gamma,X,e}(p,q) = \varphi_{\Gamma,X'',e}(p,q')$
        with $q'$ is the map such that $q'(t_i) = c_i$ and $q'(t) = q(t)$
        for $t \in X' \setminus \{t_1, \ldots, t_n\}$.
      \item if $\pol(t) = \exists$, then $q_t = \tuple{t', c}$ for some $t \to t'$
        and $c \in \dom(\interp{\mfA^{(t')}})$. Then we can set $\varphi_{\Gamma, X, e}(p,q) = \varphi_{\Gamma, X'' \cup \{t'\}, e}(p,q')$
        with $q'(t') = c$ and $q'(t) = q(t)$ otherwise.
    \end{itemize}
  \item If we are in a left $\forall$ position, we are in one of two cases:
    \begin{itemize}
      \item if $\pol(e) = \done$, then we have simply $\varphi_{\Gamma, X, e}(p,q) = \tuple{ }$ (here the winning condition will ensure
        the backward part of the reduction is equally trivial to define).
      \item if $\pol(e) = \forall$, let $e'_1, \ldots, e'_n$ the successors of $e$ in order.
      We then have
    \[\varphi_{\Gamma, X, e}(p,q) = \tuple{\varphi_{\Gamma, X, e'_1}(p,q), \ldots, \varphi_{\Gamma, X, e'_n}(p,q)}\]
    \end{itemize}
  \item If we are in a right $\exists$ position, we have three subcases according
    to Duplicator's move in the strategy:
    \begin{itemize}
      \item if Duplicator plays a $\textsc{choose}$ move
        \[\position{\Gamma}{X}{e} \quad \xrightarrow{\textsc{choose}} \quad \position{\Gamma}{X}{e'}\]
        then we set $\varphi_{\Gamma, X,e}(p,q) = \langle \code{e'}, \varphi_{\Gamma, X, e'}(p,q)\rangle.$
      \item if Duplicator plays a $\textsc{fill}$ move, with $\pol(t) = \pol(e) = a$, $t \to t'$ and $e \to e'$:
        \[\position{\Gamma}{X \uplus \{t\}}{e}
          \quad \xrightarrow{\textsc{fill}} \quad 
        \position{\Gamma \cup a}{X \cup \{t'\}}{e'}\]
        Then we set $\varphi_{\Gamma, X, e}(p, q) = \tuple{ u, \lambda x. \varphi_{\Gamma \cup \{a\}, X' \cup \{t'\}, e'}(p', q_x)}$
        where $q(t') = \tuple{u , f}$:
        \[
          p'(b) = \left\{ \begin{array}{ll}
              u &\text{for $b = a$}\\
              p(b) & \text{otherwise}
            \end{array} \right.
          \qquad \text{and} \qquad
          q_x(v) = \left\{ \begin{array}{ll}
              f \cdot x &\text{for $v = t'$}\\
              q(v) & \text{otherwise}
            \end{array} \right.
        \]
      \item if Duplicator plays a $\textsc{junk}$ move where $\pol(e) = a \in \Gamma$ and $e \to e'$
        \[\position{\Gamma}{X}{e}
          \quad \xrightarrow{\textsc{fill}} \quad 
        \position{\Gamma \cup a}{X}{e'}\]
        then we set $\varphi_{\Gamma, X, e}(p, q) = \tuple{ u, \lambda x. \varphi_{\Gamma, X' \cup \{t'\}, e'}(p, q)}$
        (here note that the $x$ in the $\lambda$ simply is not used).
     \end{itemize}
\end{itemize}

That this map is well-defined (i.e., that $p_b \in \dom(\rho(b))$ and $q_e \in \dom(\interpaut{e}_\rho)$ imply that $\varphi_{\Gamma,X,f}(p,q)$ is defined and in $\dom(\interpaut{f}_\rho)$) can be proven
by strong lexicographic induction on the pair $(N, \rank(q))$, where $N \in \Nat$ is the largest number of moves played in the Duplicator strategy before a
winning move is to be played (a \textsc{fill} move or a left move) or the game is ended, taking $\rank(q) = \bigoplus_{t \in X} \rank(q(t))$. That this pair decreases strictly along each
move is ensured by the fact that the strategy is winning.

We leave the similar definition of the backward parts of those reductions
$\psi_{\Gamma, X, e}$ as a formal exercise.
\end{proof}

Now we will explain how to prove a converse. By determinacy of B\"uchi game,
we'll see it suffices to prove the following.

\begin{lem}
\label{lem:spo-strat-to-nonred}
There is a valuation $G$ such that, if Spoiler has a winning strategy in $\simGameAut(\position{\emptyset}{\{t\}}{e})$,
then there is no Weihrauch reduction
\[\interp{\mfA^{(t)}}_G \quad \leqW \quad \interp{\mfA^{(e)}}_G\]
\end{lem}

The informal idea used to define $G(a)$ is that inputs should
be tokens that allow to make a specific $a$-labelled $t \to t'$
transition in $\mathfrak{A}$. The token also encodes an oracle budget $n \in \Nat$.
The answer to such a token will be a list of tokens holding budget $n - 1$,
corresponding to a strategy to run $\mfA^{(t)}$ until further oracle states
or final states are hit.
If a final state is hit, then a special completion token is emitted in addition
to the state tokens.

Then the separation will be obtained as follows: the problem
``given a token for a non-trivial state $t$ with budget $n$, produce a completion token'' will
typically be reducible to
$\interpaut{t}_G$ for high enough budgets, but any
$e$ such that Spoiler wins in $\simGameAut(\position{\emptyset}{\{t\}}{e})$, $\interpaut{e}_G$
won't be able to solve that problem for arbitrarily large budgets.

\begin{rem}
Here we might note that the proof will go through by manipulating finite budgets,
and that all bounds are computable in $\mfA$, the states involved and the Spoiler
strategy. These observation can be made more formal as follows: if Spoiler wins
$\simGameAut(\position{\emptyset}{\{t\}}{e})$, then there exists a \emph{finite loop-free}
automaton $\mfB^{(t')}$ such that $\interp{\mfB^{(t')}}_G \not\leqW \interpaut{e}_G$
and there is an automaton morphism $\mfB^{(t')} \to \mfA^{(t)}$ which a fortiori
witnesses that $\interp{\mfB^{(t')}}_\rho \leqW \interp{\mfA^{(t)}}_\rho$ for all $\rho$.
Similarly, when looking at two expressions $t, e \in \RegE{\Sigma}$ with $t \not\le e$,
one can find a $\diamond$-free syntactic under-approximation $t'$ of $t$ such that $\interp{t'}_G \not\leqW \interp{e}_G$.
However proving that is not necessary for our goal; we leave
the formalization of the notions of morphisms, $\diamond$-free syntactic under-approximation and the
proof of these results as an exercise.
\end{rem}

To make all of this work, we need to encode the tokens in such a way that
a reduction cannot cheat by, say, forging tokens. To do this, tokens are formally
defined as sufficiently incomparable Turing degrees. Given a lattice $L$, recall that a
\emph{strong antichain} is an antichain $A \subseteq L$ such that, for any
finite sets $X, Y \subseteq A$, we have that $\bigvee X \le \bigvee Y$ implies
$X \subseteq Y$. Countable strong antichains exist in the
Turing degree\footnote{For instance, the aforementioned existence of continuum-sized
antichains is enough to imply this.}, and that will be the starting point for
our definitions. Let us now proceed.

\begin{defi}
\label{def:genericinterp}
Call $O \subseteq Q$ the set of states $t \in Q$ such that $\pol(t) \in \Sigma$.
Fix a map $\token : \{\done\} \uplus \Nat \times O \to \Cantor$ whose range $\Tokens \subseteq \Cantor$
maps to a strong antichain in the Turing degrees.
Call $\top_\done$ the problem with $\dom(\top_\done) = \{\token(\done)\}$
and $\top_\done(\token(\done)) = \emptyset$.

We define auxiliary interpretations by recursion on $n$,
starting with $G_0(a) = G^{\top}_0(a) = 0$ and the following clauses, so that $G = G_{\ge 0}$:
\[
\begin{array}{rcl}
  \dom(G_{n+1}^\top(a)) = \dom(G_{n+1}(a)) &=& \{\token(n,t) \mid \pol(t) = a, t \to t'\}
\\
G_{n+1}^\top(a)(\token(n, t)) &=&  \emptyset \\
G_{n+1}(a)(\token(n, t)) &=&  \dom(\top_\done \star \interpaut{t'}_{G_n^\top}) ~~~ \text{when $t \to t'$}
\\
G_{\ge n}(a) &=& \bigcup\limits_{k \ge n} G_k(a)
\end{array}\]
\end{defi}
\begin{rem}
We trivially have $G^\top_n(a) \equivW \top \star G_n(a)$; we adopt the definition
above rather than $\top \star G_n$ to make notations later on less heavy.
\end{rem}

To understand the whole definition, it is helpful to focus on what is
the set $\ExeStrat_n(t) = \dom(\interpaut{t}_{\top \star G_n})$, elements of which we
call \emph{execution strategies} with budget $n$. Its definition only depends on
the states of $\mfA$ that are accessible from $t$ without going through an
oracle state. Execution strategies for $t$ describe how run through $\mfA$
as $\exists$ until oracle states or terminating states are reached. They need
to provide inputs of weight $n-1$ to all oracle states they may reach.
We call elements of $\ExeStrat_n^\done(t) = \dom(\top_\done \star \interp{t}_{\top \star G_n})$
\emph{execution strategies with (budget $n$ and) payoffs (for $t$)}. Up to encoding details, they can be
thought of as tuples $\tuple{s, p}$
where $s \in \dom(\interp{t}_{\top\star G_n})$
is an execution strategy for $t$ and $p = \token(\done)$ if $s$ reaches a final state.
Let us call $s \mapsto \underline{s}$ the map $\ExeStrat_n(t) \to \ExeStrat_n^\done(t)$
which augments $s$ with payoffs by putting $p = \tuple{ }$ if $s$ does not reach a final state.

For any state $t$ an $n \in \Nat$ with $\token(n,t) \in \dom(G(a))$ and $t \to t'$, there
is a canonical input $\tuple{\token(n,t), \execute \cdot \code{t'}}$ of
$\interpaut{t}_{G_n}$ where $\execute$ is the code of a function which
takes as inputs a state $t$ and an execution strategy with payoffs for $t$.
It essentially follows along any execution strategy provided by oracle calls.
Concretely, it can be presented by the following recursion (assuming the
$\tuple{t'_1, \ldots, t'_m}$ is the canonically ordered list of successors of $t$ in the case
where $\pol(t) = \forall$ case and $t \to t'$ in the second case):
\[
\begin{array}{lcll}
\execute \cdot \code{t} \cdot z &=& z & \text{if $\pol(t) = \done$} \\
\execute \cdot \code{t} \cdot \tuple{\tuple{u, \_}, \_} &=&
\tuple{u, \execute \cdot \code{t'}} & \text{if $\pol(t) \in \Sigma$}
\\
\execute \cdot \code{t} \cdot \tuple{\tuple{\code{t'}, c}, k} &=& \execute \cdot \code{t'} \cdot \tuple{c, k} & \text{if $\pol(t) = \exists$} \\
\execute \cdot \code{t} \cdot \tuple{\tuple{c_1, \ldots, c_m}, k}
&=&
\left\langle \begin{array}{c@{}l}
\tuple{\execute \cdot \code{t'_1} \cdot c_1, \lambda x. k \cdot \inc_1(x)} &,\\
\ldots &, \\
\tuple{\execute \cdot \code{t'_m} \cdot c_m, \lambda x. k \cdot \inc_m(x)}
\end{array}\right\rangle
& \text{if $\pol(t) = \forall$}
\end{array}
\]

Now let us start analyzing Spoiler winning strategies.
By inspection of
the winning condition, it must be the case that in any play, the number of
performed $\textsc{fill}$ moves is finite.
Since there are only finitely many moves available at each step, we can define
the \emph{price} (or \textsc{fill}-depth) of a position as the maximal number of \textsc{fill} moves that can be performed
by the strategy.

The other thing that Spoiler winning ensures is that there are only finitely
many visits through left positions. Left moves are precisely the information
that we need to build the input part of our counter-witnesses to potential reductions. Those inputs will always be tupling of the shape
$\execute \cdot \code{t} \cdot s_t$ for appropriate $s_t \in \ExeStrat_n(t)$.
Let us define them.

\begin{defi}
Given a positional Spoiler winning strategy, a position $P = \position{\Gamma}{X}{e}$, $t \in X$ and
a budget $n > 0$, define by recursion the execution strategy $s^P_{n,t} \in \ExeStrat_n(t)$ as follows.
\begin{itemize}
\item If $P$ is a right position, then $\pol(t) \in \Sigma$ and $s^P_{n,t} = \token(n,t)$.
\item If $P$ is a left position, Spoiler chooses to make it evolve into a position $P'$. If $t \in P'$,
we take $s^P_{n,t} = s^P_{n,t'}$. Otherwise, it depends on whether a $\textsc{fork}$ or $\textsc{explore}$
move was played (which is determined by $\pol(t)$):
\begin{itemize}
\item If it was a $\textsc{fork}$ move, then letting $t_1', \ldots, t_m'$ be the
successors of $t$, we set $s^P_{n,t} = \tuple{s^{P'}_{n,t_1'}, \ldots, s^{P'}_{n,t_m'}}$.
\item If it was an $\textsc{explore}$ move $\position{\Gamma}{X \uplus \{t\}}{e} \xrightarrow{\textsc{explore}} \position{\Gamma}{X \uplus \{t'\}}{e}$,
  we set $s^P_{n,t} = \tuple{ \code{t'}, s^P_{n, t'}}$.
\end{itemize}
\end{itemize}
\end{defi}

\begin{defi}
Let $P = \position{\Gamma}{X}{e}$ be a position in a positional Spoiler winning
strategy
A $P$-history with budget $n$ is a finite set $H$ that contains
\begin{itemize}
\item some $\underline{s^P_{k,t}}$ for $k \ge n$ and $t \in X$,
\item some tokens $\token(k, t)$ for $k > n$ and arbitrary $t \in S$, subject to
the following restrictions:
\begin{itemize}
\item $\pol(t) \in \Gamma$
\item there should exist an element of $G(\pol(t))(\token(k,t))$ which is computable
from $H$.
\end{itemize}
\end{itemize}
Furthermore a $P$-history must contain:
\begin{itemize}
\item at least one $\underline{s^P_{k,t}}$ for every $t \in X$
\item at least one $\token(k,t)$ with $\pol(t) = b$ for every $b \in \Gamma$.
\end{itemize}
\end{defi}

The following probably illustrates best the use of histories in the argument to come.

\begin{lem}
  \label{lem:negProofFillMove}
Assume a positional Spoiler winning strategy in $\simGameAut(P)$,
where $P = \position{\Gamma}{X}{e}$ is a right position with $\pol(e) = a \in \Sigma$.
Assume further a $P$-history $H$ with budget $n$ greater than the price of the Spoiler
strategy. Then:
\begin{itemize}
\item $H$ is a $P'$-history with budget $n$ for any move $P \xrightarrow{\textsc{junk}} P'$.
\item for any move $P = \position{\Gamma}{X' \uplus \{t\}}{e} \xrightarrow{\textsc{fill}} \position{\Gamma \cup \{a\}}{X' \cup \{t'\}}{e'} = P'$,
with $\underline{s^P_{k,t}} \in H$, then there is a unique  $\underline{s^{P'}_{k-1,t'}} \in G(a)(\token(k, t))$ such that
\[H' ~~=~~
  (H 
\setminus \{\underline{s^P_{k,t}}\})
\cup \{\token(k,t), \underline{s^{P'}_{k-1,t'}}\} 
\]
is a $P'$-history with budget $n-1$.
\end{itemize}
Furthermore, the price of the Spoiler strategy stays the same when going from $P$ to $P'$ in the first
case, while it strictly drops in the second.
\end{lem}
\begin{proof}
The first point is obvious, as following a $\textsc{junk}$ move ensures
$P' = \position{\Gamma}{X}{e'}$ for $e \to e'$ and the definition of histories
does not depend on the right-hand side term of a position.

To prove the second point, first observe that $\underline{s^{P}_{k,t}} = \tuple{\token(k,t), \tuple{ }}$,
and thus that is is computably equivalent to $\token(k,t)$.
The rest is then a matter of
unfolding the definitions, and checking that $k-1 \ge 0$ thanks to the notion of
price of a strategy.
\end{proof}

We are now ready to embark on our main inductive argument.

\begin{lem}
\label{lem:main-spo-technical}
Suppose that $P = \position{\Gamma}{X}{e}$ is a position in a Spoiler winning
strategy. Assume that $H$ is a $P$-history with budget $n$ exceeding the price
of Spoiler's strategy.

Then for every $H$-computable input $u \in \dom(\interpaut{e}_G)$,
there is $x \in \interpaut{e}_G(u)$ such that $H \cup \{x\}$ cannot compute $\token(\done)$.
\end{lem}

\begin{rem}
Let us stress there are two (realizable) scenarios where the conclusion can hold:
either there is no $u \in \dom(\interpaut{e}_G)$ whatsoever which is computable
in $H$, or any such $u$ simply cannot ask the right questions to problems in $G$.
The first will necessarily be the case for instance when $1 \in X$,
while the second will happen when $e = 1$, $X \neq \emptyset$ and $\min \{\price(t) \mid t \in X\} > 0$.
\end{rem}

\begin{proof}
The overall proof proceeds by induction on the rank of $u$ and the number of moves
required to get to a right position by the Spoiler strategy. Within that
induction, we perform a case analysis on what the position $P$ is.
\begin{itemize}
\item If $P$ is a left position and the Spoiler strategy performs the move
$P \to P'$, it suffices to compute a $P'$-history $H'$ from the given $P$-history $H$
to conclude by appealing to the inductive hypothesis. This is relatively straightforward.
For instance, if
\[P = \position{\Gamma}{X' \uplus \{t\}}{e} \qquad \xrightarrow{\textsc{explore}}\qquad \position{\Gamma}{X' \cup \{t'\}}{e}\]
with $t \to t'$, we simply replace any occurrence of $\underline{s_{k,t}^P}$ in $H$ with $\underline{s_{k,t'}^{P'}}$ in $H'$
\[ H' = (H \setminus \{\underline{s_{k,t}^P} \mid k \in \Nat\}) \cup \{\underline{s_{k, t'}^{P'}} \mid k \in \Nat \wedge \underline{s_{k,t}^P} \in H\}\]
$H$ can easily seen to be computable from $H'$, hence $u$ is computable from $H'$ as long as it is computable from $H$.
$\textsc{fork}$ moves are handled in a very similar manner.
\item If $P$ is a right position, we perform a case analysis according to $\pol(e)$:
  \begin{itemize}
    \item If $\pol(e) = \forall$, then let $e_1', \ldots, e'_m$ be the successors of $e$ in $\mfA$ in the canonical order.
      We then have $u = \tuple{u_1, \ldots, u_m}$ and the Spoiler strategy includes a move
      \[P = \position{\Gamma}{X}{e} \qquad \xrightarrow{\textsc{alea}} \qquad \position{\Gamma}{X}{e'_i} = P' \qquad \text{for some $i \in \{1, \ldots, m\}$}\]
      $H$ is also a $P'$-history, so by the induction hypothesis, there is some $x \in \interpaut{e'_i}_G(u_i)$ such that $H \cup \{x\}$ cannot
      compute an answer to $H$. Then, a fortiori, $\inc_i(x) \in \interpaut{e}_G(u)$ and $H \cup \{\inc_i(x)\}$ cannot compute $\token(\done)$.
    \item If $\pol(e) = \exists$, then $u = \tuple{\code{e'}, u'}$ and $\position{\Gamma}{X}{e'}$ is a position in the Spoiler strategy.
      Then we can directly apply the induction hypothesis and get a suitable $x \in \interpaut{e'}_G(u') = \interpaut{e}_G(u)$.
    \item If $\pol(e) = a \in \Sigma$, then we are in the specific situation outlined in~\Cref{lem:negProofFillMove} and we must
      have $u = \tuple{\token(k, t), f}$ for some $k$ and $t$. Since $\token(k,t)$ needs to be computable from $H$, then it must be the case
      that we have either $\token(k, t) \in H$ or $\underline{s_{k,t}^P} = \tuple{\token(k,t), \tuple{ }} \in H$ because $\Tokens$ is
      a strong antichain in the Turing degrees and $H$ contains only elements of the shape $\token(k',t')$ or $\tuple{\token(k',t'), \tuple{ }}$.
      We proceed by analyzing the two cases:
      \begin{itemize}
        \item If we have $\token(k,t) \in H$ and $t \to t'$, then by definition of histories, we have
          that $\position{\Gamma}{X}{t'}$ belongs to the Spoiler strategy (following a $\textsc{junk}$ move) and
          $H$ computes some $x \in G(a)(\token(k,t))$. So we actually have that $f \cdot x$ is computable from $H$,
          and we can apply the induction hypothesis to conclude that an answer $x'$ to $f \cdot x$ cannot compute $\token(\done)$,
          and neither can the answer $\tuple{x, x'}$ to $\tuple{\token(k,t), f}$.
        \item If we have $\underline{s_{k,t}^P} = \tuple{\token(k,t),\tuple{ }}$, it means we have a move
          \[P = \position{\Gamma}{X' \uplus \{t\}}{e} \qquad \xrightarrow{\textsc{fill}} \qquad \position{\Gamma \cup \{a\}}{X' \cup \{t'\}}{e'} = P'\]
          Taking $H'$ as in~\Cref{lem:negProofFillMove}, then we obviously have
          $f \cdot \underline{s^{P'}_{k-1,t'}}$ computable in $H'$, so we can apply the inductive hypothesis
          and obtain some $x \in \interpaut{e'}_G(f \cdot \underline{s^{P'}_{k-1,t'}})$ such that $H' \cup \{x\}$ cannot compute an $H'$-answer.
          Then note that $\tuple{\underline{s^{P'}_{k-1, t'}}, x} \in \interpaut{e}_G(u)$ and $H$ are computable from $H' \cup \{x\}$.
          So in particular, it means that $H \cup \{\tuple{\underline{s^{P'}_{k-1,t'}}, x}\}$ cannot compute $\token(\done)$,
      \end{itemize}
    \item Finally, if $\pol(e) = \done$, then we are in a final position
          and $u = \tuple{ }$. Since the position is winning for Spoiler,
          we have that $\token(\done)$ is not computable from $H$, nor from
          $H \cup \{\tuple{ }\}$.
  \end{itemize}
\end{itemize}
\end{proof}

We are now in a position where we can finally conclude for our converse statement,
and then put our equivalence together.

\begin{proof}[Proof of~\Cref{lem:spo-strat-to-nonred}]
Assume $t, e \in Q$, call $P$ the position $\position{\emptyset}{\{t\}}{e}$
and fix a positional Spoiler winning strategy in $\simGameAut(P)$
Let us assume that for now, that $\pol(t) \in \Sigma$.
Let $n$ be the price of the strategy and consider $u = \tuple{\token(n, t), \tuple{ }}$.
Then $\{u\}$ is a $P$-history, and we have that any
$x \in \interpaut{t}_G(\execute \cdot \code{t} \cdot u)$
computes $\token(\done)$. So, if we had $\interpaut{t}_G \leqW \interpaut{e}_G$,
$u$ should compute a question $u' \in \dom(\interpaut{e}_G)$ whose
answers all compute $\token(\done)$, which contradicts~\Cref{lem:main-spo-technical}.

The general case where $\pol(t) \not\in \Sigma$ can be carried out with the
following hack: consider some $\bullet \not\in \Sigma$, $t_{-1}, e_{-1} \not\in Q$, and extend the preautomaton
$\mfA$ into the preautomaton $\mfA' = (Q \cup \{t_{-1}, e_{-1}\}, \to \cup \{(t_{-1}, t), (e_{-1}, e)\}, \pol')$
over $\Sigma \cup \{\bullet\}$
with $\pol'(q) = \pol(q)$ for $q \in Q$ and $\pol'(t_{-1}) = \pol'(e_{-1}) = \bullet$.
Then a positional Spoiler winning strategy for $\simGameAut(\position{\emptyset}{\{t\}}{e})$
induces one for $\simGamePlain^{{\mfA'}}(\position{\emptyset}{\{t_{-1}\}}{\{e_{-1}\}})$, from which we get
that there is no reduction $\interp{{\mfA'}^{(t_{-1})}}_{G_\bullet} \leqW \interp{{\mfA'}^{(e_{-1})}}_{G_\bullet}$ by
replaying our argument for $\mfA'$. By noting that $G \subseteq G_\bullet$, we can
deduce that there is a fortiori no reduction $\interpaut{t}_G \leqW \interpaut{e}_G$.
\end{proof}

\begin{proof}[Proof of~\Cref{thm:gameAutEquiv}]
Fix a preautomaton $\mfA = (Q, \to, \pol)$ and $t, e \in Q$.

For the left-to-right direction, assume we have $\interpaut{t}_\rho \leqW \interpaut{e}_\rho$ for
every $\rho$. By determinacy (\Cref{thm:positional-det}), either Duplicator has a winning strategy,
or Spoiler does. In the first case we are done. In the second case, we run into a contradiction with~\Cref{lem:spo-strat-to-nonred}.

The right-to-left direction is handled by~\Cref{lem:dup-strat-to-red}.
\end{proof}

\subsection{Weihrauch problems and pointed (partial) Weihrauch problems}

Before closing this section, let us briefly discuss what happens if we attempt
to handle natural subclasses of the partial Weihrauch problems: the pointed partial
Weihrauch problems, and what happens when we drop ``partial''.

When it comes to pointed partial Weihrauch degrees, only some minor modifications of
the proof above are required.

\begin{thm}
\label{thm:maingame-pointed}
For any preautomaton $\mfA = (Q, \to, \pol)$ and $t,e \in Q$, the following are equivalent:
\begin{enumerate}
    \item \label{enumitem:mgp-pt}
      $\interpaut{t}_\rho \le \interpaut{e}_\rho$ is valid in the pointed partial Weihrauch degrees.
    \item \label{enumitem:mgp-game}
    Duplicator has a winning strategy in the simulation game $\psimGameAut(\pposition{\{t\}}{e})$.
\end{enumerate}
\end{thm}
\begin{proof}[Proof idea]
  Let us only highlight the key differences with the proof of~\Cref{thm:gameAutEquiv}

For the ``if'' analogous to~\Cref{lem:dup-strat-to-red}, the only
difference is that we no longer have the input $p \in (\Baire)^\Gamma$ to deal with
when defining the reduction. Then the only notable impact is when handling
a \textsc{junk} move
\[ \pposition{X}{e} \xrightarrow{\textsc{junk}} \pposition{X}{e'}\]
To define the forward component of the reduction $\varphi_{X,e}$ from $\varphi_{X,e'}$
in this case, we can use a recursive point
$u_a \in \dom(\rho(a))$ chosen in advance and simply set $\varphi_{X,e}(q) = \tuple{u_a, \lambda x. \; \varphi_{X,e'}(q)}$.

For the ``only if'' part, we need to modify the ``generic'' interpretation $G$
given in \Cref{def:genericinterp} by taking $G'(a) = 1 \sqcup G(a)$.
Then the key difference is that, in the proof of the analogue of \Cref{lem:main-spo-technical},
where in the case where $\pol(e) \in \Sigma$, we me have to potential inputs $u$
to consider: either $u = \tuple{\inc_1(\tuple{ }), f}$ or $u = \tuple{\inc_2(\token(k, t)), f}$.
The latter case is dealt with as before in the proof of \Cref{lem:main-spo-technical}, while the
former is trivial to handle by using a junk move, $f \cdot \unitelt \in \dom(\interpaut{e'}_G)$ and keeping the history unchanged.
\end{proof}

If we drop ``partial'', then we run into other issues. The main one is that
the game no longer works as is, as for instance, Spoiler wins in $\simGame(\position{\emptyset}{0 \star a}{0})$,
while $0 \star a \le 0$ is valid in the Weihrauch degree. This issue comes up
more generally in automata $\mfA^{(e)}$ with states $e$ such that $\interpaut{e}_1 = 0$
(which entails $\interpaut{e}_\rho = 0$ for any $\rho$ valued in Weihrauch problems.
Let us call such states \emph{null}. If $\mfA^{(e)}$ is finite, the it is
easy to compute the set of nullable states. Then, one can 
compute the restriction $\mfA_1^{(e)}$ of $\mfA^{(e)}$ to non-null states
and check that we have $\interp{\mfA_1^{(e)}}_\rho = \interpaut{e}_\rho$ for
any $\rho$ valued in Weihrauch problems. When either $e$ or $t$ is null,
we can trivially
decide whether $\interp{\mfA^{(t)}}_\rho \leqW \interp{\mfA^{(t)}}_\rho$ is
valid for all $\rho$ valued in Weihrauch problem by case analysis. Hence,
we can restrict to the case of preautomata with non-null states without great loss
of mathematical content\footnote{If one insists to handle both null states and non-null states in a single game, it would also be possible to
modify $\simGameAut$ to have special ending conditions when null states are involved
in positions.}.

\begin{thm}
\label{thm:maingame-nonpartial}
For any preautomaton $\mfA = (Q, \to, \pol)$ that does no contain null states and $t,e \in Q$, the following are equivalent:
\begin{enumerate}
    \item 
      $\interpaut{t}_\rho \le \interpaut{e}_\rho$ is valid in the partial Weihrauch degrees.
    \item
      Duplicator has a winning strategy in the simulation game $\simGameAut(\position{\emptyset}{\{t\}}{e})$.
\end{enumerate}
\end{thm}
\begin{proof}[Proof idea]
The direct implication analogous to~\Cref{lem:dup-strat-to-red} requires virtually
no modifications. However, the reverse implication requires more care, as the
generic interpretation given in~\Cref{def:genericinterp} is not valued in Weihrauch
problems. To fix this, one should rather consider the following alternative definition,
taking $G = G_{\ge 0}$, where the clause for domains is modified to filter out
inputs with no answers:
\[
\begin{array}{rcl}
  \dom(G_{n+1}^\top(a)) = \dom(G_{n+1}(a)) &=& \{\token(n,t) \mid \pol(t) = a, t \to t', \interpaut{t}_{G_n} \neq 0\}
\\
G_{n+1}^\top(a)(\token(n, t)) &=&  \emptyset \\
G_{n+1}(a)(\token(n, t)) &=&  \dom(\top_\done \star \interpaut{t'}_{G_n^\top}) ~~~ \text{when $t \to t'$}
\\
G_{\ge n}(a) &=& \bigcup\limits_{k \ge n} G_k(a)
\end{array}\]
Here one difference with the previous generic interpretation is that domains are smaller.
For instance, for $\mfA = \reAut$ (minus its null states), we no longer have $\token(0,a \star 1) \in \dom(G(a \star a))$.
So, for instance, if one wants to witness, say, the non-reduction $a \star a \star a \not\le a$, one
needs to consider an input $\execute \cdot \code{a \star a \star a} \cdot \tuple{\token(n, a \star a \star a)}$
for $n \ge 3$, which exceeds the number of $\textsc{fill}$ moves executed in the Spoiler strategy.
Hence, to adapt the argument, one needs to adapt the notion of price of a position.
For this, it is useful to introduce the notion of \emph{price} of a state defined by
\[\price(t) = \inf \{n \in \Nat \mid \interpaut{t}_{G_n} \neq 0\}\]
If $t$ is a non-null state, this quantity is finite and morally corresponds to
the least number of oracle calls one needs to make to $t$.
Then state that the \emph{price} of a positional winning Spoiler strategy for
is given by the following sum, where 
$\chi \subseteq Q$
ranges over all terms occurring in a reachable position:
\[ \text{\textsc{fill}-depth} \qquad + \qquad \max_{t \in \chi} \price(t)\]
Calling $\price(P)$ the price of the positional winning Spoiler strategy when starting
from a position $P = \position{\Gamma}{X}{e}$, we can then prove that
$\execute \cdot \code{t} \cdot \underline{s_{k,t}^P} \in \dom(\interpaut{t}_{G})$
for all $k \ge \price(P)$. Once that is done, one can straightforwardly adapt
the proofs above, checking that the invariants on prices are indeed satisfied in
the analogue of~\Cref{lem:main-spo-technical}.
\end{proof}

Finally, the case of the pointed Weihrauch degrees can be handled by combining
the adaptations discussed for the previous two theorems. Let us simply state the
result without further explanations.

\begin{thm}
\label{thm:maingame-pointed-nonpartial}
For any preautomaton $\mfA = (Q, \to, \pol)$ that does no contain null states and $t,e \in Q$, the following are equivalent:
\begin{enumerate}
    \item 
      $\interpaut{t}_\rho \le \interpaut{e}_\rho$ is valid in the pointed Weihrauch degrees.
    \item
      Duplicator has a winning strategy in the simulation game $\psimGameAut(\pposition{\{t\}}{e})$.
\end{enumerate}
\end{thm}

\section{Completeness}
\label{sec:completeness}

We show that $\SRKAM$, axiomatized in \Cref{fig:axioms},
is complete with respect to our
game characterization, and thus
universal Weihrauch reducibility.

\begin{thm}
\label{thm:completeness}
If Duplicator wins in $\simGame(\position{\emptyset}{\{e\}}{f})$ with $e, f \in \RegE{\Sigma}$, then $\SRKAM$ proves $e \le f$.
\end{thm}

In order to prove \Cref{thm:completeness}, we will perform what is essentially going to
amount to an induction on the syntax of positions in $\simGame$.
Because of the presence of $(-)^\diamond$ however, a naive induction would
not work. We resolve this by proving a general result that allows to build
partial derivations from certain decompositions of winning Duplicator positions
in $\simGame$ that we call \emph{places}.
That result (\Cref{thm:cover}) will occupy the bulk of this section.

\subsection{Preliminary: derivability modulo $\Gamma$}

We first give an auxiliary definition that will make it more
convenient to relate the first component $\Gamma \subseteq \Sigma$
of a position in $\simGame$ to derivability.

\begin{defi}
\label{def:le-alph}
For $\Gamma \subseteq \Sigma$
and $e,f \in \RegE{\Sigma}$, we say that $\SRKAMT$
derives $e \le f$ modulo $\Gamma$ (and simply write $e \le_\Gamma f$) if $\SRKAMT$ derives
\[\bigsqcap_{b \in \Gamma} (\top \star b) ~~ \sqcap ~~ e \quad \le \quad f \]
\end{defi}

A consequence of completeness will be that Duplicator wins in $\simGame(\position{\Gamma}{\{e\}}{f})$
if and only if $e \le_\Gamma f$ holds. Note in particular that $t \le_\emptyset u$ is the same as $\SRKAMT$ deriving $t \le u$. Before we proceed, we note the following
useful properties of the relation $\le_\Gamma$ that we shall use repeatedly in the sequel.

\begin{lem}
\label{lem:le-alph-sanity}
For any $e, f, e', f' \in \RegE{\Sigma}$ and $\Gamma \subseteq \Theta \subseteq \Sigma$:
\begin{enumerate}
\item If $e \le_{\Gamma} f$, then $e \le_{\Theta} f$.
\item If $\bigwedge_{i \in I} e_i \le f_i  ~ \Rightarrow ~ e' \le f'$
is an axiom of $\SRKAMT$, then $\bigwedge_{i \in I} e_i \le_\Gamma f_i ~ \Rightarrow e' \le_\Gamma f'$ is derivable.
\item If the free variables of $f$ are among $\Gamma$, then $1 \le_\Gamma f$.
\end{enumerate}
\end{lem}
\begin{proof}
Within this proof, let us write $i_\Gamma$ for the term
$\bigsqcap_{b \in \Gamma} (\top \star b)$.
The first point follows from $i_\Gamma \le i_\Theta$,
which holds trivially since $\Gamma \subseteq \Theta$.
For the second point, the case where $I = \emptyset$ is straightforward using the standard properties of meets, so we only discuss the cases where $I \neq \emptyset$.
\begin{itemize}
\item The transitivity axiom is straightforward to handle using the properties of meets.
\item For the axiom stating that $\sqcup$ is the least upper bound, we need to
rely on the distributivity of $\sqcup$ over $\sqcap$ before applying the relevant
axiom. Indeed, assuming
$i_\Gamma \sqcap e \le g$ and $i_\Gamma \sqcap f \le g$, we can derive
\[i_\Gamma \sqcap (e \sqcup f)
\le (i_\Gamma \sqcap e) \sqcup (i_\Gamma \sqcap e) \le g \]
\item The axiom stating that $\sqcap$ is the greatest lower bound
is trivial to handle.
\item For the monotonicity of $\star$, let us assume we have $i_\Gamma \sqcap e \le e'$
and $i_\Gamma \sqcap f \le f'$. Then applying monotonicity of $\star$, we get
$(i_\Gamma \sqcap e) \star (i_\Gamma \sqcap f) \le e' \star f'$.
Using idempotence of $\sqcap$, left half-distributivity of $\sqcap$ over $\star$,
and left-distributivity of $\star$ over $\sqcap$, we can conclude by deriving
\[i_\Gamma \sqcap (e \star f) \le i_\Gamma \sqcap (i_\Gamma \sqcap (e \star f))
\le
i_\Gamma \sqcap ((i_\Gamma \sqcap e) \star f)
=
(i_\Gamma \sqcap e) \star (i_\Gamma \sqcap f) \le e' \star f'
\]
The only non-straightforward part is the equality (third step). It holds because
we can prove more generally that $i_\Gamma \sqcap (a \star b) = a \star (i_\Gamma \sqcap b)$.
This is because, using left-distributivity of $\star$ over $\sqcap$ on the right
and $\top = a \star \top$, we can derive $i_\Gamma = \bigsqcap\limits_{b \in \Gamma} a \star \top \star b = a \star i_\Gamma$.
We can then use left-distributivity of $\star$ over $\sqcap$ again to derive the
desired equality.
\item Finally, for parameterized $\diamond$-induction, assume we have
$i_\Gamma \sqcap (e \sqcap (f \star g)) \le e$. Then we can simply use associativity
of $\sqcap$ and apply parameterized $\diamond$-induction with $a = i_\Gamma \sqcap e$ to conclude.
\end{itemize}
For the last point, it is proven by induction over $f$. All cases are
trivial applications of monotonicity and using that $\SRKAMT$ proves
$1 = 1 \sqcup 1 = 1 \sqcap 1 = 1 \star 1 = 1^\diamond$, except when we have a letter
$b \in \Gamma$. In such a case, we derive
\[i_\Gamma \sqcap 1 \le (\top \star b) \sqcap 1 \le
(\top \sqcap 1) \star b \le 1 \star b \le b\]
\end{proof}

\subsection{Places and covers}

We now introduce the notion of \emph{places} that morally decomposes any
winning Duplicator position, up to some equivalence of problems.

\begin{defi}
    \label{def:place}
Given a finite set $S \subseteq \RegE{\Sigma}$ a
\emph{$S$-place} is a tuple $(\Gamma, (X_s)_{s \in S}, f)$ where:
\begin{itemize}
    \item $\Gamma \subseteq \Sigma$
    \item $X_s \subseteq \RegE{\Sigma}$ is a finite non-empty set
    \item $f \in \RegE{\Sigma}$
    \item Duplicator wins in $\simGame(\position{\Gamma}{\{ s \star t \mid s \in S, t \in X_s\}}{f})$.
\end{itemize}
Call a place \emph{resolved at $s \in S$} if $1 \in X_s$; call a place \emph{resolved}
if it is resolved at some $s \in S$.
\end{defi}

We now turn to defining a well-founded order over places
that will allow us to reason by induction later on.
For this we will manipulate multisets of natural numbers.
We write $\multiset{\phantom{a}}$ for a multiset expression and
write disjoint unions of multisets additively. Recall that the
\emph{multiset ordering} of two multisets of natural numbers
is given by taking $M < N$ if and only if there are multisets
$X, Y$ such that
\[M = N - X + Y \qquad \text{and} \qquad \forall y \in Y. \exists x \in X. ~ y < x\]
This is a well-founded order~\cite{multisetwf}.

\begin{defi}
\label{def:place-wf}
Given two places
$\alpha = (\Gamma, (X_s)_{s \in S}, f)$
and
$\beta = (\Theta, (Y_r)_{r \in R}, g)$, say that $\alpha < \beta$
if and only if one of the following holds:
\begin{itemize}
\item $\Theta$ is strictly included in $\Gamma$
\item $\Gamma = \Theta$ and $f$ is a meet
of spawns of $g$ together with $f <_\Gamma g$.
\item $\Gamma = \Theta$, $f \le_\Gamma g$ and we have a strict inequality for the multisets
\[\sum_{s \in S} \multiset{ \size{e} \mid e \in X_s} \quad < \quad \sum_{r \in R} \multiset{ \size{t} \mid t \in Y_r}\]

\end{itemize}
\end{defi}

\begin{lem}
The order on places given in \Cref{def:place-wf} is well-founded.
\end{lem}
\begin{proof}
This is because we are considering a lexicographic product of well-founded orders.
Reverse inclusion on the alphabets is well-founded since we have a fixed maximal alphabet
$\Sigma$. 
That the second component is well-ordered follows from the fact that any term $g$
has only finitely many spawns, and since meets are provably idempotent in $\SRKAMT$,
only finitely many meets of spawns up to the equivalence generated by $\le_\Gamma$.
As a result, $<_\Gamma$ is well-founded.
The last component is well-ordered because it is a lifting of the multiset ordering
on the natural numbers.
\end{proof}

Now we give the definition of covers, which is a notion that
witnesses the existence of partial derivations in $\SRKAMT$ that make progress
towards deriving the natural inequality associated to a place.

\begin{defi}
\label{def:cover}
A \emph{cover} of an $S$-place $\alpha = (\Gamma, (X_s)_{s \in S}, f)$ is a set of
$S$-places $\cC$ such that
$\pSRKAM$ derives the following entailment for any fresh family of variables
$(x_s)_{s \in S}$:
\[
\left[\bigsqcap_{s \in S} \left( x_s \star \bigsqcap_{t' \in Y_{s}} t' \right) ~~ \le_\Theta ~~ g ~~\text{ for all $(\Theta,(Y_s)_{s \in S},g) \in \cC$} \right]
\quad \Longrightarrow \quad
\bigsqcap_{s \in S} \left( x_s \star \bigsqcap_{t \in X_{s}} t \right) ~~ \le_\Gamma ~~ f
\]
and such that, for any $\beta \in \cC$, we have $\beta \le \alpha$ or $\beta$ is resolved.
We say that a $\cC$ cover of $\alpha$ is \emph{progressing} if for every
$\beta \in \cC$, either $\beta$ is resolved or $\beta < \alpha$.
We say that a cover is \emph{resolving} if all the places in it are resolved.
\end{defi}

Let us first start with the observation that covers compose.

\begin{lem}
\label{lem:cover-compose}
If we have a cover $\cC$ of a $\alpha$, and every $\beta \in \cC$
has a cover $\cC'_\beta$, then $\bigcup_{\beta \in \cC} \cC'_\beta$ is a cover
of $\alpha$, and it is resolving if and only if all the $\cC'_\beta$ are.
\end{lem}
\begin{proof}
Straightforward from the definition.
\end{proof}

Now the main technical result of this subsection, and arguably the whole section,
is the following:

\begin{thm}
\label{thm:cover}
Every place has a resolving cover.
\end{thm}

To prove this by induction on the well-founded order we defined on places,
we argue by case analysis on the shape of the place
under consideration, and more specifically, the terms that occur in the
second component. We give below proofs of all of the key lemmas, besides \Cref{lem:cover-compose},
that allow
to run this induction.

First let us handle the case when this component only
has letters in it, which is the more interesting
base case of the overarching induction.

\begin{lem}
\label{lem:cover-letter}
Any place $(\Gamma, (X_s)_{s \in S}, f)$ where $X_s \subseteq \Sigma$ for every $s \in S$
has a resolving cover.
\end{lem}
\begin{proof}
Here we will exploit a winning strategy of Duplicator in
$\simGame(\position{\Gamma}{\{ s \star b \mid s \in S, b \in X_s\}}{f})$.
The only way to reach a winning position from
$\position{\Gamma}{\{ s \star b \mid s \in S, b \in X_s\}}{f}$ is to play moves
of the second phase until a \textsc{fill} move of the following shape is played
for some $c \in \Sigma$ and $s_0 \in S$:
\[\position{\Gamma}{\{ s \star b \mid s \in S, b \in X_s\}}{f'})
\qquad \xrightarrow{\textsc{fill}} \qquad
\position{\Gamma \cup \{c\}}{\{s_0\} \cup \{ s \star b \mid s \in S, b \in X_s\} \setminus \{s_0 \star x\}}{f''}\]
If we follow a winning Duplicator strategy, we should
reach such a move after a fixed number of steps.
In such a situation, calling $(X'_s)_{s \in S}$ the family given by $X'_s = \{ 1 \mid s = s_0\} \cup X_s \setminus \{c \mid s = s_0\}$, this tells us that $(\Gamma \cup \{c\}, (X'_s)_{s \in S}, f'' \sqcap f)$
is a resolved place. We can then prove by induction on the number of moves of Duplicator
until a \textsc{fill} that joining all such places obtainable this way yields a resolving cover.
\end{proof}

From there on, we can then focus on the case where we know there is at least
one non-letter member of the second component of the place under question.
We then do a case analysis on what can happen there.
The case where $1$ occurs in the second component of a place is trivial, as the singleton
set containing the place is then a resolving cover. Other than this, the easiest cases are
those of $\sqcup$, $\sqcap$ and their units.
\begin{lem}
\label{lem:cover-sqcup-sqcap}
Any place $\alpha = (\Gamma, (X_s)_{s \in S}, f)$ with either $0$, $\top$, $t_1 \sqcup t_2$ or $t_1 \sqcap t_2$
belonging to $\bigcup_{s \in S} X_s$ has a progressing cover.
\end{lem}
\begin{proof}
Assuming that $t_1 \sqcap t_2$ belongs to $X_{s_0}$ for some $s_0 \in S$,
setting $(Y_s)_{s \in S}$ to be
\[Y_s = \{t \in \{t_1, t_2\} \mid s = s_0\} \cup Y_s \setminus \{t_1 \sqcap t_2 \mid s = s_0\}\]
we have that the singleton $\{(\Gamma, (Y_s)_{s \in S}, f)\}$ is a progressing cover of $\alpha$.
We obtain a place because we can assume that Spoiler plays the relevant \textsc{fork} move in $\simGame$ and stay in the winning strategy of Duplicator. It is a progressing cover because the multiset corresponding to $(Y_s)_{s \in S}$ is strictly smaller than the one matching $(X_s)_{s \in S}$.
Similarly, if $\top$ belongs to some $X_{s_0}$, we can obtain
a strictly smaller place covering the original place by removing it.

If $t_1 \sqcup t_2$  belongs to $X_{s_0}$ for some $s_0 \in S$, we can
define two families
$(Y^{(i)}_s)_{s \in S}$ for $i \in \{1,2\}$ by putting
\[Y^{(i)}_s = \{t_i \mid s = s_0\} \cup Y_s \setminus \{t_1 \sqcup t_2 \mid s = s_0\}\]
and check that $\{(\Gamma, (Y^{(1)}_s)_{s \in S}, f), (\Gamma, (Y^{(2)}_s)_{s \in S}, f)\}$
is a progressing cover of $\alpha$. Its two elements are places because we can assume that Spoiler
play either relevant \textsc{explore} move while staying in a winning strategy of Duplicator.
If $0 \in \bigcup_{s \in S} X_s$, then $\emptyset$ covers the place.
\end{proof}

The next lemma deals
with the case where we have a composition $t_1 \star t_2$ there.

\begin{lem}
\label{lem:cover-star}
Assume we are given an
$S$-place $\alpha = (\Gamma, (X_s)_{s \in S}, f)$ and $t_1 \star t_2$ belongs to some $X_{s_0}$.
Set
\[S' = S \cup \{ s_0 \star t_1\} \qquad \text{and, for $s \in S'$,} \qquad
X'_s = \{t_2 \mid s = s_0 \star t_1\} \cup X_s \setminus \{t_1 \star t_2 \mid s = s_0\}\]
Given a resolving cover $\cC'$ of the
$S'$-place  $\alpha' = (\Gamma, (X'_s)_{s \in S'}, f)$, we can compute a
progressing cover for $\alpha$.
Furthermore we also always have $\alpha' < \alpha$.
\end{lem}
\begin{proof}
For each $\beta' = (\Theta, (Y'_s)_{s \in S'}, g) \in \cC'$, we define a corresponding
$S$-place $\beta$ as $(\Theta, (Y_s)_{s \in S}, g)$ by setting
\[Y_{s_0} = \{t_1 \mid 1 \in Y'_{s_0 \star t_1}\} \cup \{ t_1 \star e \mid e \in Y'_{s_0 \star t_1} \} \setminus \{t_1 \star 1\} \quad \text{and} \quad Y_s = Y'_s ~ \text{otherwise}\]
We have that either $\beta < \alpha$ (in the case where $1 \in Y'_{s_0 \star t_1}$, since
we are adding $t_1$) or $\beta \le \alpha$ and $\beta$ is resolving.
It is also a cover of $\alpha$: unfolding the definition and applying some trivial equivalences, we need to check that
\[
\left(x_{s_0} \star t_1 \star \bigsqcap_{t' \in Y'_{s_0 \star t_1}} t'\right) \sqcap \bigsqcap_{s \neq s_0} \left( x_s \star \bigsqcap_{t' \in Y'_{s}} t' \right) ~~ \le_\Theta ~~ g ~~\text{ for all $(\Theta,(Y'_s)_{s \in S},g) \in \cC'$}\]
entails
$
\bigsqcap_{s \in S} \left( x_s \star \bigsqcap_{t \in X_{s}} t \right) \le_\Gamma f
$
But this is equivalent to the entailment given to us by the fact that $\cC'$ is place if
we substitute $x_{s_0} \star t_1$ for $x_{s_0 \star t_1}$.
\end{proof}

Finally, we deal with the case of iteration, which is arguably the most challenging.
\begin{lem}
\label{lem:cover-iter}
Assume we are given an
$S$-place $\alpha = (\Gamma, (X_s)_{s \in S}, f)$ and $t^\diamond$ belongs to some $X_{s_0}$.
Set
\[S' = S \cup \{ s_0 \star t^\diamond\} \qquad \text{and, for $s \in S'$,} \qquad
X'_s = \{t \mid s = s_0 \star t^\diamond\} \cup X_s \setminus \{t^\diamond \mid s = s_0\}\]
If we have a resolving cover $\cC'$ of the
$S'$-place  $\alpha' = (\Gamma, (X'_s)_{s \in S'}, f)$, then $\alpha$ has a progressing cover.
Furthermore, $\alpha' < \alpha$.
\end{lem}
\begin{proof}
Without loss of generality, let us assume that $X_{s_0} = \{t^\diamond\}$ and $s_0 \star t^\diamond \not\in S$.
First, let us partition $\cC'$ into the disjoint union $\cC' = \cC'_< \uplus \cC'_\diamond$
where we put places $\beta' = (\Theta,(Y'_s)_{s \in S'},g)$ in $\cC'_<$ whenever either
of the following holds:
\begin{itemize}
    \item $\Theta$ is a strict superset of $\Gamma$, or
    \item $g <_\Gamma f$, or
    \item $\beta'$ is resolved at some $s \in S' - \{s_0 \star t^\diamond\}$
\end{itemize}
For any such $\beta'$, define a corresponding $S$-place
$\beta = (\Theta,(Y_s)_{s \in S},g)$ by putting
$Y_{s_0} = \{ t^\diamond \star e \mid t \in Y'_{s_0 \star t^\diamond}\}$ and $Y_s = Y'_s$ for $s \neq s_0$.
Note that we have that either $\beta$ is resolved or $\beta < \alpha$. Write $\cC_<$ for
the set of all such $\beta$s.
Note that places $\gamma' \in \cC'_\diamond$ necessarily have shape $(\Gamma, (Z'_s)_{s \in S'}, h)$
with $h \le_\Gamma f$ (by the definition of a cover) but $h \not<_\Gamma f$. Hence
$h$ is provably equivalent to $f$ modulo $\Gamma$ in $\SRKAMT$, so we may as well assume
$h = f$.
Finally define the resolved place $\alpha_1 = (\Gamma, (X''_s)_{s \in S}, f)$ by
putting $X''_{s_0} = \{1\}$ and $X''_s = X_s$ otherwise.

By construction, if $\cC = \cC_< \cup \{\alpha_1\}$ is a cover at all, it is progressing.
Let us check that it is actually a cover. Unfolding the definitions, what we have to
check is that the following entailment holds:
\[\begin{array}{lrccl}
&x_{s_0} \sqcap \bigsqcap\limits_{s \in S}
\left( x_s \star \bigsqcap\limits_{e \in X_{s}} e \right) &\le_\Gamma& f & \text{\footnotesize (base assumption)} \\
\text{and} &
\left(x_{s_0} \star t^\diamond \star \bigsqcap\limits_{e \in Y'_{s_0}} e \right) \sqcap
\bigsqcap\limits_{s \in S \setminus \{s_0\}}
\left( x_s \star \bigsqcap\limits_{e \in Y'_{s}} e \right)
&\le_\Theta & g &\text{ for all $(\Theta,(Y'_s),g)_{s \in S} \in \cC'_<$} \\
\\
\text{imply} &
\left(x_{s_0} \star t^\diamond\right) \sqcap \bigsqcap\limits_{s \in S} \left( x_s \star \bigsqcap\limits_{t \in X_{s}} t \right) &\le_\Gamma& f & (\dagger)
\end{array}
\]
Let us assume the premises. Now we will use the fact that $\cC'$ is a cover
for $\alpha'$. This means that
\[\begin{array}{lrccl}
&
\left(x_{s_0 \star t^\diamond} \star \bigsqcap\limits_{e \in Y'_{s_0}} e \right) \sqcap
\bigsqcap\limits_{s \in S \setminus \{s_0\}}
\left( x_s \star \bigsqcap\limits_{e \in Y'_{s}} e \right)
&\le_\Theta & g &\text{ for all $(\Theta,(Y'_s),g)_{s \in S} \in \cC'_<$} \\
\text{and} &
\left(x_{s_0 \star t^\diamond} \star \bigsqcap\limits_{e \in Z'_{s_0 \star t^\diamond}} e\right) \sqcap
\bigsqcap\limits_{s \in S \setminus \{s_0\}}
\left( x_s \star \bigsqcap\limits_{e \in Z'_{s}} e \right)
&\le_\Gamma & f &\text{ for all $(\Gamma,(Z'_s),f)_{s \in S} \in \cC'_\diamond$} \\\\
\text{imply} &
\left(x_{s_0 \star t^\diamond} \star t\right) \sqcap \bigsqcap\limits_{s \in S} \left( x_s \star \bigsqcap\limits_{e \in X_{s}} e \right) &\le_\Gamma& f
\end{array}
\]
Now, if we substitute $(x_{s_0} \star t^\diamond) \sqcap f$ for $x_{s_0 \star t^\diamond}$ in the entailment above, we can see that its premises are fulfilled: for premises tied to the places coming
from $\cC'_<$, this is using our previous assumption and the fact that $(x_{s_0} \star t^\diamond) \sqcap f \le x_{s_0} \star t^\diamond$, for the other this follows from the fact that
we always have $1 \in Z'_{s_0 \star t^\diamond}$ and left-distributivity of $\star$
over $\sqcap$. So all in all, we have
\[\left(\left(\left(x_{s_0} \star t^\diamond\right) \sqcap f\right) \star t\right) \sqcap \bigsqcap\limits_{s \in S} \left( x_s \star \bigsqcap\limits_{e \in X_{s}} e \right) ~~\le_\Gamma~~ f\]
Using half-distributivity on the left, we can derive further that
\[
\left( f \star t\right) \sqcap \left(x_{s_0} \star t^\diamond\right)  \sqcap \bigsqcap\limits_{s \in S} \left( x_s \star \bigsqcap\limits_{t \in X_{s}} t \right)
~~\le_\Gamma~~ f\]
We can then apply parameterized $\diamond$-induction to get
\[
\left( f \star t^\diamond\right) \sqcap \left(x_{s_0} \star t^\diamond\right)  \sqcap \bigsqcap\limits_{s \in S} \left( x_s \star \bigsqcap\limits_{t \in X_{s}} t \right)
~~\le_\Gamma~~ f\]
Now using the base assumption, we can substitute $f$ on the left-hand side and
obtain
\[
\left( \left(x_{s_0} \sqcap \bigsqcap\limits_{s \in S}
\left( x_s \star \bigsqcap\limits_{e \in X_{s}} e \right)\right) \star t^\diamond\right) \sqcap \left(x_{s_0} \star t^\diamond\right)  \sqcap \bigsqcap\limits_{s \in S} \left( x_s \star \bigsqcap\limits_{t \in X_{s}} t \right)
~~\le_\Gamma~~ f\]
Using half-distributivity of $\sqcap$ over $\star$, we can shift the first occurrence
of $t^\diamond$ on the left so that it produces another occurrence of $x_{s_0} \star t^\diamond$,
and then conclude that, by idempotence of $\sqcap$, we have the desired inequality $(\dagger)$.
\end{proof}

With these lemmas proven, we can piece together the proof of \Cref{thm:cover}
by induction over the well-founded order on places as advertised.
The details are left to the reader.

\subsection{Proof of completeness}

We prove completeness in two steps: first, we do it for inequalities $e \le_\Gamma f$
where we have $e \le 1$ and then use that to derive the general case.
Semantically, recall that $P \leqW 1$ (or equivalently, $P \equivW 1 \sqcap P$)
means that any $P$-question has a computable answer. In such a case, we can
note that if we want a reduction $1 \sqcap P \leqW Q$, it is sufficient to provide the
forward function, as the backwards can simply disregard its input and output $\inc_1(\tuple{ })$.
Recalling that $\top \star Q$ is equivalent to the problem ``give me a $Q$-question, but there is no answer'',
having a reduction at all is then equivalent to $P \leqW \top \star Q$.
A similar reasoning, tells us that for any problems $P$ and $Q$ such that $Q$
questions do not have answers (equivalently, $Q \equivW \top \star Q$), a
reduction $P \leqW Q$ induces a reduction $1 \sqcap P \leqW Q$.

$\SRKAMT$ is strong enough to make this formal, working modulo $\Gamma$.

\begin{lem}
\label{lem:topstar-galois}
Provably in $\SRKAMT$, there
is a Galois connection $(1 \sqcap -, \top \star -)$ modulo $\Gamma$:
for every $e, f$, we have
$e \le_\Gamma \top \star (1 \sqcap e)$ and
$1 \sqcap (\top \star f) \le_\Gamma f$; this implies that
\[ 1 \sqcap e \le_\Gamma f \qquad \Longleftrightarrow \qquad e \le_\Gamma \top \star f\]
\end{lem}
\begin{proof}
Throughout, we use~\Cref{lem:le-alph-sanity} to employ the axioms of $\SRKAMT$
(and note in passing that $1 \sqcap -$ and $\top \star -$ are monotone with respect to
$\le_\Gamma$).
Let us treat the two inequalities in turn:
\begin{itemize}
\item For the first one we use that $\top$ is the unit for $\sqcap$,
that $1$ is the unit for $\star$, the monotonicity of $\star$ and finally, the
left-distributivity of $\star$ over $\sqcap$.
\[e =_\Gamma \top \sqcap e =_\Gamma  
(\top \star 1) \sqcap (1 \star e) \le_\Gamma 
(\top \star 1) \sqcap (\top \star e) \le_\Gamma
\top \star (1 \sqcap e)\]
\item For the second one, the key step is the left-half distributivity
of $\sqcap$ over $\star$.
\[1 \sqcap (\top \star f) \le_\Gamma (1 \sqcap \top) \star f =_\Gamma 1 \star f =_\Gamma f\]
\end{itemize}
The equivalence afterwards holding is a standard property of Galois connections
(see e.g.~\cite[Chapter 7]{davey2002introduction} for an introduction to Galois connections).
\end{proof}

\begin{lem}
\label{lem:completeness-le-one}
If Duplicator wins in $\simGame(\position{\Gamma}{X}{f})$
and $1 \in X$, then we have $\bigsqcap_{e \in X} e \le_\Gamma f$
provably in $\SRKAMT$.
\end{lem}
\begin{proof}
Write $e_X$ for $\bigsqcap_{e \in X} e$
throughout.
Duplicator also wins from position
\[\position{\Gamma}{\{0 \star \top \star e \mid e \in X\}}{\top \star f}
\] in $\simGame$: they may play as in a winning strategy in $\simGame(\position{\Gamma}{X}{f})$
as usual, unless the game gets stuck at a position which is losing for Spoiler (with $\top$
  on the right-hand side).
Since we have such a winning strategy, it means
that we have an $\{0\}$-place $(\Gamma, (\{ \top \star e \mid e \in X\})_0, \top \star f)$.
Let us take a resolving cover $\cC$ of that, which we know exists by \Cref{thm:cover}.
Substituting $0$ for $x_0$, the cover gives us that the following entailment is
derivable in $\SRKAMT$:
\[
\left[0 \star \bigsqcap_{t \in Y} t ~~ \le_\Theta ~~ g ~~\text{ for all $(\Theta,(Y)_0,g) \in \cC$} \right]
\quad \Longrightarrow \quad
0 \star \top \star e_X ~~ \le_\Gamma ~~ \top \star f
\]

Each of the premises of the entailment is trivially derivable in $\SRKAMT$: since
the cover is resolving, we have $1 \in Y$, in which case the left-hand-side is provably equivalent to
$0$ in $\SRKAMT$ via $0 \star (1 \sqcap t) \le 0 \star 1 = 0$.
Hence $\SRKAMT$ derives the conclusion of the above entailment and a fortiori
\[  e_X ~~ \le ~~ 1 \star e_X ~~ \le  ~~ \top \star e_X ~~ \le ~~ 0 \star \top \star e_X ~~ \le_\Gamma ~~ \top \star f \]
By \Cref{lem:topstar-galois}, we thus have a derivation of $1 \sqcap e_X \le_\Gamma f$, and since $1 \in X$,
we have $e_X \le 1 \sqcap e_X$ and can conclude that $\SRKAMT$ derives $e_X \le_\Gamma f$.
\end{proof}

\begin{thm}
\label{thm:completeness-expanded}
If Duplicator wins in $\simGame(\position{\Gamma}{X}{f})$, then we have $\bigsqcap_{e \in X} e \le_\Gamma f$
provably in $\SRKAMT$.
\end{thm}
\begin{proof}
Clearly we have that $(\Gamma, (X)_1, f)$ is a $\{1\}$-place.
By \Cref{thm:cover}, we know it has a resolving cover $\cC$.
Taking the corresponding entailment (see~\Cref{def:cover}), substituting the
variable $x_1$ with $1$ and performing the simplification inducted by $1 \star t =_\Gamma t$,
we have the following entailment:
\[
\left[\bigsqcap_{t \in Y} t ~~ \le_\Theta ~~ g ~~\text{ for all $(\Theta,(Y)_1,g) \in \cC$} \right]
\quad \Longrightarrow \quad
\bigsqcap_{e \in X} e ~~ \le_\Gamma ~~ f
\]
To conclude it suffices to show that every premise of this entailment is derivable in $\SRKAMT$.
So fix $(\Theta, (Y)_1, g) \in \cC$. Since the
cover is resolving, we have that $1 \in Y$ by~\Cref{def:place}.
Since $(\Theta, (Y)_1, g)$ is a place, Duplicator has a winning strategy in $\simGame(\position{\Theta}{Y}{g})$.
Hence we can apply 
\Cref{lem:completeness-le-one} to see that $\SRKAMT$ proves $\bigsqcap_{t \in Y} t \le_\Theta g$.
\end{proof}

It is clear that \Cref{thm:completeness} is a particular case
of \Cref{thm:completeness-expanded} (take $X = \{e\}$ and $\Gamma = \emptyset$).

\subsection{Completeness for the pointed Weihrauch degrees}

Finally, we finish by explaining how to derive some completeness results
for pointed Weihrauch degrees
by way of the game characterization (\Cref{thm:maingame-pointed}), using
a similar proof strategy as above. As with \Cref{thm:maingame-pointed},
the proofs can be adaptations of what we did for partial Weihrauch
degrees, so we only highlight key differences rather than spelling them out in
full.

\begin{thm}
\label{thm:completeness-pointed}
If we have expressions $e, f \in \RegE{\Sigma}$
that do not contain $0$ or $\top$, if Duplicator wins
$\psimGame(\pposition{\{e\}}{f})$, then $e \le f$
is derivable in the theory built as follow:
\begin{itemize}
    \item take the fragment of $\SRKAMT$ where all axioms
pertaining to $\top$ and $0$ are removed
   \item add the universal axiom $1 \le a$
\end{itemize}
\end{thm}
\begin{proof}[Proof idea]
We can adapt the definition of places (\Cref{def:place}) as follows:
a pointed $S$-place is a pair $((X_s)_{s \in S},f)$ such that
Duplicator wins $\psimGame(\pposition{\{ s \star t \mid s \in S, t \in X_s\}}{f})$.
We can similarly adapt the definition of covers (\Cref{def:cover})
by simply replacing all occurrences of $\le_\Gamma$ (for all alphabets
$\Gamma$ occurring in the definition) with derivability in our theory. Then all pointed places also have covers (the analogue of \Cref{thm:cover}), following the exact same proof
strategy. The only difference is in the proof of the analogue of \Cref{lem:cover-letter}
where we need the axiom $1 \le a$ to check that our entailment holds when we follow
a \textsc{junk} move in the strategy.
Then, we can proceed directly by arguing as in \Cref{thm:completeness-expanded} and
check that the premise of the entailment holds because $1$ is the bottom element now (there is no analogue to \Cref{lem:completeness-le-one} that is required here, so no
need to introduce $\top$).
\end{proof}

\begin{defi}
\label{def:SRKA}
Call $\SRKA$ the theory obtained as follows:
\begin{itemize}
\item take the fragment of $\SRKAMT$ where all axioms
pertaining to $\top$, $0$ and $\sqcap$ are removed
\item add the $\diamond$-induction scheme
   $x \star y \le x ~~\Longrightarrow~~ x \star y^\diamond \le x$.
\end{itemize}
Call $\pSRKA$ the theory obtained by adjoining $1 \le a$ to $\SRKA$.
\end{defi}

\begin{lem}
\label{lem:pointed-spawn-le}
For every $e, f \in \RegE{\Sigma}$ not containing $0$, $\top$ or $\sqcap$,
if $f$ is a spawn of $e$, then $\pSRKA$ derives $f \le e$.
\end{lem}
\begin{proof}
We only need to establish that whenever $e \to f$, we have $f \le e$.
We do that by induction over $e$ and case analysis:
\begin{itemize}
\item if $e = e_1 \sqcup e_2 \to e_i = f$ for some $i \in \{1,2\}$,
this follows from the axiom $e_i \le e_1 \sqcup e_2$
\item if $e = e_1 \star 1 \to e_1 = f$, this is immediate by unitality of $1$
\item if $e = e_1 \star e_2 \to e_1 \star e_2' = f$, then the induction hypothesis
yields $e_2' \le e_2$, and by monotonicity we get the desired result
\item if $e = t^\diamond \to 1 = f$, this is immediate from an axiom
\item if $e = t^\diamond \to t^\diamond \star t$, this is also an axiom
\item we have covered all the cases as we cannot have $e = 1 \to f$ for any $f$.
\end{itemize}
\end{proof}

\begin{thm}
\label{thm:completeness-pointed-sqcap-free}
If we have expressions $e, f \in \RegE{\Sigma}$
that do not contain either of $0$ or $\top$ or $\sqcap$, if Duplicator wins
$\psimGame(\pposition{\{e\}}{f})$, then $e \le f$
is derivable in $\pSRKA$.
\end{thm}
\begin{proof}[Proof idea]
We can use the same strategy as for~\Cref{thm:completeness-pointed},
with a further modification to the notion of places and covers. Firstly,
as we do not have $\sqcap$, we know that the sets in the positions
are at most singletons. We can assume as much for places: now a
pointed $s$-place (for $s \in \RegE{\Sigma}$) is simply a pair $(e, f)$
of terms (without $0$, $\top$ or $\sqcap$) such that Duplicator wins in $\psimGame(\pposition{\{s \star e\}}{f})$.
We take an adapted definition of covers much like
in~\Cref{thm:completeness-pointed}, but with one additional requirement:
for a cover
$\cC$ of the pointed place $(e,f)$, we additionally
require for any place $(t, g) \in \cC$ that
$g$ be a spawn of $f$ (instead of a meet of spawns
of $f$). The only adaptation to make for this requirement is in the proof
of the analogue of \Cref{lem:cover-letter}, where in the base case we may
simply take $f''$ instead of $f \sqcap f''$ in the resolved place obtained by following
the \textsc{fill} move. This is correct for our updated notion of cover
because $f''$ is a spawn of $f$ and we have \Cref{lem:pointed-spawn-le}.
We can then proceed to reason as for \Cref{thm:completeness-pointed}.
\end{proof}

\begin{rem}
$\SRKA$ is \emph{not} complete for
Weihrauch reducibility for terms including compositions
as we cannot derive $b \star a \le a \star b \star a$;
this can be proven by showing that $\SRKA$ can only derive inequalities when there
is a simulation between the corresponding transition systems\footnote{See e.g.~\cite[Definition 7.47]
{baierkatoen} for a definition; it essentially corresponds to a version of $\psimGame$
where \textsc{junk} moves are not allowed at all.}, and it is clear there
is none between $b \star a$ and $a \star b \star a$. 
\end{rem}

\section{Computational complexity of deciding universal validity}
\label{sec:complexity}

In this section we look at the complexity of deciding whether an inequality
is valid for various fragments of our signature, both in the ordinary
Weihrauch degrees and the pointed ones. Interestingly, we have some
genuinely different complexities for the pointed case and the general case
for the signature that does not contain any $\sqcap$.
Recall that we write $\Wei$ for the structure of (ordinary) Weihrauch degrees and $\ptWei$ for the
(ordinary) pointed Weihrauch degrees. Here, proofs that are stated for
ordinary Weihrauch degrees also go through for partial degrees, using the same
game characterization.

\subsection{Upper bounds}

Let us first start by easy observations that follows from \Cref{thm:positional-det}.

\begin{lem}
Given $e,f \in \RegE{\Sigma}$, solving ``is $e \le f$ valid in the partial Weihrauch degrees?'' is doable in
$\Exptime$.
\end{lem}
\begin{proof}
This follows from the fact that only at most $2^{1 + 2 \cdot \size{e} + \size{\Sigma}} \cdot \size{f}$ positions are reachable in $\simGame(\position{\emptyset}{\{e\}}{f})$ and
\Cref{thm:positional-det}. The bound on the number of reachable positions
follows from the fact that, for any
reachable position $\position{\Gamma}{X}{g}$, we have that $g$ is a
spawn of $f$ and $X$ is included in the set of spawns of $e$.
\end{proof}

\begin{lem}
Given $e,f \in \RegE{\Sigma}$ with $e$ not containing any $(-)^\diamond$ or $\sqcap$, solving ``is $e \le f$ valid in $\ptWei$?'' is doable in
$\Ptime$.
\end{lem}
\begin{proof}
Note that the positions reachable from $\pposition{\{e\}}{f}$
are necessarily of the shape $\pposition{\{e'\}}{f'}$ with $e'$ and $f'$
spawns of $e$ and $f'$ respectively because no \textsc{fork} moves are available
(the only term of polarity $\spopol$ in the fragment under consideration is $1$).
Therefore, there are only polynomially many positions in this B\"uchi game, and
we can apply \Cref{thm:positional-det} to conclude.
\end{proof}

Now we move onto a more involved argument which still uses coarse
quantitative bounds.

\begin{lem}
Given $e,f \in \RegE{\Sigma}$ with $e$ not containing any $(-)^\diamond$, solving ``is $e \le f$ valid in $\Wei$?'' is doable in
$\Pspace$.
\end{lem}
\begin{proof}
We use the fact that $\Pspace$ is the same as alternating polynomial-time.
The alternating algorithm that we use will simply consist of playing the game,
recording the trace of all played moves. After each move, we check if we
land in a position we have seen before; if so, we then check if there is a
winning Duplicator position on the cycle. If this is the case, we accept,
otherwise, we reject. If we do not run into a cycle, we let the game play
to completion and accept if and only if we end in a winning position for
Duplicator.
This algorithm is sound because the game is positionally determined, and
all cycles in positional Duplicator (respectively Spoiler) winning strategies
must (respectively must not) contain a winning Duplicator positions.

Now we need to show that this algorithm runs in polynomial time. Clearly
the bookkeeping associated with recording the move history is polynomial
in the size of the position and the length of the plays we are simulating.
So it suffices to show that cycle-free plays in $\simGame(e,f)$ are of size
at most polynomial in $(e,f)$.
Consider one such path $(\position{\Gamma_i}{X_i}{f_i})_{i \le n}$.
The sequence $(\Gamma_i)_{i \le n}$ is monotonic for inclusion; as a result,
there are at most $\size{\Sigma}$ distinct $\Gamma_i$s along such a path.
All of the $f_i$ are spawns of $f$, and there are at most $\size{f}$ of those.
Since $e$ does not contain any $(-)^\diamond$, $X_i$ consists of the
minimal spawns in $\bigcup_{j \le i} X_j$ at every step. This union can contain
at most $\size{e}$ elements. As a result, cycle-free paths have length inferior
to $\size{\Gamma} \cdot \size{f} \cdot \size{e}^2$.
\end{proof}

This line of argumentation does not apply in the presence of iterations.

\begin{exa}
\label{ex:exp}
Consider the $n$-indexed family of reducible pairs of terms
\[ e_n = \left(\bigsqcap_{i = 0}^n (x_i \star y_i)\right)^\diamond  \quad \le\quad  \left(\bigsqcap_{i = 0}^{n} (x_i \sqcup y_i)\right)^\diamond = f_n\]
There are positional Duplicator winning strategies in $\psimGame(\pposition{\{e_n\}}{f_n})$,
and all such positional winning strategies have cycles of exponential size in $n$.

To analyze this situation, first note that Duplicator can never play
a \textsc{junk} move from any position $X \vdash f_n \star z$ for $z \in \{x_j, y_j\}$:
otherwise, by analysing the previous moves, we can see that Spoiler would force a
losing loop for Duplicator:
\[X \vdash f_n
\xrightarrow{\textsc{choose}}
X \vdash f_n \star \bigsqcap_{i = 0}^n x_i \sqcup y_i
\xrightarrow{\textsc{alea}}
X \vdash f_n \star (x_j \sqcup y_j)
\xrightarrow{\textsc{choose}}
X \vdash f_n \star z
\xrightarrow{\textsc{junk}}
X \vdash f_n\]
That there is a winning strategy is relatively easy to check: for
$I \subseteq \{0,\ldots, n\}$, let us consider the sets $X_{n, I}$
and $Y_{n, I}$
of terms given by
\[
X_{n, I} \quad =\quad
\{ e_n \star x_i \mid i \in I\}
\qquad\text{and}\qquad 
  Y_{n, I} \quad =\quad
\{ e_n \star x_i \star y_i \mid i \le n, i \in I\}
\]
Analyzing the unfolding of the initial Spoiler moves, 
we can see it suffices to define winning Duplicator strategies from all positions
$X_{n, I} \cup Y_{n, I'} \vdash f_n$ with $I \cup I' = \{0, \ldots, n\}$.
\begin{enumerate}
\item On any position of the shape $X_{n, I} \cup Y_{n, I'} \vdash f_n$, Duplicator
should first \textsc{choose}
to unfold $f_n$ to reach
$X_{n, I} \cup Y_{n, I'} \vdash f_n \star \bigsqcap_{i = 0}^n x_i \sqcup y_i$.
\item
\label{enumitem:expstratfill}
Then 
Spoiler plays an \textsc{alea} move that yields some $j \le n$ so that the game
goes to $X_{n, I} \cup Y_{n, I'} \vdash f_n \star (x_j \sqcup y_j)$.
\item Then Duplicator moves to
$X_{n, I} \cup Y_{n, I'} \vdash f_n \star x_j$
if $j \in I$. Otherwise if $j \in I'$, then Duplicator moves to
$X_{n, I} \cup Y_{n, I'} \vdash f_n \star y_j$.
\item
\label{enumitem:expstratend}
Then Duplicator plays the unique available \textsc{fill} move. If the
previous position was $X_{n, I} \cup Y_{n, I'} \vdash f_n \star y_j$, then
we reach the position $X_{n, I \cup \{j\}} \cup Y_{n, I' \setminus \{j\}}$.
Otherwise, depending on Spoiler's \textsc{explore} move, we either
reach the position $X_{n, I \setminus \{j\}} \cup Y_{n, I'} \cup \{1\} \vdash f_n$
(where Duplicator can end the game and win by choosing not to unfold $f_n$),
or, after a few \textsc{fork} moves, we reach $X_{n, I \setminus \{j\}} \cup Y_{n, \{0, \ldots, n\}} \vdash f_n$.
\end{enumerate}
The strategy outlined above is winning for Duplicator as it forces Spoiler
to play an \textsc{explore} move at step~\ref{enumitem:expstratend}.
One should note that any positional winning strategy for Duplicator essentially
must have the same behaviour: the differences can only happen at
steps~\ref{enumitem:expstratfill} and~\ref{enumitem:expstratend}:
\begin{itemize}
\item At step~\ref{enumitem:expstratfill}, one could play any of the two
alternatives if we have $j \in I \cap I'$.
\item At step~\ref{enumitem:expstratend}, it would be possible for Duplicator not
to end the game immediately on a position $X_{n, I \setminus \{j\}} \cup Y_{n, I'} \vdash f_n$
if and only if $j \in I'$.
\end{itemize}
Now let us show that such a winning positional Spoiler strategy must have exponential
cycles. To do so, we essentially want to show there is 
a single cycle with distinct positions $X_{n, I} \cup Y_{n, I'} \vdash f_n$ for \emph{all} possible
values of $I \subseteq \{0, \ldots n\}$ (in contrast, we won't have much control on $I'$),
which we can assume is enumerated by flipping membership of a single element at a time.
To do this, it suffices to show that, for any $I \subseteq \{0, \ldots, n\}$ and
$j \le n$, Spoiler can force the game to follow a path where $X_{n,I}$ becomes
$X_{n, I \mathrel{\Delta} \{j\}}$ (writing $\Delta$ for the symmetric difference operator)
without going through any other $X_{n, I'}$.

More formally, we need to show that for any $I, I' \subseteq \{0, \ldots, n\}$
(with $I \cup I' = \{0, \ldots, n\}$)
and $j \le n$, there exists a play $(Z_k \vdash f'_k)_{k = 0}^{K}$ in the strategy
with $Z_0 = X_{n, I} \cup Y_{n, I'}$, $f'_0 = f'_K = f_n$, $Z_K \cap X_{n, \{0, \ldots, n\}} = 
X_{n, I \mathrel{\Delta} \{j\}}$ and $Z_m \cap X_{n, \{0, \ldots, n\}} = X_{n,I}$ for $m < K$.
Since Duplicator is winning, we must have $Z_1 = Z_0$ and
\[f'_1 ~~=~~ f_n \star \bigsqcap_{i = 0}^n x_i \sqcup y_i\]
Spoiler plays the $j$th option, so that we get $Z_2 = Z_1$ and
$f'_2 ~~=~~ f_n \star (x_j \sqcup y_j)$.
Then, we have two scenarios:
\begin{itemize}
\item if $j \not\in I$, then $j \in I'$ and Duplicator must enforce
$Z_3 = Z_2$ and $f'_3 = f_n \star y_j$, and then after the subsequent \textsc{fill}
move, we get $Z_4 = X_{n, I \cup \{j\}} \cup Y_{n, I' \setminus \{j\}}$ and $f'_4 = f_n$,
which concludes the partial play.
\item otherwise $j \in I$. There are two possible scenarios
(where again $Z_3 = Z_2$):
\begin{itemize}
\item either $f'_3 = f_n \star x_j$. Then, following the \textsc{fill} move of
Duplicator, and then the \textsc{explore} and \textsc{fork} moves of Spoiler
that do not produce $1$, we end up with $Z_K = X_{n, I \setminus \{j\}} \cup Y_{n, \{0, \ldots, n\}}$.
\item otherwise, $f'_3 = f_n \star y_j$, in which case we must have $j \in I'$
as well. Following the \textsc{fill} move, we end up 
with $Z_4 = X_{n, I} \cup Y_{n, I' \setminus \{j\}}$ and $f'_4 = f_n$
again. In which case we can follow the same procedure again to complete the play, with termination
ensured
since $j$ has been removed from $I'$.
\end{itemize}
\end{itemize}
\end{exa}

\begin{exa}
Building on the previous example, note that we have that $\top \star e_n \not\le f_n$.
Hence there are \emph{some} Spoiler strategies of exponential length in $n$:
against the (obvious modification of the) winning Duplicator of the previous example,
Spoiler may force to enumerate $2^{n-1}$ positions $Z_i \vdash f_n$ with
$Z_i \cap X'_{n, \{0, \ldots, n\}} = X'_{n, I_i}$ with $\{I_i \mid i \le 2^{n-1}\} = \mathcal{P}(\{0, \ldots, n-1\})$,
assuming
\[
X'_{n, I} \quad =\quad \{ \top \star (e_n \star x_i) \mid i \in I\}
\]

Then, they can force going to a position $Z_{2^{n-1} + 1} \vdash f_n$ with
$Z_i \cap X'_{n, \{0, \ldots, n\}} = X'_{n, I_i \cup \{n\}}$, and then force Duplicator
to ``clear out $Z_i \setminus X'_{n, \{n\}}$'' before forcing the left-hand side to be 
$\empty$\footnote{Note that there are no cycles in such a winning strategy for Spoiler because of positionality, whence the complication}.
\end{exa}

Nevertheless, we still conjecture that the general problem might be
solvable in $\Pspace$.

\begin{conj}
Given a Spoiler winning strategy in $\simGame$, we can optimize it into a positional winning strategy that must reach a cycle in any play 
in polynomial time; therefore solving ``is $e \le f$ valid in $\Wei$?'' for arbitrary $e$ and
$f$ is doable in $\Pspace$.
\end{conj}

\subsection{Lower bounds}

We now turn to some easy hardness bounds obtained by reducing the standard $\Pspace$-complete
problem $\mathsf{TQBF}$ (\emph{True Quantified Boolean Formula})
to our problems.

\begin{lem}
\label{lem:pspace-hard-pointed}
Given $e,f \in \RegE{\Sigma}$ that do not contain $(-)^\diamond$ or units, solving ``is $e \le f$ valid in $\ptWei$?'' is $\Pspace$-hard.
\end{lem}
\begin{proof}
Consider a quantified boolean formula
\[ \phi = \forall x_1 \exists x_2 \ldots \exists x_{2n} ~~ \bigvee_{i=1}^k \bigwedge_{j = 1}^{m_i} l_{i,j}\]
with the $l_{i,j}$ being literals in $\cL = \{x_1, \overline{x_1}, \ldots, x_{2n}, \overline{x_{2n}}\}$. Let us order literals by setting $x_1 < \overline{x_1} < \ldots < x_n$
and assume that for every $i \in I$, the $(l_{i,j})_{j = 1}^{m_i}$ sequence is antitone.
We now claim that $\phi$ holds if and only if the following is valid in the pointed
Weihrauch degrees:
\[
\bigsqcap\limits_{i = 1}^k \Compo\limits_{j=1}^{m_i} l_{i,j} \quad \le \quad
\left(x_{2n}^\diamond \sqcup \overline{x_{2n}}^\diamond\right)
 \star 
 \ldots \star
\left(x_2^\diamond \sqcup \overline{x_2}^\diamond\right)
 \star
 \left(x_1^\diamond \sqcap \overline{x_1}^\diamond\right)
\]
The idea is that a play in $\psimGame$ will determine a valuation according
to the quantifier alternation and that an optimal strategy for Duplicator
will be to remove all exposed occurrence of a literal $l$ is the right
hand-side of the position has shape $\ldots \star l^\diamond$.

We can suppress the $(-)^\diamond$ on the right hand-side term as follows:
replace any occurrence of $l^\diamond$ with a $N$-fold product
$(1 \sqcup l) \star \ldots \star (1 \sqcup l)$ for $N$ larger than the
number of occurrences of $l$ in the left hand-side.
\end{proof}

The lemma that follows could be derived as an immediate corollary of
\Cref{lem:pspace-hard-pointed}. But we give an alternative proof, which
exploits the fact that degrees are not necessarily pointed.

\begin{lem}
\label{lem:pspace-hard}
Given $e,f \in \RegE{\Sigma}$, solving ``is $e \le f$ valid in $\Wei$?'' is 
$\Pspace$-hard, even if $e$ and $f$ do not contain $(-)^\diamond$ or units.
\end{lem}
\begin{proof}
We again reduce $\mathsf{TQBF}$ to the problem in question; so consider an arbitrary
quantified boolean formula $\phi$ as in the proof of~\Cref{lem:pspace-hard-pointed}
(with the same conventions for variables, literals) and let us turn that
into terms of $u_\phi, v_\phi \in \RegE{\cL \cup \{\bullet\}}$ %
defined as follows, with $\bullet$ being a fresh variable:
\[
\begin{array}{llclclcl}
& u_\phi &=& e_{x_{2n}} \star \ldots \star e_{x_2} \star a_{x_1} \star \bullet &\qquad
 \text{and} \qquad& v_\phi &=& t_\phi \star e'_{x_{2n}} \star \ldots e'_{x_2} \star a_{x_1} \star \bullet\\\\
\text{with} &
e_x &=& (x \star \bullet) \sqcap (\overline{x} \star \bullet)
& &
a_x &=& x \sqcup \overline{x} \\
& e'_x &=& (x \sqcup \overline{x}) \star \bullet & &
t_\phi &=& \bigsqcup\limits_{i = 1}^k \Compo\limits_{j = 1}^{m_i} l_{i,j} \\
\end{array}
 \]
Clearly $u_\phi$ and $v_\phi$ can be produced in linear time from $\phi$.
Now we claim that $u_\phi \le v_\phi$ holds in the Weihrauch degrees if and only
if $\phi$ holds. The high-level idea is that any play in the simulation
game must first step through the matching pairs of terms $a_x$ or $e_x$ and $e'_x$
simultaneously, determining a valuation through the set of literals that are
simulated along the way. The terms are designed to mimic the behaviour of
quantifiers in $\simGame$. Duplicator cannot hope to win by attempting both
branches of an $e_l$ because they only have $n+1$ $\bullet$ available to them,
so without loss of generality, we can omit this behaviour.
Then, for any valuation $\rho$ seen as a subset
of $\cL$, we have that the reached position $\position{\rho \cup \{\bullet\}}{\{1, \zeta_2 \star \bullet, \ldots, \zeta_{2n} \star \bullet\}}{t_\phi}$ is
winning for Duplicator if and only if
$\rho$ satisfies the body of $\phi$.
\end{proof}

\begin{lem}
\label{lem:coNP-hard}
Given $e,f \in \RegE{\Sigma}$ that do not contain $\sqcap$, $(-)^\diamond$ or units, solving ``is $e \le f$ valid in $\Wei$?'' is $\coNP$-hard.
\end{lem}
\begin{proof}
  We can reduce from $\mathsf{TQBF}$ restricted to $\Pi_1$ formulas.
  Then the same reduction as in~\Cref{lem:pspace-hard} works,
  as we only used $\sqcap$ to simulate existential quantifiers.
\end{proof}

\begin{oprob}
Is the bound given in \Cref{lem:coNP-hard} tight? If yes, is it also the case when we allow
$(-)^\diamond$?
\end{oprob}

\section{Conclusion}
\label{sec:conc}

We have given a complete axiomatization $\SRKAM$ of universal validity of single
inequalities in the partial Weihrauch degrees and the (usual) pointed
Weihrauch degrees for terms built from the lattice and composition operators.
By way of our game characterization, we also showed that we can decide
whether a given inequality is valid (or, equivalently, if it is derivable in
$\SRKAM$) and gave some easy complexity bounds.
Further analysis of the simulation game $\simGame$ might be required to get tight bounds.

Future work could attempt to generalize elements of
\Cref{thm:mainloop} across several dimensions. We already touched on one dimension
in our discussion of related works, which is the wider applicability of $\SRKAM$
to other notions of degrees than partial Weihrauch degrees. There are also two other
of improvements one could target, irrespective of the degree structure to be studied.
\begin{itemize}
\item One could try to add more operators in the signature. One natural
option worth considering for an extension would be the parallel product $\times$\footnote{
Defined by $\dom(P \times Q) \simeq \dom(P) \times \dom(Q)$ and $(P \times Q)(\tuple{u,v}) \simeq P(u) \times Q(v)$.}, which might lead to consider (a skewed
version of) concurrent Kleene algebras~\cite{hoare2009concurrent} with distributive meets.
Having
a notion of simulation between branching automata~\cite{lodaya1998series},
might be helpful.
When it comes to the axiomatizing this, a roadblock is that
we are currently not aware of a complete axiomatization of the $\times, \sqcap$ fragment
of Weihrauch degrees~\cite{NPP24}, and believe this problem to be challenging.

Another compelling extension would be to target signatures allowing for general
fixpoint operators on well-chosen terms, as hinted in~\cite{PP26}. In particular,
enriching our present signature with least fixpoints of ``fibred'' terms would allow us to
capture all problems definable by automata (as per~\Cref{def:automaton}). Those
terms should admit a natural axiomatization, which we hope can be shown to be
complete using the same techniques as in~\Cref{sec:completeness}. We further expect that adding
greatest fixpoints would yield a correspondence with alternating parity automata over infinite
words. Here we would hope that one may improve and adapt the techniques both in~\Cref{sec:game} and
~\Cref{sec:completeness} to this setting. Finally, adding on top of that $\zeta$
fixpoints as introduced in~\cite{PP26} would allow us to capture the infinite iteration
operator introduced in~\cite{brattka2025loops}. This should presumably lead to
a correspondence with alternating automata working on countable ordinals.
Merely adapting the results of~\Cref{sec:game} to that setting already sounds like
a compelling project.
\item We gave an axiomatization that covers single (in)equations and showed
that validity was decidable.
In contrast, checking the validity of
general first-order formulas in the Weihrauch degrees is extremely
undecidable~\cite[Theorem 1.9]{lmpsv}. It might be interesting to look at intermediate cases,
such as the Horn theories induced by a selected set of operators. For instance, one
result is that the Horn theory of Kleene algebras is $\Pi^1_1$-complete~\cite{kozen1997complexity}. Is this also the
case for $\SRKAM$?
\end{itemize}

\section*{Acknowledgments}
\noindent I want to thank Eike Neumann, Arno Pauly and Manlio Valenti for
discussions about this work, and for inspiring it in the first place. I want to
also thank Damien Pous for discussions on related work on simulations. And finally,
many thanks to the anonymous reviewers whose comments improved substantially the
presentation of this paper.

\bibliographystyle{alphaurl}
\bibliography{biblio}
\end{document}